\documentclass{article}
\usepackage{fullpage}
\usepackage{graphicx} 

\usepackage{amsmath, amsthm, amssymb, amsfonts}
\usepackage{bbm}
\usepackage{bbold}
\usepackage{multirow}
\usepackage[basic]{complexity}
\usepackage[pdftex,bookmarks,colorlinks]{hyperref}
\usepackage{enumitem}
\usepackage{xcolor}
\usepackage[linesnumbered,ruled,vlined,boxed]{algorithm2e}
\usepackage{cleveref}

\newtheorem{theorem}{Theorem}
\newtheorem{lemma}[theorem]{Lemma}
\newtheorem{corollary}[theorem]{Corollary}

\newtheorem{definition}[theorem]{Definition}
\newtheorem{remark}[theorem]{Remark}

\hypersetup{
	colorlinks,
	linkcolor={blue!50!gray},
	citecolor={blue!50!gray},
	urlcolor={blue!50!gray}
}

\renewcommand{\R}{\mathbb{R}}

\renewcommand{\E}{\mathbb{E}}
\newcommand{\SMsym}{\mathbb{S}}
\newcommand{\SM}[1]{\SMsym^{#1 \times #1}}
\newcommand{\PSD}[1]{\SMsym_{+}^{#1 \times #1}}
\newcommand{\PD}[1]{\SMsym_{++}^{#1 \times #1}}

\newcommand{\boldVar}[1]{\mathbf{#1}}
\newcommand{\mvar}[1]{\boldVar{#1}}
\newcommand{\A}{\mvar{A}}\def\a{\mvar{a}}
\newcommand{\ma}{\A}
\newcommand{\B}{\mvar{B}}

\def\C{\mvar{C}}
\def\U{\mvar{U}}

\def\W{\mvar{W}}

\newcommand{\x}{\mvar{x}}
\newcommand{\z}{\mvar{z}}
\def\b{\mvar{b}}
\def\v{\mvar{v}}
\newcommand{\y}{\mvar{y}}
\newcommand{\mI}{\mvar{I}}\newcommand{\I}{\mvar{I}}
\newcommand{\mm}{\mvar{M}}\newcommand{\M}{\mvar{M}}
\newcommand{\mn}{\mvar{N}}\newcommand{\N}{\mvar{N}}
\def\S{\mvar{S}}

\newcommand{\mPi}{\mvar{\Pi}}
\newcommand{\defeq}{:=}
\newcommand{\otilde}{\tilde{O}}
\def\tr{\mathrm{tr}}

\newcommand{\lambAvg}{\overline{\lambda}}
\newcommand{\sigAvg}{\overline{\Sigma}}

\newcommand{\nnz}{\mathrm{nnz}}
\newcommand{\norm}[1]{\|#1\|}
\newcommand{\runtime}{\mathcal{T}}

\newcommand{\code}[1]{\mathtt{#1}}
\newcommand{\precon}{\code{PAGD}}
\newcommand{\rprecon}{\code{RPAGD}}
\newcommand{\genSolve}{\code{solve}}
\newcommand{\baseSolve}{\code{bsolve}}
\newcommand{\preSolve}{\code{psolve}}

\newcommand{\chain}{\mathcal{C}}
\newcommand{\primalchain}{\mathcal{P}}
\newcommand{\primaldualchain}{\mathcal{D}}

\newcommand{\outsim}{\mathrm{out}}
\newcommand{\xout}{\x_{\outsim}}
\newcommand{\xopt}{\x^{*}}

\newcommand{\mycomment}[1]{\tcp{{\color{gray} #1}}}

\newcommand{\argmin}{\mathrm{argmin}}

\renewcommand{\ln}{\log}

\definecolor{darkgreen}{cmyk}{1,0,1,0}

\newcommand{\sidford}[1]{\textcolor{darkgreen}{[Aaron: #1]}}

\title{Approaching Optimality for Solving\\ 
Dense Linear Systems with Low-Rank Structure}

\author{Micha{\l} Derezi\'nski\\
\texttt{derezin@umich.edu}\\
University of Michigan
\and
Aaron Sidford\\
\texttt{sidford@stanford.edu}\\
Stanford University}

\date{}

\begin{document}

\maketitle
\thispagestyle{empty}

\begin{abstract}
We provide new high-accuracy randomized algorithms for solving linear systems and regression problems that are well-conditioned except for $k$ large singular values. For solving such $d \times d$ positive definite system our algorithms succeed whp.\ and run in time $\otilde(d^2 + k^\omega)$. For solving such regression problems in a matrix $\ma \in \R^{n \times d}$ our methods succeed whp.\ and run in time $\otilde(\nnz(\ma) + d^2 + k^\omega)$ where $\omega$ is the matrix multiplication exponent and $\nnz(\A)$ is the number of non-zeros in $\A$. Our methods nearly-match a natural complexity limit under dense inputs for these problems and improve
upon a trade-off in prior approaches that obtain running times of either $\tilde O(d^{2.065}+k^\omega)$ or $\tilde O(d^2 + dk^{\omega-1})$ for $d\times d$ systems. Moreover, we show how to obtain these running times even under the weaker assumption that all but $k$ of the singular values have a suitably bounded generalized mean. Consequently, we give the first nearly-linear time algorithm for computing a multiplicative approximation to the nuclear norm of an arbitrary dense matrix. 
 Our algorithms are built on three general recursive preconditioning frameworks, where matrix sketching and low-rank update formulas are carefully tailored to the problems' structure.

\tableofcontents
\end{abstract}

\newpage

\setcounter{page}{1}
\addtocontents{toc}{\protect\setcounter{tocdepth}{1}}

\section{Introduction}
\label{s:intro}

Linear system solving and (least squares) regression are some of the most fundamental problems in numerical linear algebra, algorithm design, and optimization theory. These problems appear in numerous applications and are key subroutines for solving many algorithmic problems, see e.g., \cite{woodruff2014sketching,frostig2016principal,garber2016faster,musco2018spectrum}. In this paper we consider the following standard settings for these problems:
\begin{itemize}
    \item \emph{(Positive Definite) Linear System}: Given real-valued symmetric positive definite (PD) $d\times d$ matrix $\mm$, denoted $\mm \in \PD{d}$ (see \Cref{s:prelim} for notation), vector $\b \in \R^d$, and target relative accuracy $\epsilon \in (0,1)$, output an \emph{$\epsilon$-(approximate )solution to $\mm\x = \b$}, i.e., 
    \[
    \xout \in \R^d\quad
    \text{such that}\quad
    \norm{\xout - \xopt}_\mm \leq \epsilon \norm{\xopt}_\mm\,,
    \quad\text{where}\quad
    \mm \xopt = \b
    \,.
    \]
    \item \emph{(Least Squares) Regression}: Given a matrix $\ma \in \R^{n \times d}$ with full column rank, vector $\b \in \R^d$, and target relative accuracy $\epsilon \in (0,1)$, output an \emph{$\epsilon$-(approximate )solution to $\min_{\x \in \R^{d}} \norm{\ma\x - \b}^2$}, i.e.,
    \[
    \xout \in \R^d\quad
    \text{such that}\quad
    \norm{\A\xout - \A\xopt} \leq \epsilon \norm{\A\xopt}\,,\quad
    \text{where}\quad
    \xopt = \argmin_{\x \in \R^{d}} \norm{\ma\x - \b}^2.
    \]
\end{itemize}
Regression encompasses linear system solving (e.g., by setting $\ma = \mm$) up to slight differences in measuring error. We distinguish the linear system case as it arises commonly and often admits stronger guarantees.

There has been extensive research on solving these problems. Advances include continued improvements in the matrix multiplication exponent $\omega< 2.372$ \cite{alman2025more}, running times that depend on the sparsity of the input matrix or the distribution of its eigenvalues and singular values, and more. (See \Cref{s:intro:related_work} for a longer discussion of related work.) This paper is particularly motivated by a recent line of work on solving these problems when they are what we call \emph{$k$-well-conditioned} (Problem 1.1 in \cite{derezinski2024faster}). In this case, all of the singular values of the input matrix (either $\M$ or $\A$) are the same up to a constant factor, except for the top $k$, which can be distributed arbitrarily. 
Motivated by applications in machine learning, optimization and statistics, the $k$-well-conditioned setting has been studied (under different guises) for decades, dating back to classical results on the convergence of Krylov subspace methods \cite{axelsson1986rate}.

Many algorithmic approaches to the $k$-well-conditioned problem have focused on utilizing fast matrix multiplication (e.g., see \cite{musco2018spectrum}). This approach can be used to solve a $k$-well-conditioned $d\times d$ linear system in $\tilde O(d^2k^{\omega-2})$ time. (Throughout the paper $\otilde(\cdot)$ hides polylogarithmic factors in $n$ and $d$; in just \Cref{s:intro,s:overview}, to simplify runtime statements, we assume that conditions numbers of $\ma$ and $\mm$ are polynomially bounded in $n$ and $d$ and that accuracies, $\epsilon$, are inverse polynomially bounded in $n$ and $d$ for a known polynomial.) However, recently, \cite{derezinski2023solving} gave a randomized iterative algorithm that can solve such systems in $\tilde O(d^2+dk^{\omega-1})$ time via stochastic optimization. A different approach was proposed in \cite{derezinski2024faster}, which gave an $\tilde O(d^{2.065} + k^\omega)$ time algorithm that uses a randomized preconditioner based on matrix sketching. In each case, the methods can be naturally extended from linear systems to regression at an additional $\tilde O(\nnz(\A))$ cost, where $\nnz(\A)\leq nd$ denotes the number of non-zero entries in $\A$. (See Table \ref{tab:comparison} for an overview.)

A key open problem presented in \cite{derezinski2024faster} is whether it is possible to further improve upon these running times and obtain a $\tilde O(d^2+k^\omega)$ time solver for a $k$-well-conditioned $d\times d$ linear system. This question is motivated by their lower bound (Theorem 7.1 in \cite{derezinski2024faster}), which states that the time complexity of solving a $k$-well-conditioned $d\times d$ linear system to within $1/\poly(n)$ accuracy is at least $\Omega(d^2+k^\omega)$, as long as the complexity of solving an arbitrary $k\times k$ linear system is at least $\Omega(k^\omega)$. This open problem can be naturally extended to regression, where we would hope to obtain a $\tilde O(nd+k^\omega)$ running time for dense inputs.

In this paper, we provide randomized algorithms which resolve these open problems, obtaining the desired running times. In fact, we show how to efficiently solve an even broader class of linear systems and regression problems. Rather than requiring the system to be $k$-well-conditioned, we obtain nearly-linear time algorithms when the matrix is suitably $k$-well-conditioned \emph{on average}. Formally, we introduce a notion of the $(k,p)$-averaged condition number of a matrix which corresponds to the generalized mean of the ratio of the non-top-$k$ singular values to the smallest singular value. We obtain linear system solving and regression algorithms with runtimes that depend on this average conditioning and run in nearly linear time when the matrix is $(k,p)$-well-conditioned, i.e., the $(k,p)$-averaged conditioned number is $O(1)$, for suitable $p$. 

The starting point in obtaining our results is a broad optimization technique known as \emph{recursive preconditioning}. This technique involves carefully applying in a recursive fashion a preconditioned iterative method (we use preconditioned accelerated gradient descent (AGD) \cite{js_talg}, though alternative similar techniques like preconditioned Chebyshev \cite{golub1961chebyshev} may also work) to reduce linear system solving and regression to solving simpler versions of related problems. Variants of recursive preconditioning have been used extensively in prior work, particularly on solving graph-structured linear systems \cite{SpielmanT14,KoutisMP10,KoutisMP11,CohenKMPPRX14,CohenKPPRSV17,js_talg}. 
We start with a relatively straightforward variant of this approach as a warm-up (\Cref{s:ls}), and then develop two more sophisticated recursive frameworks which allow us to further exploit average conditioning in regression (\Cref{s:primal-dual}) and positive definiteness in the linear system task (\Cref{s:psd}). We combine these recursive frameworks with matrix sketching and oblivious subspace embeddings \cite{sarlos-sketching,cw-sparse,nn-sparse,cohen2016optimal}, reduction techniques based on the Woodbury matrix identity, and a careful running time analysis. (See \Cref{s:overview} for further explanation.)

To illustrate the utility of these results we show how our algorithms can be used as subroutines to obtain improved running times for spectrum approximation problems involving estimating properties of the singular values of a matrix. 
Perhaps most notably, we show how to apply our results to obtain the first nearly-linear time algorithm for obtaining a multiplicative approximation to the nuclear norm of a dense matrix.

\paragraph{Organization.} In the remainder of this introduction we present our main results (\Cref{s:intro:main_results}) and discuss related work (\Cref{s:intro:related_work}). An overview of our approach is then presented in \Cref{s:overview} and preliminaries are presented in \Cref{s:prelim}. Regression algorithms are covered in \Cref{s:ls,s:primal-dual}, algorithms for solving PSD systems are covered in \Cref{s:psd}, and spectrum approximation algorithms are covered in \Cref{s:spectrum}. 

\subsection{Main Results}
\label{s:intro:main_results}

Here we present the main results of the paper. Our first result, \Cref{t:ls}, implies that $k$-well conditioned regression problems in $\ma \in \R^{n \times d}$ can indeed be solved in $\otilde(nd + k^\omega)$ time.
\Cref{t:ls} attains this running time even under the relaxed assumption that $\ma$ is only $k$-well conditioned on average in the sense that the values of $\sigma_i(\ma)^2/\sigma_d(\ma)^2$ for $i > k$ are $O(1)$ on average. We prove this theorem in \Cref{s:ls} using a straightforward strategy of recursively preconditioning using matrices $\ma_t \in \R^{n_t \times d}$ with progressively fewer rows ($n_t$ roughly decreases as $t$ increases). The resulting algorithm is one of the simplest in this paper and we present it as a warm-up to more advanced results.

\begin{theorem}[Regression Warm-up]\label{t:ls}
    There is an algorithm that given $\A\in\R^{n\times d}$, $\b\in\R^d$, $\epsilon\in(0,1)$, and $k \in [d]$ such that $\bar\kappa_{k,2}(\A)^2 :=\frac1{d-k}\sum_{i>k}\frac{\sigma_i(\ma)^2}{\sigma_d(\ma)^2} = O(1)$, outputs $\xout\in\R^d$ such that $\|\A\xout-\A\xopt\|\leq \epsilon\|\A\xopt\|$ whp., where $\x^* = \argmin_{\x\in\R^d}\|\A\x-\b\|^2$, in $\tilde O((\nnz(\A) + d^2)\log(1/\epsilon)) + O(k^\omega)$ time.
\end{theorem}
The relaxed notion of well-conditioned in \Cref{t:ls}, defined via $\bar\kappa_{k,2}(\A)$, for general matrices (along with an analogous $\bar\kappa_{k,1}(\M)=\frac1{d-k}\sum_{i>k}\frac{\lambda_i(\M)}{\lambda_d(\M)}$ for PD matrices $\M$) has been relied upon extensively in literature on linear systems, regression and quadratic optimization, e.g., \cite{strohmer2009randomized,leventhal2010randomized,lee2013efficient,musco2018spectrum}, including the recent prior works on the $k$-well-conditioned problem \cite{derezinski2023solving,derezinski2024faster,derezinski2024fine,derezinski2025randomized}. Thus, it might appear as though \Cref{t:ls} largely resolves the major questions regarding the complexity of $k$-well-conditioned regression. However, perhaps surprisingly, we show that this is not the case. Our algorithmic results reveal that these existing notions of averaged condition numbers 
are still \textit{too restrictive} to capture the true complexity of these problems. Formally, we  propose the following definition of a generalized averaged condition number.
\begin{definition}[$(k,p)$-averaged conditioned number]\label{d:kappabar}
We define the \emph{$(k,p)$-averaged condition number of $\A$} for $\A \in \R^{n\times d}$, $k\in[d]$, and $p \in (0, \infty]$, as
\begin{align*}
    \bar\kappa_{k,p}(\A) := \bigg(\frac1{d-k}\sum_{i>k}\frac{\sigma_i(\ma)^p}{\sigma_d(\ma)^p}\bigg)^{1/p}
    \text{ if } p \neq \infty
    \text{ and }
    \bar\kappa_{k,\infty}(\A) := \lim_{p \rightarrow \infty}  
    \bar\kappa_{k,p}(\A) 
    = \frac{\sigma_{k+1}(\ma)}{\sigma_d(\ma)}
    \,.
\end{align*}
Additionally, we say that $\A$ is \emph{$(k,p)$-well-conditioned} if $\bar\kappa_{k,p}(\A) = O(1)$.
\end{definition} 
The $(k,p)$-averaged condition number is simply the generalized mean with exponent $p$ computed over the ratios $\sigma_i/\sigma_d$ for $i>k$. Thus, $\bar\kappa_{k,p_1}(\A)\leq \bar\kappa_{k,p_2}(\A)$ for any $p_1<p_2$. Also, $\ma$ is $(k,\infty)$-well-conditioned if and only if it is $k$-well-conditioned. This monotonicity of $\bar\kappa_{k,p}(\A)$ in $p$ highlights how \Cref{t:ls} corresponds to a more general result than solving dense $k$-well-conditioned regression problems in nearly linear time. 

Remarkably, we show how to carefully apply and extend our framework to obtain similar regression results for $(k,p)$-well-conditioned $\ma$ with any constant $p > 1/2$. \Cref{t:primal-dual-optimized}, proved in \Cref{s:primal-dual} and simplified in \Cref{t:ls-optimized} below, obtains the same (up to polylogarithmic factors) running times as in \Cref{t:ls} but only requires that $\sum_{i > k} \frac{1}{d - k} \sqrt{\sigma_i(\ma) / \sigma_d(\ma)} = O(1)$, rather than $\sum_{i > k} \frac{1}{d - k} \sigma_i(\ma)^2 / \sigma_d(\ma)^2 = O(1)$. To see why this is significant, note that when each $\sigma_i(\ma) / \sigma_d(\ma) = O(d^2)$, \Cref{t:primal-dual-optimized} yields a nearly linear time algorithm for dense inputs whereas $\sum_{i > k} \frac{1}{d - k} \sigma_i(\ma)^2 / \sigma_d(\ma)^2$ could be as large as $d^3$ for any $k\leq d/2$.

\begin{theorem}[Regression]\label{t:ls-optimized}
    There is an algorithm that given $\A\in\R^{n\times d}$, $\b\in\R^d$,  $\epsilon\in(0,1)$, $k \in [d]$, and constant $p > 1/2$ where $\ma$ is $(k,p)$-well-conditioned outputs $\xout\in\R^d$ such that $\|\A\xout-\A\xopt\|\leq \epsilon\|\A\xopt\|$ whp., where $\x^* = \argmin_{\x\in\R^d}\|\A\x-\b\|^2$, in time $\tilde O((\nnz(\A) + d^{2})\log(1/\epsilon)) + O(k^\omega)$.
\end{theorem}

We also show how to apply our framework to solving PD linear systems. Our resulting algorithm, developed in \Cref{s:psd}, provides a similar recursive scheme as our regression algorithm but the recursive steps are tailored to the structure of PD linear systems. While solving PD linear systems can essentially be viewed as a special case of the regression task, the PD structure of the input matrix allows us to further relax the conditioning assumptions. As is common for comparing regression and PD linear system solving algorithms (e.g., \cite{derezinski2023solving}), we essentially obtain algorithms where previous dependencies on the squared singular values of $\ma$ are replaced with the eigenvalues of $\ma$. Consequently, as the following \Cref{t:psd-optimized} asserts, we obtain nearly linear time algorithms for solving dense $d \times d$ PD systems in the desired $\otilde(d^2 + k^\omega)$ time provided that the input matrix is $(k,p)$-well-conditioned for any constant $p > 1/4$. This result improves upon the $p > 1/2$ requirement from our regression result, applies to a broader class of matrices than those that are $k$-well-conditioned, and matches the conditional lower bound of \cite{derezinski2024faster}.

\begin{theorem}[PD linear system]\label{t:psd-optimized}
    There is an algorithm that given $\M\in\PD{d}$, $\b\in\R^d$, $\epsilon\in(0,1)$, $k \in [d]$, and constant $p > 1/4$ such that $\mm$ is $(k,p)$-well-conditioned outputs $\xout\in\R^d$ such that $\|\xout-\xopt\|_{\M}\leq \epsilon\|\xopt\|_{\M}$ whp., where $\x^* = \M^{-1}\b$, in time $\otilde(d^{2}\log(1/\epsilon)) + O(k^\omega)$. \end{theorem}
In Sections \ref{s:primal-dual} and \ref{s:psd}, we give more general results which include explicit dependence on the averaged condition numbers, obtaining $\tilde O(\nnz(\A)+d^2\bar\kappa_{k,p}(\A)+k^\omega)$ runtime for solving regression with any $\A\in\R^{n\times d}$ and $p>1/2$ (\Cref{t:primal-dual-optimized}), as well as $\tilde O(d^2\sqrt{\bar\kappa_{k,p}(\M)}+k^\omega)$ runtime for solving a linear system with any $\M\in\PD{d}$ and $p>1/4$ (\Cref{t:dual-main-optimized}). In the extreme cases of $p=1/2$ and $p=1/4$, respectively, our results imply only slightly weaker running times (matching the above up to sub-polynomial factors, see Remark \ref{r:almost-linear}).

\begin{table}\centering\begin{tabular}{c||c|c|c|c}
   & \multirow{2}{*}{Runtime guarantee} & \multicolumn{2}{c|}{\underline{$(k,p)$-Well-conditioned matrices}} & \multirow{2}{*}{Ref.}\\
   &&General & Positive definite&\\
    \hline\hline
\textit{Conditional lower bound} & $\Omega(d^2 + k^\omega)$ & \multicolumn{2}{c|}{$p\in(0,\infty]$}
& \cite{derezinski2024faster}\\
\hline
Krylov Subspace Methods & $\tilde O(d^2k)$  
&\multicolumn{2}{c|}{$p=\infty$}& \cite{axelsson1986rate}\\
\hline
Fast Matrix Multiplication & $\tilde O(d^2k^{\omega-2})$  
& \multicolumn{2}{c|}{$p\in[2,\infty]$}
& \cite{musco2018spectrum}\\
\hline
Randomized Block Kaczmarz &$\tilde O(d^2 + dk^{\omega-1})$ & $p\in[2,\infty]$&
$p\in[1,\infty]$& \cite{derezinski2023solving}\\
\hline
Sketched Preconditioning  & $\tilde O(d^{2.065} + k^\omega)$ & $p\in[2,\infty]$&
$p\in[1,\infty]$& \cite{derezinski2024faster}\\
\hline
\hline
\textbf{Our results} &  $\tilde O(d^2 + k^\omega)$ & $p\in(1/2,\infty]$&
$p\in(1/4,\infty]$&\textbf{Thm.~\ref{t:ls-optimized}, \ref{t:psd-optimized}} 
  \end{tabular}
  \caption{Running time comparison for solving a $(k,p)$-well-conditioned, $d\times d$ linear system (Definition \ref{d:kappabar}). Lower values of $p$ and larger values of $k$ correspond to broader classes of matrices.
  }
  \label{tab:comparison}
\end{table}

\paragraph{Comparison to Prior Work.}
\Cref{t:ls-optimized,t:psd-optimized} can be directly compared with recent prior work on solving $(k,p)$-well-conditioned problems \cite{musco2018spectrum,derezinski2023solving,derezinski2024faster}, as they provide guarantees in terms of averaged condition numbers for both regression and linear systems. We improve upon these results in two ways (see \Cref{tab:comparison}). 

First, we avoid the complexity trade-off that previously existed in the literature with respect to the dependence on $k$ and the matrix dimensions. To illustrate this, let us focus on the case of $d\times d$ linear systems. Comparing the state-of-the-art results, $\tilde O(d^2+dk^{\omega-1})$ by \cite{derezinski2023solving} and $\tilde O(d^{2.065}+k^\omega)$ by \cite{derezinski2024faster}, we observe a trade-off between the two complexities, leading to each method being preferable in a different regime of $k$. Our algorithms overcome this trade-off, achieving the best of both worlds, including the first near-optimal complexity bounds for solving $k$-well-conditioned $d\times d$ linear systems when $k$ is between $d^{0.73}$ and $d^{0.87}$, closing the last remaining complexity gap in this problem. 

Second, we obtain these running times using more relaxed notions of averaged condition numbers $\bar\kappa_{k,p}$ compared to the previous works. In the case of regression, we require $p>1/2$ compared to the previous best of $p\geq 2$, and in the case of PD linear systems, we require $p>1/4$ compared to the previous best of $p\geq 1$.

Interestingly, \Cref{t:ls-optimized,t:psd-optimized}  imply new complexity guarantees even for $k=0$, where the associated algorithms do no need to  rely on fast matrix multiplication.
Thus, it is also helpful to compare with state-of-the-art iterative methods which do not use fast matrix multiplication, and often rely on somewhat different notions of averaged condition numbers. Namely, sampling and accelerated stochastic variance reduced gradient methods can solve regression in $\otilde(\sum_{i} d \norm{\a_i} \sqrt{\tau_i(\ma)}/\sigma_d(\ma))$ time where $\a_i$ is the $i$th row of $\A$ and $\tau_i(\ma)$ is the leverage score of that row \cite{agarwal2020leverage}, and accelerated coordinate descent methods \cite{ZhuQRY16,NesterovS17} can solve PD systems in $\otilde(\sum_{i} d \sqrt{\M_{ii} / \lambda_d(\mm)})$ time. In certain corner cases, e.g., when $\A$ has orthogonal rows or $\M$ is diagonal, these bounds reduce to $\tilde O(d^2\bar\kappa_{0,1})$ and $\tilde O(d^2\sqrt{\bar\kappa_{0,1/2}})$ running times, respectively; however attaining these latter running times in general was an open problem. Remarkably, we not only resolve this problem, but also obtain an even better dependence on the conditioning parameter $p$.

\paragraph{Application to Spectrum Approximation.}
To illustrate the significance of our new algorithmic results, we describe how they can be used to obtain improved runtimes for a range of tasks involving estimation of the spectrum of a PD matrix or the singular values of an arbitrary matrix, where linear system and regression solvers arise as subroutines. In particular, via a direct reduction from Schatten $p$-norm approximation to solving $(0,p)$-well-conditioned linear systems (\Cref{c:schatten-reduction}), we obtain the first nearly linear time approximation algorithm for the nuclear norm $\|\A\|_1=\sum_i\sigma_i(\A)$ of a dense matrix $\A$. The first methods for this problem which avoided an $\Omega(d^\omega)$ term in their running time were provided by \cite{musco2018spectrum}, attaining a $\tilde O((\nnz(\A)+d^{2.18})\poly(1/\epsilon))$ time $(1\pm\epsilon)$-approximation algorithm. 

\begin{theorem}[Nuclear norm approximation]\label{t:nuclear}
Given a matrix $\A\in\R^{n\times d}$ and $\epsilon>0$, we can compute $X\in\R$ that whp.\ satisfies $X\in(1\pm\epsilon)\|\A\|_1$ in time $\tilde O((\nnz(\A)+d^{2})\poly(1/\epsilon))$.
\end{theorem}

This result is based on combining the algorithmic framework for spectrum estimation of \cite{musco2018spectrum} with our regression solvers, improving on the previous best runtime of $\tilde O\big((\nnz(\A)+d^{2.11})\poly(1/\epsilon)\big)$ by \cite{derezinski2024faster}. Notably, unlike the previous running times, our nuclear norm approximation algorithm does \emph{not} rely on fast matrix multiplication to achieve its claimed running time. In Section \ref{s:spectrum}, we illustrate how our solvers can be used to provide faster runtimes for more general spectrum approximation tasks, in particular estimating any Schatten $p$-norm $\|\A\|_p=(\sum_i\sigma_i^p)^{1/p}$. Here, we obtain the first nearly linear time approximation algorithms for all $p\in(1/2,2)$, and provide the new best known time complexities for all $p\in(0,2)$.

\subsection{Related Work}
\label{s:intro:related_work}

Extensive literature has been dedicated to algorithms
for solving linear systems and regression tasks. Many of these works have approached the topic from a numerical linear algebra perspective \cite{golub2013matrix}, including through direct solvers such as Gaussian elimination \cite{higham2011gaussian} as well as iterative solvers such as Richardson and Chebyshev iteration \cite{golub1961chebyshev}, and finally Krylov subspace methods such as conjugate gradient \cite{hestenes1952methods}.
From the computer science perspective, early algorithmic improvements have come primarily from fast matrix multiplication, beginning with \cite{Strassen1969}, which continues to be an active area of research, e.g., see \cite{coppersmith1987matrix,williams2012multiplying,williams2024multiplication}. This has led to algorithms that solve a $d\times d$ linear system in $O(d^\omega)$ time. 
More recently, many works have focused on using randomization in order to speed up these algorithms for various classes of linear system and regression tasks. This is where we will focus our discussion of related work, as it is the most relevant to our approach.

One of the most important randomized tools for computational linear algebra is matrix sketching, along with its underlying theoretical framework of subspace embeddings \cite{sarlos-sketching}.
In the context of solving linear systems, sketching-based techniques have focused primarily on the highly over-determined regression setting ($n\gg d$; see, e.g., \cite{woodruff2014sketching} for an overview). This line of works, including \cite{rokhlin2008fast,cw-sparse,nn-sparse}, has produced algorithms for solving the $n\times d$ regression problem in time $\tilde O(\nnz(\A) + d^\omega)$, by using a subspace embedding to reduce the large dimension $n$, and then using fast matrix multiplication to construct a preconditioner for an iterative solver. Here, once again, $d^\omega$ arises as the key bottleneck.

The $k$-well-conditioned problem setting was first explicitly introduced by \cite{derezinski2023solving}, although it has appeared in a number of prior works either as a motivation or an algorithmic subroutine \cite{axelsson1986rate,sw09,gonen2016solving,musco2018spectrum}. These prior works have relied either on a careful convergence analysis of Krylov subspace methods \cite{axelsson1986rate,sw09}, or on using approximate low-rank approximation via the block power method, accelerated with fast matrix multiplication, to construct preconditioners that capture the top-$k$ part of the input's spectrum \cite{gonen2016solving,musco2018spectrum}. This approach attains $\tilde O(d^2k^{\omega-2})$ running time for $d\times d$ systems. This can be further optimized via fast rectangular matrix multiplication \cite{le2012faster}, but even with those improvements, the runtime is larger than $d^{2+\theta}$ for some $\theta>0$ given any $k=\Omega(d^{0.33})$. \cite{derezinski2023solving} took a different approach to solving $k$-well-conditioned systems, by showing that optimized variants of randomized block Kaczmarz and block coordinate descent exhibit a preconditioned convergence, without requiring to build a preconditioner. They obtained $\tilde O(d^2+dk^{\omega-1})$ running time, which was later fine-tuned (mainly in terms of the dependence on the conditioning) by \cite{derezinski2024fine,derezinski2025randomized}.

Most recently, \cite{derezinski2024faster} developed a new sketching-based preconditioning approach for the $k$-well-conditioned problem setting. Their algorithms skip the block power method used by \cite{gonen2016solving,musco2018spectrum}. This avoids the expensive dense matrix multiplication but also reduces the quality of the resulting preconditioner. They also utilize a second level of preconditioning to approximately apply the Woodbury formula for low-rank matrix inversion. This results in a $\tilde O(d^{2.065}+k^\omega)$ running time, where the sub-optimal dependence on $d$ is a consequence of the weaker preconditioner.  Our preconditioning chain technique can be viewed as a way of addressing this key limitation of their approach by chaining together many weak preconditioners.

There is also a line of work on randomized iterative methods for solving linear systems. For instance, the well-known stochastic variance reduced gradient method \cite{johnson2013accelerating} can solve regression in $\otilde(\sum_{i} d \cdot \|\a_i\|^2 / \sigma_d^2(\A) )$ time
and coordinate descent can solve PD linear systems in $\otilde(\sum_{i} d \cdot \mm_{ii} / \lambda_d(\mm))$ \cite{Nesterov12}. After multiple improvements, e.g., \cite{lee2013efficient,frostig2015regularizing,lin2015universal,ZhuQRY16,NesterovS17,allen2018katyusha,agarwal2020leverage}, the state-of-the-art runtime of these methods (in certain regimes) is $\otilde(\sum_{i} d \norm{\a_i} \sqrt{\tau(\ma)_i}/\sigma_d(\ma))$ for regression \cite{agarwal2020leverage} and $\otilde(\sum_{i} d \sqrt{\M_{ii} / \lambda_d(\mm)})$ for PD systems \cite{ZhuQRY16,NesterovS17}.

Finally, we note that, broadly, the idea of computing multiple matrix approximations for solving linear systems
has arisen in a number of applications. For example, iterative row sampling, sketching, or both to build preconditioners have been used for regression, \cite{iterative-row-sampling, CohenLMMPS15}, sparsification \cite{KapralovLMMS17}, and solving structured linear systems \cite{kyng2016sparsified, kyng2016approximate,cohen2018solving,PengS22,Jambulapati0MSS23}. The more particular idea of recursive preconditioning, i.e., solving a matrix approximation, by solving a matrix approximation, etc., has been used in a variety of prior work for solving graph structured linear systems, e.g., \cite{SpielmanT14,KoutisMP10,KoutisMP11,CohenKMPPRX14,CohenKPPRSV17,js_talg}. At a high level, we follow a similar template, however our algorithms differ in how the approximations are computed and chained together, as well as in how we reason about the running time to obtain fine-grained dependencies on the spectrum of~the~matrices.

\section{Overview of Approach}
\label{s:overview}

Here we provide an overview of our approach. We first introduce the main recursive preconditioning technique via a simplified variant of our regression solver (Section~\ref{s:ls}). Then we describe how this solver can be refined to attain our main regression result (Section~\ref{s:primal-dual}), discuss how to further optimize our techniques for PD systems (Section~\ref{s:psd}), and we overview the applications to spectrum approximation (Section~\ref{s:spectrum}).

\subsection{Warm-up: Recursive Regression Solver}
\label{s:overview-warm-up}
Here we discuss our approach to proving \Cref{t:ls}, and illustrate how it can lead to an improved dependence on averaged condition numbers. We consider the more general problem of solving linear systems in $\ma^\top \ma$ for full column rank $\ma \in \R^{n \times d}$, since regression, $\min_\x\|\A\x-\b\|^2$ , is equivalent to solving $\A^\top\A\x=\A^\top\b$.

Our approach is rooted in two prominent sets of tools for solving the problem: (1) 
Matrix sketching and low-rank approximation techniques from randomized linear algebra, e.g., \cite{woodruff2014sketching,cohen2016optimal,cohen2016nearly}; and
(2) preconditioning and acceleration techniques for iterative linear system solvers, e.g., \cite{js_talg}. 

\paragraph{Tool \#1: Low-Rank Approximation.} It is well known that given a matrix $\ma\in\R^{n\times d}$ there are distributions of matrices $\S \in \R^{s \times n}$ with $s\ll n$ such that $\ma^\top \S^\top \S \ma$ whp.\ approximates $\ma^\top \ma$ and $\S\A$ can be computed efficiently. For example, without any assumptions on the singular values~of~$\ma$, \cite{cw-sparse,nn-sparse,cohen2016nearly,chenakkod2024optimal} provide ``sparse oblivious subspace embeddings'' $\S \in \R^{\tilde O(d)\times n}$ such that $\S \ma$ can be computed in $\otilde(\nnz(\ma))$ time and whp.\ $\ma^\top \S^\top \S \ma \approx_2 \ma^\top \ma$, where $\M\approx_c\N$ denotes that $c^{-1}\N\preceq\M\preceq c\N$ in the Loewner ordering.

Additionally, when better bounds are given on the tail of the singular values of $\A$, it is possible to reduce the dimension of $\S$ further and obtain better approximations. For example, for all $k\leq d$ there are $\S \in \R^{\tilde O(k) \times d}$ for which it is possible to ensure that, whp., $\ma^\top \S^\top \S \ma + \nu \I \approx_2 \ma^\top \ma + \nu \I$ for $\nu = \frac1k\sum_{i>k}\sigma_i^2(\A)$ and maintain that $\S \A$ is computable in $\otilde(\nnz(\ma))$ time. Consequently this $\S \ma$ effectively provides a low-rank approximation of matrix $\A$ with respect to the singular values above $\nu$ threshold, where $\nu$ depends on $k$ and the heaviness of the tail of the singular values, i.e., $\frac1k\sum_{i>k}\sigma_i^2$ (we use $\sigma_i:=\sigma_i(\A)$ as a shorthand).

\paragraph{Tool \#2: Preconditioning.} It is also known that to efficiently solve a linear system $\mm \x = \b$ for $\mm\in\PD{d}$, 
it suffices to be able to repeatedly solve a linear system in $\tilde{\M}$, where $\tilde{\M}$ is an approximation of $\M$, and to be able to apply $\mm$ to a vector. 
For example, this can be achieved via the preconditioned Richardson iteration, $\x_{t+1} = \x_t - \tilde\M^{-1}(\M\x_t-\b)$. As long as $\M\preceq \tilde\M\preceq \kappa\M$, preconditioned Richardson started at $\x_0=\mathbf{0}$ after $t=O(\kappa\log1/\epsilon)$ steps finds $\xout=\x_t$ such that $\|\xout - \mm^{-1}\x\|_{\mm}^2\leq\epsilon\|\b\|_{\mm^{-1}}^2$.

However, when only a crude approximation of $\mm$ is known (i.e., $\kappa$ is larger than a constant), then accelerated iterative methods, such as the preconditioned Chebyshev iteration \cite{golub1961chebyshev}, are preferable to Richardson iteration since they require only $O(\sqrt\kappa\log1/\epsilon)$ steps to converge to an $\epsilon$-solution. In this paper we use a closely related method of preconditioned AGD from \cite{js_talg}, see  \Cref{lem:precon}.

\paragraph{Na\"ive Approach.} These powerful tools lead to fast running times for solving linear systems in many different settings. For example, when $\A$ is tall ($n\gg d$), then simply combining subspace embeddings with preconditioned Richardson or Chebyshev reduces solving linear systems in $\A^\top\A$ to solving $\tilde O(1)$ linear systems in $\A^\top\S^\top\S\A$, where $\S\A$ is $\tilde O(d)\times d$,
yielding an $\otilde(\nnz(\ma) + d^\omega)$ time solver for $\A^\top\A$ \cite{rokhlin2008fast}.

By leveraging the low-rank structure we can further reduce the cost of constructing the preconditioner: If we pick an embedding $\S \in \R^{\tilde O(k) \times d}$ and precondition the system with $\tilde\M=\A^\top\S^\top\S\A+\nu\I$ where $\nu=\frac1k\sum_{i>k}\sigma_i^2$, then it is possible to obtain a crude approximation $\A^\top\A\preceq \tilde\M\preceq\kappa\A^\top\A$ with:
\begin{align}
\kappa = O\Big(\frac{\nu}{\sigma_{d}^2}\Big) = O\bigg(\frac1k\sum_{i>k}\frac{\sigma_i^2}{\sigma_d^2}\bigg) = O\Big(\frac dk\bar\kappa_{k,2}(\A)^2\Big).\label{eq:overview-warm-up-naive}
\end{align}
When the matrix $\A$ is $k$-well-conditioned, then $\bar\kappa_{k,2}(\A)=O(1)$, and therefore to solve $\ma^\top \ma$ with preconditioned Chebyshev iterations it suffices to solve $\tilde O(\sqrt{d/k})$ linear systems with $\tilde\M$ (along with as many applications of $\A$). A version of this approach was used by \cite{derezinski2024faster}, which (after optimizing over the size of $\S$) obtains an $\tilde O(\nnz(\A)+d^{2.065}+k^\omega)$ time algorithm for regression. Unfortunately, due to the $\sqrt{d/k}$ factor, this strategy does not yield a nearly-linear running time, which leads to sub-optimal dependence on $d$.

\paragraph{Our Approach.} To improve upon this na\"ive approach we leverage recursive preconditioning, an idea that has been used for solving graph-structured linear systems, e.g., Laplacians \cite{SpielmanT14,KoutisMP10,KoutisMP11,CohenKMPPRX14,js_talg} and directed Laplacians \cite{CohenKPPRSV17}. Rather than directly approximating $\ma$ by a rank $k$ matrix, we build a sequence of matrices that successively approximate each other. We then use a solver for the second as a preconditioner for the first, which we develop by using a solver for the third, etc.

To adapt this idea for regression, we construct what we call a \emph{regularized preconditioning chain} (adapting terminology from the prior work on solving graph structured systems), which is a sequence $((\A_t,\nu_t)_{t\in\{0,1,...,T\}})$, where the first matrix $\A_0$ is our given input matrix $\A$, and subsequent matrices $\A_t=\S_t\A\in\R^{\tilde O(k_t)\times d}$ are progressively smaller (and coarser) approximations of $\A$ such that 
\begin{align*}
\A_t^\top\A_t+\nu_t\I\approx_4\A_{t-1}^\top\A_{t-1}+\nu_t\I,
\end{align*}
 where the 4 is chosen somewhat arbitrarily (see \Cref{footnote:4}). We show that this can be done so that decreasing $k_t$ by a constant factor 
 (say $\beta$) only increases $\nu_t$ by roughly the same factor. Therefore, while each individual matrix $\M_t = \A_t^\top\A_t+\nu_t\I$ is only a $O(d/k_t)$-approximation of $\A^\top\A$, it is a much better $O(\beta)$-approximation of $\M_{t-1}$. Thus, we can solve a linear system with $\M_{t-1}$ using only $\tilde O(\sqrt{\beta})$ applications of a solver for $\M_t$. However, the cost of applying $\M_t$ decreases by a factor of $\beta$ compared to the cost of applying $\M_{t-1}$. Consequently, at each level of preconditioning the costs decay faster than the number of recursive calls grows (so long as $\beta$ is large enough);  a similar phenomenon occurs in the prior work on graph structured linear systems. The overall running time is then dominated by the top-level operations, i.e., $\tilde O(\nnz(\A) + d^2)$ time.

The above analysis used that decreasing $k_t$ by $\beta$ sufficiently increases $\nu_t$. When $\ma$ is $k$-well-conditioned this occurs for $k_t \gg k$, however once we reach $\A_T\in\R^{\tilde O(k)\times d}$ this may not be the case; consequently, at this point we stop the recursion and solve the corresponding system directly. Since the matrix $\M_T=\A_T^\top\A_T+\nu_T\I$ is given through a low-rank decomposition, a well-known strategy is to use the Woodbury matrix identity: 
\begin{align}
    (\A_T^\top\A_T+\nu_T\I)^{-1} = \frac1{\nu_T}\big(\I-\A_T^\top(\A_T\A_T^\top+\nu_T\I)^{-1}\A_T\big).\label{eq:overview-woodbury}
\end{align}
To decrease the costs associated with inverting $\A_T\A_T^\top+\nu_T\I$, we follow a strategy similar to \cite{derezinski2024faster}. We perform a final preconditioning step of computing $\A_T\mPi^\top\in\R^{\tilde O(k)\times \tilde O(k)}$ for a subspace embedding $\mPi\in\R^{\tilde O(k)\times d}$, to reduce the second dimension of $\A_T$ and construct the preconditioner $\A_T\mPi^\top\mPi\A_T^\top+\nu_T\I$. This idea (we call it \emph{dual preconditioning}, and we revisit it in greater detail in the next section) brings the cost of constructing all of the preconditioning matrices to $\tilde O(T\cdot\nnz(\A)+ k^\omega)$ and since $T = O(\log(d/k))$, recovers \Cref{t:ls}.

\paragraph{Improved Conditioning Dependence.} 
To further improve the above result, we depart from the above analysis and take a closer look at tuning the preconditioning chain. Rather than simply incrementing the regularization amount $\nu_t$ and reducing the dimension $k_t$ by multiplicative constants, we fine-tune this process at each level, to balance out the quality of the preconditioner with the cost of each iteration.

To illustrate this, consider the total cost of all of the matrix-vector products at the $t$-th preconditioning level if we choose the ideal regularization threshold $\nu_t = \frac1{k_t}\sum_{i>k_t}\sigma_i^2$ for that embedding dimension. 
Using $\runtime_{\A_{t-1}}$ to denote the cost of applying $\A_{t-1}$ to a vector (the matrix which is being preconditioned), we get:
\begin{align*}
    \tilde O\Big(\exp\big(\tilde O(t)\big)\runtime_{\A_{t-1}} \sqrt{\nu_{t}/\sigma_d^2}\Big) 
    = \tilde O\bigg(\exp\big(\tilde O(t)\big)d k_{t-1}\sqrt{\frac1{k_{t}}\sum_{i>k_{t}}\frac{\sigma_i^2}{\sigma_d^2}}  \bigg) 
    = \tilde O\bigg(\exp\big(\tilde O(t)\big)d\, \frac{k_{t-1}}{k_{t}} \sqrt{k_{t}\sum_{i>k_{t}}\frac{\sigma_i^2}{\sigma_d^2}} \bigg),
\end{align*}
where $\exp(\tilde O(t))$ accounts for the overhead of propagating the recursive updates to the top level, $k_{t-1}/k_t$, represents the cost associated with the dimension reduction at the $t$-th preconditioning level, and $k_t\sum_{i>k_t}\sigma_i^2/\sigma_d^2$ controls the quality of the preconditioner.
Through an elementary calculation (\Cref{l:generalized-mean}), we bound the preconditioner quality at level $t$ by a quantity that depends on $\bar{\kappa}_{0,1}(\ma)$ and is independent of $k_t$: 
\begin{align}
    \sqrt{k\sum_{i>k}\frac{\sigma_i^2}{\sigma_d^2}} \ \leq \ \sum_{i \in[n]} \frac{\sigma_i}{\sigma_d} = d\cdot\bar\kappa_{0,1}(\A)\quad\text{for all $k$.}\label{eq:ls-fine-bound}
\end{align}
We note that this inequality is essentially tight, which can be observed by choosing $\sigma_i = 1/i$ and comparing both sides.
So, if we were able to control both $k_{t-1}/k_t$  and $\exp(\tilde O(t))$, then this algorithm would attain a $\tilde O(\nnz(\A) + d^2\bar\kappa_{0,1}(\A))$ running time (where $\nnz(\A)$ accounts for the top-level preconditioning),
and a more careful application of \eqref{eq:ls-fine-bound} allows us to further replace $\bar\kappa_{0,1}$ with $\bar\kappa_{k,1}$. Consequently, this would extend our solver from $(k,2)$-well-conditioned systems to the broader class of those that are $(k,1)$-well-conditioned.

Balancing the terms $k_{t-1}/k_t$ and $\exp(\tilde O(t))$ leads to a trade-off in selecting the number of preconditioning levels. If we use too many levels, then $\exp(\tilde O(t))$ grows, whereas if we select too few levels, then the dimension reduction term $k_{t-1}/k_t$ grows.
We balance this trade-off by a recursive scheme where $k_t= O(d\exp(-\alpha t^2))$ for a carefully chosen $\alpha=\tilde O(1)$, so that there are only $T= O(\sqrt{\log(d/k)})$ preconditioning levels (fewer than in the previous scheme). This leads to both terms scaling sub-polynomially, as $\exp(\tilde O(\sqrt{\log d}))$.

The above argument leads to an almost linear running time, up to sub-polynomial factors, scaling with $\bar\kappa_{k,1}(\A)$. However, in Section \ref{s:warm-up-improved}, we show that a further generalization of the spectral bound in \eqref{eq:ls-fine-bound} allows us to absorb those sub-polynomial factors by slightly restricting the condition number, obtaining $\tilde O(\nnz(\A) + d^2\bar\kappa_{k,1+\eta}(\A)\log(d)^{O(1/\eta)})$ running time for any $\eta>0$, after $\tilde O(\nnz(\A)+k^\omega)$ preprocessing (see Theorem \ref{thm:regression:optimized}).

\subsection{Alternating Primal-Dual Regression Solver}
\label{s:overview-primal-dual}

The analysis described in Section \ref{s:overview-warm-up} does not directly yield \Cref{t:ls-optimized}, since it applies only as broadly as $(k,1)$-well-conditioned systems; it does not directly yield a nearly linear running time for $(k,p)$-well-conditioned systems with $p\in(1/2,1)$. Here we describe our approach to proving \Cref{t:ls-optimized}, by devising a more elaborate recursive scheme that reduces both of the dimensions of $\A$, rather than just the leading one.

\paragraph{Tool: Dual Preconditioning.}
A key limitation of the recursive preconditioning scheme described in Section \ref{s:overview-warm-up} is that it only reduces the leading dimension $n$ of $\A\in\R^{n\times d}$, while keeping the second dimension $d$ unchanged. This may appear necessary, since our preconditioning chain relies on spectral approximations of the form $\A_t^\top\A_t+\nu_t\I\approx_4\A_{t-1}^\top\A_{t-1}+\nu_t\I$, and such an approximation requires matching matrix dimensions. What if, instead, our matrices satisfy a ``dual'' spectral approximation guarantee:
\begin{align*}
    \A_t\A_t^\top+\nu_t\I\approx_4 \A_{t-1}\A_{t-1}^\top+\nu_t\I,
\end{align*}
where the second dimension is being reduced? In fact, this guarantee can still be used to precondition a solver for $\A_{t-1}^\top\A_{t-1}+\nu_t\I$, by relying on the same Woodbury formula \eqref{eq:overview-woodbury} we used to solve the bottom level of recursion in the previous section. This suggests that reducing the second dimension in $\A$ is also a viable preconditioning strategy, but it does not fit under the previous recursive preconditioning chain paradigm.

\paragraph{Our Approach.} To incorporate this ``dual'' preconditioning strategy in a recursive framework, we define what we call a \emph{primal-dual regularized preconditioning chain}: $((\A_0,\nu_0),(\A_t,\B_t,\nu_t)_{t \in [T]})$, where again, $\A_0$ is our initial matrix $\A$, and the $\A_t, \B_t$ alternate between satisfying primal and dual approximation conditions:
\begin{align*}
    \B_{t}^\top\B_{t}+\nu_t\I \ \approx_4\ \A_{t-1}^\top\A_{t-1}+\nu_t\I \quad\text{and}\quad
    \A_t\A_t^\top+\nu_t\I \ \approx_4\ \B_t\B_t^\top+\nu_t\I\quad\text{for all } t\in[T].
\end{align*}
To solve this chain recursively, we start by using primal preconditioning, solving for $\A_0^\top\A_0+\nu_0\I$ using $\B_1^\top\B_1+\nu_1\I$, thereby reducing the leading dimension and increasing the regularizer. Then, we use dual preconditioning via the Woodbury formula, twice: first going from $\B_1^\top\B_1+\nu_1\I$ to $\A_1\A_1^\top+\nu_1\I$, in order to reduce the second dimension, and then going from $\A_1\A_1^\top+\nu_1\I$ to $\A_1^\top\A_1+\nu_1\I$, so that we can close the loop and restart the process recursively. We construct the matrices in this chain by sketching $\A$ on both sides, so that, for example, $\A_t = \S_t\A\mPi_t^\top\in\R^{s_t\times \tilde O(s_t)}$, where $s_t=\tilde O(k_t)$ and $\S_t\in\R^{s_t\times n}$, $\mPi_t\in\R^{\tilde O(s_t)\times d}$ are both subspace embeddings, reducing dimensions on both sides of $\A$, with $\B_t$ constructed similarly.

\paragraph{Improved Conditioning Dependence.}
Due to the modified recursion structure in the primal-dual chain, we are able to extend our solver even further, to cover all $(k,1/2)$-well-conditioned systems. To see why this is possible, observe that unlike in the primal case, the above preconditioning chain reduces both of the matrix dimensions as it progresses, going from a $\tilde O(k_t)\times \tilde O(k_t)$ matrix $\A_t$ to a $\tilde O(k_{t+1})\times \tilde O(k_{t+1})$ matrix $\A_{t+1}$. So,  the cost incurred by the solver along the $t$th level of preconditioning can be written as:
\begin{align*}
    \tilde O\Big(\exp\big(\tilde O(t)\big)\runtime_{\A_{t-1}}\sqrt{\nu_{t}/\sigma_d^2}\Big) 
    = \tilde O\bigg(\exp\big(\tilde O(t)\big)
    k_{t-1}^2\sqrt{\frac1{k_{t}}\sum_{i>k_{t}}\frac{\sigma_i^2}{\sigma_d^2}} \, \bigg) 
    = \tilde O\bigg(\exp\big(\tilde O(t)\big)
    \Big(\frac{k_{t-1}}{k_{t}}\Big)^2 \sqrt{k_{t}^3\sum_{i>k_{t}}\frac{\sigma_i^2}{\sigma_d^2}}\,\bigg).
\end{align*}
The different exponent of $k_t$ in the square root term affects the trade-off with the spectral sum, and by extending the bound in \eqref{eq:ls-fine-bound}, we are able to show that this trade-off is captured by $\bar\kappa_{0,1/2}(\A)$ as follows:
\begin{align}
    \sqrt{k^3\sum_{i>k}\frac{\sigma_i^2}{\sigma_d^2}} \ \leq \bigg(\sum_{i \in [n]}\frac{\sigma_i^{1/2}}{\sigma_d^{1/2}}\,\bigg)^2 = d^2\bar\kappa_{0,1/2}(\A)\quad\text{ for all $k$.}\label{eq:overview-primal-dual-bound}
\end{align}
Using a similar recursion sequence $k_t = O(d\exp(-\alpha t^2))$ as in the previous section, we can balance the remaining two terms, $(k_{t-1}/k_t)^2$ and $\exp(\tilde O(t))$, so that they are both sub-polynomial, thus attaining an almost linear running time that scales with $\bar\kappa_{0,1/2}(\A)$. By stopping the recursion at a dimension $\tilde O(k)\times\tilde O(k)$ matrix, and restricting the conditioning dependence slightly, we further obtain the desired nearly linear running time of $\tilde O(\nnz(\A) + d^2\bar\kappa_{k,1/2+\eta}(\A)\log(d)^{O(1/\eta)})$ for any $\eta>0$, after $\tilde O(\nnz(\A)+k^\omega)$ time preprocessing, thus recovering Theorem \ref{t:ls-optimized} (see Theorem \ref{t:primal-dual-optimized} for a detailed statement).

\paragraph{Estimating Regularization Parameters.}
One remaining component to address is finding the right parameter values to attain those running times. Here, the main challenge lies with the regularization parameters $\nu_t \approx \frac1{k_t}\sum_{i>k_t}\sigma_i^2$. Since we do not have the singular values of $\A$, and there are too many parameters to perform a brute force grid search efficiently, we select the parameters by algorithmically verifying that the sequence $(\A_t,\B_t,\nu_t)_{t\in[T]}$ satisfies the approximation conditions associated with the primal-dual chain. 

We do this by induction, starting with $\nu_T$ (which we find by grid search). Then, having obtained an intermediate primal-dual chain $((\A_t,\nu_t),(\A_i,\B_i,\nu_i)_{i\in\{t+1,...,T\}})$, we find the next regularization parameter $\nu_t$ that satisfies $\B_t^\top\B_t+\nu_t\I\approx_4\A_{t-1}^\top\A_{t-1}+\nu_t\I$. To do this, we observe that the approximation condition can be verified by estimating the spectral norm of two matrices, because:
\begin{align}
    \A^\top\A+\nu\I\approx_c \B^\top\B+\nu\I\quad\Leftrightarrow\quad
    \|\A_{\nu}(\B^\top\B+\nu\I)^{-1}\A_{\nu}^\top\|\leq c\text{ and }\|\B_{\nu}(\A^\top\A+\nu\I)^{-1}\B_{\nu}^\top\|\leq c,\label{eq:overview-power}
\end{align}
where $\A_\nu = [\A^\top\mid \sqrt\nu\I]^\top$, and $\B_{\nu}=[\B^\top\mid\sqrt\nu\I]^\top$. Thus, in order to test the spectral approximation condition in the primal-dual chain, it suffices to be able to solve linear systems with $\A_{t-1}^\top\A_{t-1}+\nu\I$ and $\B_t^\top\B_t+\nu\I$, so that we can apply the power method for estimating the two spectral norms. Fortunately, this is precisely what the intermediate preconditioning chain enables us to do.  Iterating this procedure all the way to $\nu_0$, we recover the entire primal-dual chain needed for the algorithm. In order to formalize this scheme, in Section \ref{s:primal-dual-testing} we provide an analysis of the power method for approximating the spectral norm of a matrix given access only to approximate matrix-vector products with that matrix.

\subsection{Implicit Dual Solver for PD Systems}

Here we describe our approach to proving  \Cref{t:psd-optimized}. In order to take advantage of the PD structure when solving the system $\M\x=\b$ given $\M\in\PD{d}$, further insights are required and the key tools of the previous sections do not readily apply. While we ultimately follow a recursive scheme that is similar to our scheme for regression in its analysis, how we compute the low-rank approximation and how we precondition both change substantially. As we shall see, this leads to even more improved conditioning dependence for our PD solver (Theorem \ref{t:psd-optimized}) than is obtained for the regression task (see details in Section \ref{s:psd}).

\paragraph{Tool: Nystr\"om Approximation.} Since we only have access to $\M\in\PD{d}$ and not its square root or a factorization $\A^\top \A$, the previous approach of sketching ``on the inside'' does not apply. Consequently, we cannot readily compute a low-rank approximation in the same way we did. Prior works \cite{avron2017faster,frangella2021randomized,derezinski2024faster} get around this by instead performing the sketching ``on the outside'' via a generalized version of the so-called Nystr\"om method \cite{Williams01Nystrom}. To construct a Nystr\"om approximation, one uses an embedding matrix $\S\in\R^{s\times d}$ as follows:
\begin{align*}
    \tilde\M = \M\S^\top(\S\M\S^\top)^{-1}\S\M = \C\W^{-1}\C^\top,
\end{align*}
where $\C=\M\S^\top\in\R^{d\times s}$ is a tall matrix and $\W=\S\M\S^\top\in\R^{s\times s}$ is a small PD matrix. If we write $\M$ as a factorization $\A\A^\top$, then the Nystr\"om approximation can be written as $\tilde\M=\A\mathbf{P}_{\S\A}\A^\top$, where $\mathbf{P}_{\S\A} = (\S\A)^\top(\S\A(\S\A)^\top)^{-1}\S\A$ is the orthogonal projection onto the row-span of $\S\A$. Note that we are using a ``dual'' factorization $\A\A^\top$ for $\M$, as opposed to the ``primal'' factorization $\A^\top\A$. We do this so that, in each case, the leading dimension of $\A$ is the one being sketched. The intuition behind the Nystr\"om approximation is that if $\S\A$ yields a good primal low-rank approximation of $\A$, then the projection $\mathbf{P}_{\S\A}$ should retain most of the top part of $\M$'s spectrum.

\paragraph{Na\"ive Approach.} 
This primal-dual connection can be leveraged to use the Nystr\"om approximation as a preconditioner. Concretely, if $\A^\top\S^\top\S\A+\nu\I\approx_4\A^\top\A+\nu\I$, then $\|\M-\tilde\M\|\leq 4\nu$ and therefore,
\begin{align}
    \tilde\M + \nu\I \preceq \M +\nu\I \preceq 5(\tilde\M + \nu\I). \label{eq:nystrom}
\end{align}
Building on this, \cite{derezinski2024faster} observed that in the $k$-well-conditioned setting $\tilde\M+\nu\I$ serves as a useful crude preconditioner for $\M$.
Using a calculation similar to that in the regression setting, $\S \in \R^{\tilde O(k)\times d}$ can be constructed so that whp.\ the preconditioned system has relative condition number 
\begin{align*}
    O\Big(\frac{\nu}{\lambda_d(\M)}\Big) = O\bigg(\frac1k\sum_{i>k}\frac{\lambda_i(\M)}{\lambda_d(\M)}\bigg) = O\Big(\frac dk\bar\kappa_{k,1}(\M)\Big),
\end{align*}
where note that, when compared with \eqref{eq:overview-warm-up-naive}, we replaced $\bar\kappa_{k,2}^2$ by $\bar\kappa_{k,1}$ thanks to the PD structure of $\M$. 
Thus, using preconditioned Chebyshev, $\tilde O(\sqrt{\bar\kappa_{k,1}(\M)\cdot d/k})$ iterations are sufficient to solve for $\M$.
Once again, the crude preconditioner incurs a sub-optimal $\sqrt{d/k}$ factor which prevents this na\"ive approach from attaining nearly-linear time. However, it is not immediately clear how to build a regularized preconditioning chain out of Nystr\"om approximations, given that applying the Nystr\"om matrix $\tilde\M$ to a vector requires solving an unregularized PD linear system with $\W=\S\M\S^\top$.

\paragraph{Our Approach.} To address these challenges, we extend the Nystr\"om method into a black-box reduction for solving a large regularized PD system using a smaller regularized PD system. The first step in this reduction is an observation that we can still retain the Nystr\"om preconditioning guarantee \eqref{eq:nystrom} if we replace the inner PD matrix $\W=\S\M\S^\top$ with its regularized version, $\W_{\nu} = \S\M\S^\top+\nu\I$. Writing this in terms of the dual factorization, one might hope that we can use it to reduce from solving for
$\M_{t}+\nu_{t}\I = \A_{t}\A_{t}^\top+\nu_{t}\I$ to solving for a strictly smaller matrix $\M_{t+1}+\nu_{t+1}\I = \S_{t+1}\A_{t}\A_{t}^\top\S_{t+1}^\top+\nu_{t+1}\I$, which would offer the key recursion for a regularized preconditioning chain in this setting.

We show that this reduction is indeed possible, although it requires a careful multi-stage preconditioning procedure of its own. Let $\tilde\M=\C\W_{\nu_t}^{-1}\C^\top$ be our modified Nystr\"om approximation for $\M_{t}=\A_{t}\A_{t}^\top$ and let $\A_{t+1}=\S_{t+1}\A_{t}$. Our goal is to solve for $\M_{t}+\nu_t\I$ given only a solver for $\M_{t+1}+\nu_{t+1}\I$, where $\M_{t+1}=\A_{t+1}\A_{t+1}^\top$, using  $\tilde\M+\nu_t\I$ as an intermediate preconditioner. If all we needed was to apply  $\tilde\M$ to a vector, then we would have the desired result, since $\W_{\nu_t} = \M_{t+1}+\nu_t\I$. However, we need to solve linear systems in the Nystr\"om preconditioner $\tilde\M+\nu_t\I$; for this, we apply the Woodbury matrix identity, to obtain:
\begin{align*}
    (\tilde\M+\nu_t\I)^{-1} = \frac1{\nu_t}\big(\I - \C(\C^\top\C+\nu_t\W_{\nu_t})^{-1}\C^\top\big).
\end{align*}
It remains to address solving $\C^\top\C+\nu_t\W_{\nu_t}$, which we do by showing that $\C^\top\C+\nu_t\W_{\nu_t}\approx_{O(1)} (\W_{\nu_t/2})^2$, essentially as a byproduct of \eqref{eq:nystrom}.
Thus, overall, we solve $\M_{t}+\nu_t\I$ by preconditioning with $\tilde\M+\nu_t\I$, which we solve by inverting $\C^\top\C+\nu_t\W_{\nu_t}$, which we do by preconditioning with $(\W_{\nu_t/2})^2$, which can be done by preconditioning twice with $\W_{\nu_{t+1}}=\M_{t+1}+\nu_{t+1}\I$. All of these preconditioners are constant factor approximations of the corresponding matrices, except for the last one which requires $\tilde O(\sqrt{\nu_{t+1}/\nu_t})$ iterations.

Fortunately, this somewhat elaborate reduction can be treated black box as implementing a recursive preconditioning scheme for $\M_t+\nu_t\I\in\R^{\tilde O(k_t)\times \tilde O(k_t)}$ with $\M_{t+1} = \S_{t+1}...\S_1\M_0\S_1^\top...\S_{t+1}^\top=\S_{t+1}\M_{t}\S_{t+1}^\top$, to solve $\M_0=\M$. What remains is to tune the decreasing sizes $k_t$ and increasing regularizations $\nu_t$. 

\paragraph{Improved Conditioning Dependence.} 
The computational cost analysis follows similarly for the PD solver as in the primal-dual regression solver, since here again we are reducing both dimensions as we go from $\M_t$ to $\M_{t+1}$ at each step of the recursion. However, due to the implicit dual preconditioning, the spectral sum bounds that arise are based on the singular values of $\A$, which correspond to the square roots of the eigenvalues of $\M=\A\A^\top$. This means that the averaged condition number from \eqref{eq:overview-primal-dual-bound}, when converted from $\A$ to $\M$ yields an even better conditioning dependence, since:
\begin{align*}
    \bar\kappa_{k,1/2}(\A) = \bigg(\sum_{i>k}\frac{\sigma_i(\A)^{1/2}}{\sigma_d(\A)^{1/2}}\bigg)^2 = \bigg(\sum_{i>k}\frac{\lambda_i(\M)^{1/4}}{\lambda_d(\M)^{1/4}}\bigg)^2 = \sqrt{\bar\kappa_{k,1/4}(\M)}.
\end{align*}
Thus, following the same reasoning as in Section \ref{s:overview-primal-dual}, again letting $k_t = O(d\exp(-\alpha t^2))$ with the factor $\alpha$ adjusted for the implicit dual preconditioning scheme, we obtain the nearly linear running time of $\tilde O(d^2\sqrt{\bar\kappa_{k,1/4+\eta}(\M)}\log(d)^{O(1/\eta)})$ for any $\eta>0$ after $\tilde O(\nnz(\A)+k^\omega)$ time preprocessing, recovering Theorem \ref{t:psd-optimized} (see Theorem \ref{t:dual-main-optimized} for a detailed statement). Similarly to the primal-dual case, we rely on an inductive approach based on inexact power method to find the appropriate regularization parameters $\nu_t$ during runtime. Here, the implementation requires additional care, since we do not have direct access to matrices $\A_t$ but rather only to $\M_t$ (details in \Cref{s:dual-condition-tester}).

\subsection{Applications: Spectrum Approximation}

Finally, we discuss how our linear solvers imply faster algorithms for a range of spectrum approximation tasks, in particular leading to the  nearly linear time algorithm for nuclear norm approximation (Theorem~\ref{t:nuclear}). Our approach is primarily based on the algorithmic framework of \cite{musco2018spectrum} which studies algorithms for finding approximate histograms of the singular value distribution of an arbitrary matrix $\A\in\R^{n\times d}$. They show that, given access to a linear solver for $\A^\top\A+\lambda\I$ with an appropriate value of $\lambda$, we can use rational polynomial functions and trace estimation techniques to efficiently count singular values in a specific range. This tool can then be used to efficiently approximate various functions of the singular values, including the Schatten, Ky Fan, and Orlicz norms, among others. 

Since the linear systems that arise as subroutines in their framework are regularized, they exhibit properties similar to those of $k$-well-conditioned systems. So, it is natural to ask what specific conditioning property best fits into this model. Remarkably, as a by-product of their analysis, we are able to show that our notion of a $(0,p)$-well-conditioned system \emph{precisely} corresponds to the linear systems that arise in Schatten $p$-norm approximation (\Cref{c:schatten-reduction}). This follows since, to produce a constant-factor approximation of $\|\A\|_p^p$, we need to solve a linear system with regularizer $\lambda=\Omega\big((\frac1d\|\A\|_p^p)^{2/p}\big)$. Letting $\A_\lambda$ denote the matrix $\A$ augmented by $\sqrt\lambda\I_d$ so that $\A_\lambda^\top\A_\lambda=\A^\top\A+\lambda\I$, this gives:
\begin{align*}
    \bar\kappa_{0,p}(\A_\lambda) = \bigg(\frac1d\sum_{i \in [d]}\frac{\sigma_i^p}{(\sigma_d+\sqrt{\lambda})^p}\bigg)^{1/p} \leq \frac{\|\A\|_p}{\sqrt\lambda d^{1/p}} = O(1).
\end{align*}
This implies a reduction from Schatten $p$-norm approximation to solving a $(0,p)$-well-conditioned system $\tilde O(1)$ times. Applying our primal-dual regression solver with $p=1$ (Theorem \ref{t:ls-optimized}) as a subroutine in this framework recovers the nuclear norm approximation guarantee from Theorem \ref{t:nuclear}. Similar nearly linear time guarantees follow for all Schatten norms with $p>1/2$ (or $p>1/4$ when $\A$ is~PD). On the other hand, when $p\in(0,1/2)$ (for general $\A$), then we can still use our solver, but it will require larger $k$ to compensate for the fact that $\bar\kappa_{k,1/2}$ is bigger than constant. Choosing larger $k$ makes $\bar\kappa_{k,1/2}$ progressively smaller, but we also incur the preprocessing cost of $\tilde O(k^\omega)$. Optimizing over this trade-off leads to an $\tilde O\big((\nnz(\A)+d^{2+(\omega-2)(1+\frac{p\omega}{0.5-p})^{-1}+o(1)})\poly(1/\epsilon)\big)$ running time for any fixed $p\in(0,1)$, which is the new state-of-the-art time complexity for Schatten $p$-norm approximation in this regime (see~Section~\ref{s:spectrum}).

\addtocontents{toc}{\protect\setcounter{tocdepth}{2}}
\section{Preliminaries}
\label{s:prelim}
\paragraph{Notation.}
For integer $k$, $[k] \defeq \{1,\ldots,k\}$. We let $\log(\cdot)$ denote that natural logarithm (base $e$) and write $\log_c(\cdot)$ when the base is instead $c$. We let $\SM{d}$ denote the set of $d \times d$ real-valued symmetric matrices, $\PSD{d}$ denote the set of $d \times d$ real-valued symmetric positive semi-definite (PSD) matrices, and $\PD{d}$ denote the set of $d \times d$ real-valued symmetric positive definite (PD) matrices. For $\M \in \SM{d}$, we let $\lambda_1(\M)\geq\lambda_2(\M)\geq ...\geq\lambda_d(\M)$ denote its eigenvalues, and for $\M \in \PSD{d}$, we let $\|\x\|_\M =\sqrt{\x^\top\M\x}$. Similarly, for $\A\in\R^{n\times d}$ with $n\geq d$ we let $\sigma_1(\A)\geq \sigma_2(\A)\geq ...\geq \sigma_d(\A)$ denote its singular values. For $\M \in \SM{d}$ and $\A\in\R^{n\times d}$ with $n\geq d$ we denote the normalized total tail of eigenvalues and squared singular values as 
\[
\lambAvg_k(\mm) \defeq \frac{1}{k} \sum_{i > k} \lambda_i(\mm)
\quad\text{ and }\quad\sigAvg_k(\ma) \defeq \frac{1}{k} \sum_{i>k}\sigma_i(\A)^2\,,
\quad\text{ respectively.}
\]
Also, we define $\kappa(\A) \defeq \sigma_{1}(\A)/\sigma_{d}(\A)$ as the condition number of $\A$.
For $\A,\B \in \PSD{d}$, we use $\A\preceq \B$ to denote the Loewner ordering, and we use $\A\approx_{c}\B$ to denote the pair of inequalities $c^{-1}\B\preceq\A\preceq c\B$. We use $\|\v\|$ to denote the Euclidean norm of a vector, $\|\A\|$ 
to denote the spectral norm of a matrix, and $\|\A\|_F$ to denote  the Frobenius norm of a matrix.

\paragraph{Asymptotic Notation.}
$\otilde(T)$ denotes $O(T \cdot \poly \log(n+d))$. 
We use the term ``with high probability'' (whp.) in reference to the properties of an algorithm if it has the property with probability $1 - \frac{1}{(n+d)^c}$ for any $c$ that can be obtained by changing the constants hidden in the $O(\cdot)$ notation describing the algorithm's running time.
When $\log(\cdot)$ is used inside $O(\cdot)$ notation we replace the input argument with the maximum of the argument and $e^2$, e.g., so $\otilde(n^2 \log^c(\kappa(\ma)))$ runtimes do not convey faster than $\otilde(n^2)$ time algorithms as $\kappa(\ma) \rightarrow 1$ or fail to grow as $c$ grows. 
We use $\nnz(\ma)$ to denote the number of nonzero entries in $\ma \in \R^{n \times d}$ and when $\nnz(\ma)$ appears in $O(\cdot)$ notation we add $n + d$ to it to simplify runtime bounds.

\paragraph{Numerical Precision.} In this paper, for simplicity, we perform our analysis in the real-RAM model of computation
where arithmetic operations are performed exactly. Showing that a variant of our algorithms obtain comparable running times when arithmetic operations are performed approximately is a natural question to improve this work. Prior work on numerical stability and bit precision of the types of methods we use in this paper \cite{demmel2007fast,greenbaum1989behavior,musco2018stability, Peng18} and the fact that our algorithms rely on matrix approximations and preconditioned AGD with error in the preconditioner gives hope that such an improvement may be possible. 

\paragraph{Runtimes and Solvers.} In the remainder of the preliminaries we give facts regarding linear system solvers and matrix sketches that we use throughout the paper.

\begin{definition}[PD Solvers]
We call $\widehat{\x} \in \R^d$ an \emph{$\epsilon$-(approximate )solution} to $\M \x = \b$ for $\M \in \PD{d}$ and $\b \in \R^{d}$ if $\norm{\widehat{\x}-\mm^{-1}\b}_{\mm}^{2}\leq\epsilon\norm \b_{\mm^{-1}}^{2}$. 
We call $f:\R^{d}\rightarrow\R^{d}$
an\emph{ $\epsilon$-solver for $\mm\in \PD{d}$} if for all $\b \in \R^d$, $f(\b)$ is an $\epsilon$-solution to $\M \x = \b$.
\end{definition}

We use $\runtime_f$ to denote the runtime upper-bound for an algorithm $f$, and for a matrix $\A\in\R^{n\times d}$, we use $\runtime_{\A}$ to denote the maximum cost of applying $\A$ or $\A^\top$ to a vector, with an added $n+d$.
We assume that when $f$ is a solver for $\mm \in \R^{d \times d}$ then $\runtime_f = \Omega(d)$ to simplify our expressions, and when $\epsilon=\Theta(1)$, we say that it is constant accuracy solver.

\begin{lemma}[Preconditioned AGD; Adapted from Theorem 4.4 \cite{js_talg}\footnote{This lemma specializes Theorem 4.4 of \cite{js_talg} to deterministic solvers. Also note that Theorem 4.4 \cite{js_talg} did not provide an explicit iteration count, but the one stated follows from their proof and associated pseudocode.}]
\label{lem:precon} For any $\mm,\mn \in \PD{d}$ with $\mm\preceq\mn\preceq\kappa\mm$, $\kappa\geq1$, and $f$ that is a $\frac{1}{10\kappa}$-solver
for $\mn$ there is procedure, $\precon_{\mm,f,\kappa,\epsilon}(\cdot)$, that is an $\epsilon$-solver for $\mm$ that applies $f$ and $\mm$ at most $\lceil 4 \sqrt{\kappa}\log(2/\epsilon)\rceil$ times each, and uses $O(n \sqrt{\kappa} \log(1/\epsilon))$ additional time to run.
\end{lemma}

\begin{definition}[$(\epsilon,\nu)$-embedding]
    We call $\S\in\R^{s\times n}$ an \emph{$(\epsilon,\nu)$-(regularized )embedding  for $\A \in \R^{n \times d}$} if
    \begin{align*}
         \A^\top\S^\top\S\A+\nu\I&\approx_{1+\epsilon}\A^\top\A+\nu\I.
    \end{align*}    
\end{definition}

\begin{definition}[Oblivious subspace embeddings]
    A random matrix $\S\in \R^{s\times n}$ has $(\epsilon,\delta,d,\ell)$-OSE moments if for all matrices $\U\in\R^{n\times d}$ with orthonormal columns, $\E\|\U^\top\S^\top\S\U-\I_d\|^\ell<\epsilon^\ell\delta$.
\end{definition}
 Note that, via a simple application of Markov's inequality, if $\S$ has $(\epsilon,\delta,d,\ell)$-OSE moments for $\epsilon,\delta\in(0,1/2)$, then it is a $(2\epsilon,\nu)$-embedding for any $\A\in\R^{n\times d}$ and any $\nu\geq 0$.
 
\begin{definition}[Sparse embedding matrices]\label{d:sketching}
We say that a random matrix $\S\in\R^{s\times n}$ is a sparse embedding matrix with sketch size $s$ and $b$ non-zeros per column, if it has independent columns, and each column consists of $b$ random $\pm1/\sqrt b$ entries placed uniformly at random without replacement.
\end{definition}
In the following lemma, we cite one of many standard guarantees for sparse oblivious subspace embeddings, e.g., \cite{nn-sparse,cohen2016nearly,chenakkod2023optimal}. We choose the one which gives the optimal dependence on the sketch size $s$ to illustrate how one could optimize the construction of our recursive schemes.
\begin{lemma}[\cite{chenakkod2024optimal}]\label{l:ose-sparse}
    For $\epsilon,\delta\in(0,1/2)$, a sparse embedding matrix $\S$ with sketch size $s=O((d+\log(1/\delta\epsilon))/\epsilon^2)$ and $b=O(\log^2(d/\delta\epsilon)/\epsilon+\log^3(d/\delta\epsilon))$ non-zeros per column has $(\epsilon,\delta,d,16\log(d/\delta\epsilon))$-OSE moments.
\end{lemma}

\section{Warm-up: Recursive Regression Solver}
\label{s:ls}

In this section, we provide our simplest recursive preconditioning method for solving regression, and more broadly, solving linear systems in $\A^\top \A$ given $\A\in\R^{n \times d}$. Leveraging this solver, we prove \Cref{t:ls}. While the approach presented here does not fully recover our main algorithmic result for regression (\Cref{t:ls-optimized}), we use it to introduce many of the key ideas necessary for the proof of that result (which is in \Cref{s:primal-dual}).

\subsection{Reduction Tools}
\label{sec:primal:tools}

Here we provide our tools for low-rank approximation that we leverage to develop our regression solver.

\begin{lemma}\label{l:reg-approx}
    Given $k\leq d$, $\epsilon\in(0,1)$, $\A \in \R^{n \times d}$, suppose that $\S\in\R^{s\times n}$ has $(\epsilon,\delta,2k,\ell)$-OSE moments for some $\ell\geq 2$.
    Then, with probability $1-\delta$, $\S$ is an $(6\epsilon,\nu)$-embedding for $\A$ with $\nu = \sigAvg_k (\ma)$.
\end{lemma}
\begin{proof}
Let $\sigma_i \defeq \sigma_i(\ma)$.
If $\S$ has $(\epsilon,\delta,2k,\ell)$-OSE moments, then we will use a standard approximate matrix product guarantee \cite{cohen2016optimal} applied to
$\B=\A(\A^\top\A+\nu\I)^{-1/2}$. Namely, Theorem 1 of \cite{cohen2016optimal} implies that $\S$ satisfies:
\begin{align}
  \|\B^\top\S^\top\S\B-\B^\top\B\|\leq \epsilon\big(\|\B\|^2+\|\B\|_F^2/k \big).\label{eq:B}
\end{align}
Observe that $\|\B\|\leq 1$. Also, if $k=d$, then $\nu=0$ and $\|\B\|_F^2/d=1$, whereas if $k<d$ then $\nu>0$ and:
\begin{align*}
\|\B\|_F^2=\tr(\A(\A^\top\A+\nu\I)^{-1}\A^\top)=\sum_i\frac{\sigma_i^2}{\sigma_i^2+\nu}\leq k + \frac1\nu\sum_{i>k}\sigma_i^2=2k,
\end{align*}
we obtain that the right-hand side of \eqref{eq:B} can be bounded by $3\epsilon$, which can be written as $-3\epsilon\I\preceq \B^\top\S^\top\S\B -\B^\top\B\preceq3\epsilon\I$. Multiplying all expressions by $(\A^\top\A+\nu\I)^{1/2}$ on each side, we get:
\begin{align*}
    -3\epsilon(\A^\top\A+\nu\I)\preceq\A^\top\S^\top\S\A-
  \A^\top\A&\preceq 3\epsilon(\A^\top\A+\nu\I),
\end{align*}
which in turn implies that $\A^\top\S^\top\S\A+\nu\I \approx_{1+6\epsilon}\A^\top\A+\nu\I$, as long as $\epsilon\leq 0.1$. 
\end{proof}
\begin{corollary}\label{c:fast-embedding}
    Given $k\leq d$, $\epsilon\in(0,1)$, $\A\in\R^{n\times d}$, in $O(\nnz(\A)\log^3(k/\delta)/\epsilon)$ time we can compute $\S\A\in\R^{s\times d}$ such that $s=O((k+\log(1/\delta\epsilon))/\epsilon^2)$ and with probability $1-\delta$, $\S$ is an $(\epsilon,\sigAvg_k (\ma))$-embedding for $\A$.
\end{corollary}
\begin{proof}
    The claim follows by using Lemma \ref{l:ose-sparse} to show that a sparse subspace embedding matrix $\S$ with $O(\log^3(k/\delta)/\epsilon)$ non-zeros per column has $(\epsilon,\delta,2k,\ell)$-OSE moments, and then applying Lemma \ref{l:reg-approx} to $\S$. 
\end{proof}

Next, we present a linear system solver for matrices which are described via a low-rank decomposition of the type that arises in the Woodbury formula. 
\begin{lemma}[Woodbury solver]\label{l:woodbury}
    Consider $\M = \C\W^{-1}\C^\top+\nu\I$, where $\nu>0$, $\W\in\PD{d}$, and $\C\in\R^{n\times d}$. Given an $\frac{\epsilon\nu^2}{\|\M\|^2}$-solver $f$ for $\C^\top\C+\nu\W$, consider $g(\b) = \frac1\nu(\b-\C\cdot f(\C^\top\b))$. Then, $g$ is an $\epsilon$-solver for $\M$ that applies $f$, $\C$, and $\C^\top$ once each, and takes $O(n)$ additional time.
\end{lemma}
\begin{proof}
    Let $\y^* = (\C^\top\C+\nu\W)^{-1}\C^\top\b$, and let $\hat\y=f(\C^\top\b)$. Then, $\|\hat\y-\y^*\|_{\C^\top\C+\nu\W}^2\leq \tilde\epsilon\|\y^*\|_{\C^\top\C+\nu\W}^2$ where $\tilde\epsilon = \frac{\epsilon\nu^2}{\|\M\|^2}$. Using the Woodbury formula we have:
    \begin{align*}
        (\M+\nu\I)^{-1} = \frac1\nu\big(\I-\C(\C^\top\C+\nu\W)^{-1}\C^\top\big),
    \end{align*}
    which implies that our goal is to approximate $\x^* = \frac1\nu(\b-\C\y^*)$. We do this using $\hat\x=\frac1\nu(\b-\C\hat\y)=g(\b)$. First note that we have $\|\hat\x-\x^*\|_{\M}^2 = \frac1{\nu^2}\|\C(\hat\y-\y^*)\|_{\M}^2=\frac1{\nu^2}\|\hat\y-\y^*\|_{\C^\top\M\C}^2$. 
    Observe that $\C^\top\M\C\preceq \|\M\|(\C^\top\C+\nu\W)$. From this, it follows that:
\begin{align*}
    \|\hat\x-\x^*\|_{\M}^2
    &\leq \frac{\|\M\|}{\nu^2}\|\hat\y-\y^*\|_{\C^\top\C+\nu\W}^2
    \leq \frac{\|\M\|}{\nu^2}\tilde\epsilon\|\y^*\|_{\C^\top\C+\nu\W}^2
    = \frac{\epsilon}{\|\M\|} \|\b\|_{\C(\C^\top\C+\nu\W)^{-1}\C^\top}^2
    \leq \frac{\epsilon}{\|\M\|}\|\b\|^2,
    \end{align*}
    and observing that $\|\b\|^2=\|\M^{1/2}\M^{-1/2}\b\|^2\leq \|\M\|\|\b\|_{\M^{-1}}^2$  concludes the claim.
    \end{proof}

Using the above Woodbury solver, we construct an efficient solver for the bottom level of our recursion.
\begin{lemma}\label{l:msp}
There is an algorithm that given $\A \in \R^{s \times d}$ and $\nu > 0$ in $\otilde(\nnz(\A)\log^3(s/\delta) + (s+\log(1/\delta))^\omega)$ time with probability $1-\delta$ returns an $\epsilon$-solver, $p$, 
for $\A^\top \A + \nu\I$ with $\runtime_p = O((\nnz(\A) + s^2)\log(\kappa/\epsilon))$ for $\kappa=1+\|\A\|^2/\nu$.
\end{lemma} 
\begin{proof}
Without loss of generality, assume that $s\leq d$, since otherwise the result is immediate. Based on Lemma \ref{l:woodbury} (with $\C=\A^\top$ and $\W=\I$) it suffices to construct a solver for $\A\A^\top+\nu\I$ of sufficient accuracy.

To construct a solver for $\A\A^\top+\nu\I$, we observe that since $\A$ is a wide matrix, we can sketch it using a $(1,\nu)$-embedding $\S$ for $\A^\top$, so that $\A\S^\top\S\A^\top+\nu\I\approx_{2}\A\A^\top+\nu\I$. To do this, in fact it suffices to use a sparse subspace embedding $\S\in\R^{s'\times d}$ with $s'=O(s+\log(1/\delta))$ and $b=O(\log^3(s/\delta))$ non-zeros per column. We compute $\S\A^\top$ and then $(\A\S^\top\S\A^\top+\nu\I)^{-1}$, at the cost of $O(\nnz(\A)\log^3(s/\delta)+(s+\log(1/\delta))^\omega)$, and use this as a preconditioner for solving $\A\A^\top+\nu\I$. Thus, via Lemma \ref{lem:precon}, we can obtain an $\epsilon_f$-solver $f$ for $\A\A^\top+\nu\I$ which runs in time $O(\nnz(\A)+ s^2)\log(1/\epsilon_f))$. 

By Lemma \ref{l:woodbury}, it suffices to use accuracy $\epsilon_f=\frac{\epsilon\nu^2}{(\|\A\|^2+\nu)^2}$. This gives us an $\epsilon$-solver for $\A^\top\A+\nu\I$ that runs in time $O((\nnz(\A) + s^2)\log(\kappa/\epsilon))$. Note that, in order to find $\epsilon_f$, we can use the power method to get a $2$-approximation of $\|\A\|$ in $O(\nnz(\A)\log(s)\log(1/\delta))$ time (Lemma \ref{l:power}).
\end{proof}

\subsection{Recursive Preconditioning}
\label{sec:primal:recursive}

Below we define a \emph{preconditioning chain} which is a sequence of matrices that relatively spectrally approximate each other. This definition and notation is similar to prior work on solving graph-structured linear systems such as Laplacians \cite{SpielmanT14,KoutisMP10,KoutisMP11,CohenKMPPRX14,js_talg}.
\begin{definition}[Preconditioning Chain]
\label{def:pchain}
We call $C = (\mm_0, ((\mm_t, \kappa_t))_{t \in [T]})$ a \emph{length $T$, $d$-dimensional preconditioning chain}, denoted $C \in \chain_T^d$, if $T \geq 1$, $\mm_t \in \PD{d}$ and $\kappa_t \in \R_{\geq 1}$ for all $t \in \{0\} \cup [T]$, and
\[
    \mm_{t-1} \preceq \mm_{t} \preceq \kappa_{t} \mm_{t-1}
\text{ for all }
    t \in [T]\,.
\]
\end{definition}

\begin{algorithm}
\caption{Recursive Preconditioning}
\label{alg:rprecon}
\SetKwProg{function}{function}{:}{}
\mycomment{$\rprecon$ parameters: $C = (\mm_0, ((\mm_t, \kappa_t))_{t \in [T]}) \in \chain_T^d$, $\baseSolve$ that is a $\frac{1}{10 \kappa_T}$ solver for $\mm_T$, $\epsilon \in (0, 1)$}
\mycomment{$\rprecon$ input: $\b \in \R^d$}
\mycomment{$\rprecon$ output: $\xout \in \R^d$ with $\norm{\xout - \M^{-1} \b}_\M^2 \leq \epsilon \norm{\b}_\M^2$ where $\mm = \mm_0(C)$ }
\BlankLine
\function{$\rprecon_{C, \baseSolve, \epsilon}(\b)$}{
    \lIf{$T = 1$}
    {
        $\preSolve = \baseSolve(\cdot)$
    }
    \Else{
        \mycomment{$C_1$ is a preconditioning subchain of $C$ that $\preSolve$ uses to solve $\mm_1$}  
        \label{line:rprecon:rsolver}
            $\preSolve = \rprecon_{C_1, \baseSolve, \frac{1}{10 \kappa_1}}(\cdot)$
            where $C_1 \defeq (\mm_1,
        ((\mm_2,\kappa_2),\ldots,(\mm_T,\kappa_T)))$
            \;
    }
    \textbf{return}: $\xout = \precon_{\mm_0, \preSolve, \kappa_1, \epsilon}(\b)$
    \label{line:rpagd:out}\;
}
\end{algorithm}

Next, we give a general-purpose recursive preconditioning guarantee with respect to the above definition of a preconditioning chain. We will build on this to construct our simple solver for regression (more elaborate recursive frameworks  will be needed for the primal-dual and implicit dual solvers in later sections). 
\begin{lemma}[Recursive Preconditioning]
\label{lem:recurse}
Let $C = (\mm_0, ((\mm_t, \kappa_t))_{t \in [T]}) \in \chain_T^d$, let $\baseSolve$ be a $\frac{1}{10 \kappa_T}$-solver for $\mm_T$, and let $\epsilon \in (0,1)$. 
Then $\genSolve = \rprecon_{C, \baseSolve, \epsilon}(\cdot)$ (see \Cref{alg:rprecon}) is an $\epsilon$-solver for $\mm_{0}$ with 
\[
\runtime_{\genSolve} =O\left(K_T \runtime_{\baseSolve}+ \textstyle{\sum_{i\in[T]}} K_i \runtime_{\mm_{i-1}}\right)
\]
    where 
\[
\kappa_0 \defeq (10 \epsilon)^{-1}
	\text{ and }
K_t \defeq \prod_{i \in [t]} \lceil 4 \sqrt{\kappa_i} \log(20\kappa_{i-1}) \rceil
	\text{ for all }
t \in [T]
\,.
\]
\end{lemma}

\begin{proof}
First we show that it suffices to prove that 
\begin{equation}
\label{eq:recurse_suffice}
\text{
$\preSolve$
 is a $\frac{1}{10 \kappa_1}$-solver for $\mm_1$ with
 }
 \runtime_\preSolve = 
 \frac{K_T \cdot O(\runtime_{\baseSolve})  + \sum_{i\in\{2,\ldots,T\}} K_i \cdot O(\runtime_{\mm_{i-1}})}{K_1}\,.
\end{equation}
$\mm_0 \preceq \mm_1 \preceq \kappa_1 \mm_0$
by the definition of $\chain_T^d$ (\Cref{def:pchain}) and consequently \Cref{lem:precon} applies to the call to $\precon$ on \Cref{line:rpagd:out}. \Cref{lem:precon} shows that $\xout$ is an $\epsilon$-approximate solution to $\mm_0 \x = \b$ and that 
\begin{equation*}
\runtime_{\genSolve} = O(\lceil 4 \sqrt{\kappa_1}) \log(2/\epsilon_0) \rceil (\runtime_{\mm_0} + \runtime_\preSolve + d)
= O(K_1 (\runtime_{\mm_0} + \runtime_\preSolve))
\end{equation*}
since $\kappa_0 = (10\epsilon)^{-1}$, $\runtime_{\mm_0} = \Omega(d)$, and $C_1$ can be computed implicitly in $O(1)$ time. \eqref{eq:recurse_suffice} then yields the result.

We now prove \eqref{eq:recurse_suffice} by induction on $T$. The base case of $T = 1$ follows trivially from the definition of $\baseSolve$ as $\preSolve = \baseSolve$. Next we suppose \eqref{eq:recurse_suffice} holds for $T = T'\geq 1$ and consider the case of $T = T' + 1$. In this case,  $C_1 \in C_{T'}^{d}$ (\Cref{line:rprecon:rsolver}) by construction. Consequently, the preceding paragraph and that \eqref{eq:recurse_suffice} holds for $T = T'$ implies that $\preSolve = \rprecon_{C_1, \baseSolve, \frac{1}{10 \kappa_1}}(\cdot)$ (\Cref{line:rprecon:rsolver}) is a $\frac{1}{10\kappa_1}$ for $\mm_1$ with $\runtime_{\preSolve}$ with 
\[
\runtime_{\genSolve} = \hat{K}_T \cdot O(\runtime_{\baseSolve}) + \sum_{i\in\{2,\ldots,T\}} \hat{K}_i \cdot O(\runtime_{\mm_{i-1}})
\]
    where 
\[
\hat{K}_t \defeq \prod_{i = 2}^{T} \lceil 4 \sqrt{\hat{\kappa}_i} \log(20\hat{\kappa}_{i-1}) \rceil
    \text{ , }
\hat{\kappa}_1 \defeq (10 \cdot (1/(10 \kappa_1)))^{-1} 
\text{ , and }
\hat{\kappa}_t \defeq \kappa_t
\text{ for all }
t \in \{2,\ldots,T\}
\,.
\]
Since each $\hat{K}_t = K_t / K_1$ for all $t \in \{2,\ldots,T\}$ the result follows by induction.
\end{proof}

\subsection{Regularized Preconditioning for Regression}
\label{sec:primal:regression}

Building on the analysis of the previous section, here we give an algorithm for regression and prove \Cref{t:ls}, one of our key results regarding regression. Our algorithms use a particular type of preconditioning chain, which we call a \emph{regularized preconditioning chain}.

\begin{definition}[Regularized Preconditioning Chain]
\label{def:primalchain}
We call $P = ((\ma_t, \nu_t)_{t \in \{\{0\} \cup [T]\}})$ a \emph{length $T$, $d$-dimensional, regularized preconditioning chain}, denoted $P \in \primalchain_T^{d}$, if $T \geq 1$, each $\ma_t \in \R^{n_t \times d}$, and\footnote{\label{footnote:4}The constant of $4$ in $\ma_{t}^\top\ma_{t} + \nu_t \mI \approx_4
\ma_{0}^\top\ma_{0} + \nu_t \mI$ is somewhat arbitrary but chosen to be consistent with \Cref{s:psd} where the constant is made large enough to algorithmically verify the spectral approximation condition in Lemma \ref{l:dual-condition-tester}.}
\[
0 < \nu_0 \leq \cdots \leq \nu_T
    \text{ and }
\ma_{t}^\top\ma_{t} + \nu_t \mI \approx_4
\ma_{t - 1 }^\top\ma_{t - 1} + \nu_t \mI
    \text{ for all }
    t \in [T]\,.
\]
\end{definition}

Our regression algorithms then apply the general recursive preconditioning framework to regularized preconditioning chains. The following theorem gives the main guarantees on linear system solvers that we obtain from regularized preconditioning chains.

\begin{theorem}[Recursive Regularized Preconditioning]
\label{thm:apx-gen}
There is an algorithm that given $((\ma_t, \nu_t)_{t \in \{0\} \cup [T]\}}) \in \primalchain_T^{d}$ where each $\ma_T \in \R^{n_T \times d}$, in $\otilde(\nnz(\ma_T)) + O(n_T^\omega)$ time develops a $\runtime$-time $\epsilon$-solver for $\ma_0^\top \ma_0 + \nu_0 \I$ with 
\begin{equation}
\label{eq:regression-runtime-general}
\runtime = 
O\left(\sum_{t\in[T]} R_t \cdot \nnz(\ma_{t -1}) + R_T \cdot \left(\nnz(\ma_t) + n_T^2 
    \log\left(\frac{\nu_t + \norm{\ma}^2}{\nu_{T - 1}}\right)\right)\right)
\end{equation}
where
\begin{equation}
\label{eq:regression-quantities-general}
R_t \defeq \sqrt{\frac{\nu_t}{\nu_{0}}}
\left(20
\left(
\log(320) + \frac{1}{t}
\log\left(\frac{\nu_{t} }{ \nu_0 }\right)
\right)
\right)^{t}
\text{ for all }
t \in [T]
\,.
\end{equation}
\end{theorem}

\begin{proof}
Let $\kappa_0 \defeq (10 \epsilon)^{-1}$, $\kappa_t \defeq 16 \cdot \frac{\nu_{t}}{\nu_{t-1}}$, and $\mm_t \defeq 4^t (\ma_t^\top \ma_t + \nu_t \mI)$ for all $t \in \{0 \cup [T]\}$. Note that each $\mm_t \approx_4 4^t (\ma_0^\top \ma_0 + \nu_t \mI)$. Since, the $\nu_i$ are non-negative and monotonically increasing, for all $t \in [T]$, 
\begin{align*}
&\mm_{t - 1}
    \preceq 4^{t - 1}  (\ma_{t - 1}^\top \ma_{t - 1} + \nu_{t} \mI)
    \preceq 4^{t}  (\ma_t^\top \ma_t + \nu_t \mI) 
    \preceq \mm_t \\
&\mm_{t}
    \preceq 4^{t} \cdot 4 \cdot (\ma_{t-1}^\top \ma_{t - 1} + \nu_{t} \mI)
    \preceq 4^{t - 1} \cdot (16 \nu_t/\nu_{t -1}) \cdot (\ma_{t - 1}^\top \ma_{t - 1} + \nu_{t - 1} \mI)
    = \kappa_{t} \mm_{t - 1}.
\end{align*}
Consequently, $C = (\mm_0, ((\mm_t, \kappa_t))_{t \in [T]}) \in \chain_{T}^d$. Additionally, \Cref{l:msp} implies that in $\otilde(\nnz(\ma_{T})) + O(n_T^\omega)$  time it is possible to develop a $\frac{1}{10 \kappa_T}$-solver, $\baseSolve$, for $\mm_T$ that can be applied in time
\[
O\left(\left(\nnz(\ma_T) + n_T^2\right)\log(\kappa_T \cdot (1 + \norm{\ma}^2 / \nu_T)\right)
    =
O\left(\left(\nnz(\ma_T) + n_T^2\right)\log((\nu_T + \norm{\ma}^2) / \nu_{T-1}\right).    
\]
Since each $\mm_t$ can be applied to a vector in $O(\nnz(\ma_t))$ time, applying \Cref{lem:recurse} with these settings of $C$, $\baseSolve$, and $\epsilon$ yields that there is an $\otilde(\nnz(\ma_{T})) + O(n_T^\omega)$ time algorithm that develops a $\runtime$-time $\epsilon$-solver for $\ma_0^\top \ma_0 + \nu_0 \I$ where for $K_t \defeq \prod_{i \in [t]} \lceil 4 \sqrt{\kappa_i} \log(20\kappa_{i-1}) \rceil$ for all $t \in [T]$, with
\[
\runtime = 
O\left(\left(\sum_{t\in[T]} K_t \cdot \nnz(\ma_{t -1}) + K_T \cdot \left(\nnz(\ma_T) + n_T^2 
    \log\left(\frac{\nu_t + \norm{\ma}^2}{\nu_{T - 1}}\right)\right)\right)
    \log(1/\epsilon)
    \right)
\,.
\]

To complete the proof it suffices to show that $K_t = O(R_t \log(1/\epsilon))$. To see this, note that each $\kappa_i \geq 1$ and $\log(20 \kappa_i) \geq 1$. Consequently, for all $t \in [T]$,
\[
K_t \leq \prod_{i \in [t]} 5 \sqrt{\kappa_i} \log(20 \kappa_{i-1})
= \sqrt{\frac{\nu_t}{\nu_0}} 20^t \log(1/\epsilon) \prod_{i \in [t- 1]} \log(320 \nu_t / \nu_{t - 1})\,.
\]
The result follows as $\prod_{i \in [t- 1]} \log(320 \nu_t / \nu_{t - 1}) \leq \prod_{i \in [t]} \log(320 \nu_t / \nu_{t - 1})$ and\footnote{Tighter bounds are obtainable. In multiple places in this paper we choose somewhat crude bounds on logarithmic factors in favor of simplifying the analysis and claimed bounds.}
\[
\left(\prod_{i \in [t]} \log\left(320 \frac{\nu_i}{\nu_{i-1}}\right)\right)^{1/t}
\leq 
\frac{1}{t} \sum_{i \in [t]} \log\left(320 \frac{\nu_i}{\nu_{i-1}}\right)
=
\log(320) + 
\frac{1}{t} \log\left(\frac{\nu_{t}}{\nu_0}\right)
\] 
by the arithmetic-mean geometric-mean inequality.
\end{proof}

Applying \Cref{thm:apx-gen} with geometrically decreasing $\nu_t$ suffices to obtain the following  \Cref{thm:regression:square} on solving particular linear systems of the form $\ma^\top \ma + \lambda \mI$. Applying \Cref{thm:regression:square} then almost immediately yields the proof of Theorem \ref{t:ls} via a simple reduction.

\begin{theorem}
\label{thm:regression:square}
There is an algorithm that given $\ma \in \R^{n \times d}$ with  $n = \otilde(d)$, $k \in [d]$, and $u \approx_2 \sigAvg_k (\ma)$ in $\otilde(d^2) + O(k^\omega)$ time, whp., produces an $\otilde(d^2 \log(\kappa(\ma)) \log(\epsilon^{-1}))$-time $\epsilon$-solver for $\ma^\top \ma + \frac{k}{d} u\I$.
\end{theorem}

\begin{proof}
Let $\nu = 10^6$, $T = \lceil \ln_{\nu}(\frac{d}{k}) \rceil$, $\nu_0 = u k / d$, and $\nu_t = \nu_0 \nu^{t}$ and $k_t \defeq \lceil d \nu^{1 - t}\rceil$ for all $t \in [T]$. Additionally, for all $t \in [T]$ let $\ma_t = \S_t \ma \in \R^{n_t \times d}$ where $\S_t$ is $(1,\sigAvg_{k_t})$-embedding for $\ma$ and $n_t = O(k_t)+\otilde(1)$. By \Cref{c:fast-embedding} these $\ma_t$ can be computed whp.\ in $\otilde(\nnz(\A) T) = \otilde(d^2)$ time (since $n = \otilde(d)$ and $T = \otilde(1)$). Note that, $k_t \geq 2 d \nu^{1 - \ln_v(d/k)} k \geq d$ for all $t \in [T]$ and therefore, 
\[
\nu_t = \frac{ku}{d} \nu^{t} \geq \frac{k}{2d} \nu^{t} \sigAvg_k(\ma)
\geq \frac{k_t}{2d} \nu^{t} \sigAvg_{k_t}(\ma) \geq \sigAvg_{k_t}(\ma)\,.
\]
Consequently, for $\ma_0 = \ma$ and all $t \in [T]$, 
$\ma_t^\top \ma_t + \nu_t \mI \approx_2 \ma^\top \ma + \nu_t \mI$ and $\ma_{t -1}^\top \ma_{t-1} + \nu_t \mI \approx_2 \ma^\top \ma + \nu_t \mI$ since $\nu_t \geq \nu_{t - 1}$. Therefore, $\ma_t^\top \ma_t + \nu_t \mI \approx_4 \ma_{t-1}^\top \ma_{t-1} + \nu_t \mI$ for all $t \in [T]$ and \Cref{thm:apx-gen}, applied to $((\A_t, \nu_t))_{t \in \{\{0\} \cup [T]\}}$, implies that there is an algorithm that in $\otilde(\nnz(\ma_T)) + O(n_T^\omega) = \otilde(d^2) + O(k^\omega)$ time develops a $\runtime$-time $\epsilon$-solver for $\ma_0^\top \ma_0 + \nu_0 \mI$ with $\runtime$ bounded as in \eqref{eq:regression-runtime-general} and \eqref{eq:regression-quantities-general}. 

For $R_t$ as defined in \eqref{eq:regression-quantities-general} and all $t \in [T]$,
\begin{align*}
k_t R_{t} &= \left\lceil d \nu^{1-t} \right\rceil \cdot \sqrt{\frac{\nu_{t}}{\nu_{0}}}
\left(20
\left(
\log(320) + \frac{1}{t}
\log\left(\frac{\nu_{t} }{ \nu_0 }\right)
\right)
\right)^{t} 
\leq
    2 d \nu \nu^{- t} \nu^{\frac{1}{2} t} 
\left(20\left(\log(320) + \log(\nu) \right)\right)^{t}
\,.
\end{align*}
Since $20 (\log(320) + \log(\nu)) \leq \sqrt{\nu}$ and $\nu$ is a constant, this implies that $k_t R_t = O(d)$ for all $t \in [T]$. Additionally (since $\nu$ is a constant), $\nnz(\ma_0) = O(d k_1)$, $\nnz(\ma_t) = O(d n_t) = \otilde(d k_t)$ for all $t \in [T]$, $k_t = O(k_{t+1})$ for all $t \in [T - 1]$. Combining these facts and that $k_t = O(d)$ for all $t \in [T]$ yields that $R_t \cdot \max\{\nnz(\ma_{t-1}),\nnz(\ma_t),n_t^2\} = \otilde(d^2)$ for all $t \in [T]$.  Since $\nu_{T - 1} \geq \nu_T / \nu = \Omega(u) = \Omega(\sigAvg_k(\ma))$, $\log\big(\frac{\nu_t + \norm{\ma}^2}{\nu_{T - 1}}\big) = O\big(\log\big(\frac{\|\ma\|^2}{\sigAvg_k(\ma)}\big)\big) = O(\log(d \kappa(\ma))) = \otilde(\log(\kappa(\ma)))$ and
the result follows by the bound on $\runtime$ in \eqref{eq:regression-runtime-general} and \eqref{eq:regression-quantities-general}.
\end{proof}

\begin{proof}[Proof of \Cref{t:ls}]
    First, we use a standard reduction from solving a tall least square regression task $\min_{\x\in\R^d}\|\A\x-\b\|^2$, where $\A\in\R^{n\times d}$ and $n\gg d$, to one that satisfies $n=\tilde O(d)$. This reduction involves one additional round of preconditioning. First, using Lemma \ref{l:reg-approx} with $k=d$ we construct a matrix $\tilde \A=\S\A$, for a sparse subspace embedding $\S$ of size $\tilde O(d)\times d$ in time $\tilde O(\nnz(\A))$, so that whp.\ $\tilde\A^\top\tilde\A\approx_2\A^\top\A$. Now, given $u\approx_2\sigAvg_k(\tilde\A)\approx_2\sigAvg_k(\A)$, we can solve the linear system $\A^\top\A\x=\A^\top\b$ by preconditioning it with $2(\tilde\A^\top\tilde\A + \frac kd u\I)$. Observe that:
    \begin{align*}
        \A^\top\A\preceq 2\tilde\A^\top\tilde\A\preceq 2\Big(\tilde\A^\top\tilde\A+\frac kd u\I\Big)
        \preceq 8\bigg(\A^\top\A + \frac{\sigAvg_k(\A)}{d\sigma_{d}^2(\A)}\A^\top\A\bigg)
        \preceq 16\bar\kappa_{k,2}^2\cdot \A^\top\A.
    \end{align*}
    Thus, according to Lemma \ref{lem:precon}, given a $\frac1{320\bar\kappa_{k,2}^2}$-solver $f$ for $\tilde\A^\top\tilde\A+\frac kd u\I$, we can construct an $\epsilon$-solver for $\A^\top\A$ that applies $f$ and $\A$ at most $\tilde O(\bar\kappa_{k,2}\log1/\epsilon)$ times. 
    We can now use Theorem \ref{thm:regression:square} to obtain the solver $f$ in $\tilde O (d^2) + O(k^\omega)$ time. The running time of the resulting solver $g$ is thus:
    \begin{align*}
        \runtime_g = \tilde O\Big((\nnz(\A) + \runtime_f)\bar\kappa_{k,2}\log1/\epsilon\Big) 
        = \tilde O\Big((\nnz(\A) + d^2)\bar\kappa_{k,2}\log1/\epsilon\Big).
    \end{align*}
    According to our assumptions, $\bar\kappa_{k,2} = O(1)$, which yields the desired running time. 
    Finally, observe that we can find $u\approx_2\sigAvg_k(\tilde\A)$ and $\tilde\kappa\approx \kappa(\A)$ via log-scale grid-search. We do this by first computing $M=\|\A\|_F^2$, then trying values $\tilde\kappa = 2,4,8,...$ (until it is larger than the polynomial assumed in \Cref{s:intro}), and for each of those values running the algorithm with all candidate values for $u$ from the set $\{M 2^{-t} \mid t\in[\lceil\log_2(d\tilde\kappa)\rceil]\}$.  We output the point $\x$ with the lowest value of $\|\A \x - \b\|^2=\|\A\x-\A\x^*\|^2+\|\A\x^*-\b\|^2$. This introduces an additional $O(\log(d\kappa(\A))^2)$ factor into the overall running time of the procedure.
\end{proof}

\subsection{Improved Conditioning Dependence}
\label{s:warm-up-improved}

Here we show how to apply \Cref{thm:apx-gen} with a different choice of $\nu_t$ and associated $\ma_t \in \R^{n_t \times d}$ to obtain nearly linear time algorithms for $\ma$ with weaker assumptions on the conditioning of $\ma$. Rather than $\nu_t$ increase geometrically and $n_t$ decrease geometrically we choose $n_t = \otilde(\max\{d \exp(-\alpha t^2),2k\})$ for an appropriate choice of $\alpha$ and have $\nu_t \approx \sigAvg_{k_t}(\ma)$. In short, we decrease the $n_t$ at faster rate and pick $\nu_t$ to be nearly the smallest value our analysis allows (rather than changing it at a fixed schedule). 

Provided that the $\sigma_i(\ma)$ are known approximately (so that, e.g., the appropriate $\sigAvg_{k_t}(\ma)$ and $\kappa(\ma)$ can be approximated) we show (see \Cref{t:ls-optimized}) that this allows us to obtain algorithms with similar performance as in \Cref{sec:primal:regression} except with their dependence on $\bar{\kappa}_{k,2}$ to be nearly replaced with a dependence on  $\bar{\kappa}_{k,1+\eta}$ for any constant $\eta \in (0, 1)$. In \Cref{s:primal-dual} we apply a similar scheme for choosing the $\nu_t$ and $n_t$ but a different recursion and additional techniques to both obtain a further improvement to the conditioning dependence and remove the need for approximations to the $\sigAvg_{k_t}(\ma)$ to be given. Nevertheless, we provide the analysis in this section as we believe it serves as helpful warm-up to that in \Cref{s:primal-dual} and because the key technical lemma which we use to analyze the algorithm (\Cref{l:running-time-general}) is also applied in \Cref{s:primal-dual} and \Cref{s:psd} (albeit with different parameters). 

We start by giving the following \Cref{thm:regression:optimized}, which analyzes the described scheme using \Cref{thm:apx-gen} (our theorem regarding recursive regularized preconditioning) provided particular bounds on $\alpha$ hold.

\begin{theorem}\label{thm:regression:optimized}
For $\ma \in \R^{n \times d}$ where $n = \otilde(d)$ and $d \geq 1$, $k \in [d]$ with $2k \leq d$, and $\alpha > 1$  define 
$k_t \defeq \max\{\lceil d \cdot \exp(-\alpha t^2)\rceil, 2k\}$ for all 
$t \in \{0\} \cup [T]$ where $T \defeq \lceil \sqrt{\ln(d/(2k)) / \alpha}\rceil$.
Additionally, let $0 < \nu_0 \leq \nu_1 \leq \cdots \leq \nu_T$ satisfy $\nu_t \approx_4 \sigAvg_{k_t}(\ma)$ for all $t \in [T]$ and $\nu_0 \approx_4 \sigma_d(\ma)^2$ and $\log(20 \log(320 \nu_t / \nu_0)) \leq \alpha$.
There is an algorithm that given $\ma$, $k$, $\alpha$, the $k_t$, and the $\nu_t$, in $\otilde(d^2) + O(k^\omega)$ time, whp., produces a $\runtime$-time $\epsilon$-solver for $\ma^\top \ma + \nu_0 \mI$ with
\begin{equation}
\label{eq:regression_warmup_runtime_unsimplified}
\runtime = \otilde\left( 
\max_{t \in [T]}  d k_{t-1} 
\sqrt{\frac{\sigAvg_{k_t}(\A)}{\sigma_d(\ma)^2}} \exp(\alpha t)
    \ln\left(\kappa(\ma)\right)
    \ln(\epsilon^{-1})\right)\,.
\end{equation}
\end{theorem}

\begin{proof}
For all $t \in [T]$ let $\ma_t = \S_t \ma \in \R^{n_t \times d}$ where $\S_t$ is $(1,\sigAvg_{4 k_t}(\ma))$-embedding for $\ma$ and $n_t = O(k_t)+\otilde(1)$. By \Cref{c:fast-embedding} these $\ma_t$ can be computed whp.\ in $\otilde(\nnz(\A) T) = \otilde(d^2)$ time (since $n = \otilde(d)$ and $T = \otilde(1)$). Note that $\sigAvg_{4k_t}(\ma) \leq \frac{1}{4} \sigAvg_{k_t}(\ma) \leq \nu_t$ and consequently, for $\ma_0 = \ma$ and all $t \in [T]$, 
$\ma_t^\top \ma_t + \nu_t \mI \approx_2 \ma^\top \ma + \nu_t \mI$ and $\ma_{t -1}^\top \ma_{t-1} + \nu_t \mI \approx_2 \ma^\top \ma + \nu_t \mI$ since $\nu_t \geq \nu_{t - 1}$. Therefore, $\ma_t^\top \ma_t + \nu_t \mI \approx_4 \ma_{t-1}^\top \ma_{t-1} + \nu_t \mI$ for all $t \in [T]$ and \Cref{thm:apx-gen}, applied to $((\A_t, \nu_t))_{t \in \{0\} \cup [T]}$, implies that in $\otilde(\nnz(\ma_T) + n_T^\omega) = \otilde(d^2 + k^\omega)$ time it is possible to develop a $\runtime$-time $\epsilon$-solver for $\ma^\top \ma + \nu_0 \mI$ with $\runtime$ bounded as in \eqref{eq:regression-runtime-general} and \eqref{eq:regression-quantities-general}. 

Note that $\nu_t \leq 4 \sigAvg_{k_t}(\ma)$ for all $t \in [T]$, $\nu_0 \geq \frac{1}{4} \sigma_d(\ma)$, and $\nu_t$ increases monotonically with $t$. Therefore, 
\begin{align}
R_{t} &= 
    \sqrt{\frac{\nu_t}{\nu_{0}}}
\left(20\left(\log(320) + \frac{1}{t}
    \log\left(\frac{\nu_{t} }{ \nu_0 }\right)
    \right)\right)^t
\leq 
4
\sqrt{\frac{\sigAvg_{k_{t}}(\ma)}{\sigma_d(\ma)^2}} \exp( \alpha t)
\text{ for all }
t \in [T]
\,.
\end{align}
Therefore, since $\nnz(\ma_0) = \otilde(d^2) = \otilde(d k_0)$, $\nnz(\ma_t) = \otilde(d k_t)$ and $n_t^2 = \otilde(k_t^2) = \otilde(d k_t)$ for all $t \in [T]$, 
\[
R_t \cdot \max\{\nnz(\ma_{t -1}), \nnz(\ma_{t}), n_t^2\} 
 =   O\left( \sqrt{\frac{\sigAvg_{k_{t}}(\ma)}{\sigma_d(\ma)^2}} \exp( \alpha t)
    \right)
    \text{ for all }
    t \in [T]\,.
\]
The result follows by using $\nu_t \leq 4 \sigAvg_{k_{T}}(\ma) \leq 4 \norm{\ma}^2$ and $\nu_{T-1} \geq \nu_0 \geq \frac{1}{4}\sigma_d(\ma)^2$ to bound $\log((\nu_t + \norm{\ma}^2)/\nu_{T-1})$ and using that $T = \otilde(1)$ (to obtain $\max_{t \in [T]}$ in the bound on $\runtime$ rather than $\sum_{t \in [T]}$).\footnote{Alternatively, it is possible to bound the sum directly by opening up the proof of \Cref{l:running-time-general} and bounding the sum of the $\exp(V_t)$ terms in that proof. We instead replace with the max to simplify the proof as its cost is hidden in $\otilde(\cdot)$.} 
\end{proof}

To analyze the running time in \Cref{thm:regression:optimized} we give a straightforward mathematical fact (\Cref{l:generalized-mean}) and use it to prove a key technical lemma, \Cref{l:running-time-general}, which bounds the $k_{t-1} (\sigAvg_{k_{t-1}} / \sigma_d(\ma)^2)^{1/2}) \exp(\alpha t)$ term in the running time (which in turn essentially corresponds to the cost of the $t$-th round of recursion in our regression algorithm). Though we only need the lemma with $p = 1$ for this purpose, we provide it more generally as in \Cref{s:primal-dual} and \Cref{s:psd}  we use it with $p = 2$ to bound running times similarly.

\begin{lemma}\label{l:generalized-mean}
    Consider $\lambda_1\geq\lambda_2\geq...\geq\lambda_d>0$, $k\in[d]$, and $c\in(0,1)$. Then,
    \begin{align*}
        \frac1k\sum_{i=k+1}^d\lambda_i \leq  \Big(\frac1k\sum_{i=1}^d\lambda_i^{c}\Big)^{1/c}.
    \end{align*}
\end{lemma}

\begin{proof}
    First, using monotonicity of $\lambda_i$'s and that generalized means upper-bounds the minimum,
    \begin{align*}
        \lambda_k = \min_{i\in[k]}\lambda_i \leq \Big(\frac1k\sum_{i \in [k]}\lambda_i^c\Big)^{1/c}.
    \end{align*}
    Using the above, we obtain the claim as follows:
    \begin{align*}
        \sum_{i>k}\lambda_i \leq \sum_{i>k}\lambda_i^c\lambda_k^{1-c} \leq        \Big(\sum_{i>k}\lambda_i^c\Big)\Big(\frac1k\sum_{i\in[k]}\lambda_i^c\Big)^{\frac{1-c}c}
        \leq \Big(\sum_{i\in[n]}\lambda_i^c\Big)\Big(\frac1k\sum_{i\in[n]}\lambda_i^c\Big)^{\frac1c - 1}\leq k\cdot 
        \Big(\frac1k\sum_{i \in [n]}\lambda_i^c\Big)^{\frac1c}.
    \end{align*}
\end{proof}

\begin{lemma}\label{l:running-time-general}
        For $\A\in\R^{n\times d}$, $k\in [d]$ with $2k \leq d$, and $\alpha\geq 1$, define  $k_t \defeq \max\{\lceil d\exp(-\alpha t^2)\rceil,2k\}$ for all $t \in \{0\} \cup [T]$ where $T \defeq \lceil\sqrt{\log(d/(2k))/\alpha}\rceil$. For all $p \in [1,100]$,\footnote{The constant of $100$ is arbitrary and can be replaced with any real number larger than $1$.} $\eta \in (0,1)$, and $t \in [T]$,
        \begin{align*}
            k_{t-1}^p\sqrt{\frac{\sigAvg_{k_t}(\A)}{\sigma_d^2(\A)}}\exp(\alpha t) = O\left(d^2\bar\kappa_{k,p^{-1}+\eta}(\A) \exp\left(\alpha\left(p (c_p^2 - 1) + \frac{c_p^2}{\eta}\right)\right)\right)
        \text{ where }
        c_p \defeq 1 + \frac{1}{2p}\,.\footnote{\sidford{I didn't bound the sum since there are different ways to do it and unsure if better to just pay the $\sqrt{\log}$ or get $1/\eta$ factors (which I think is actually improvable $1/\sqrt{\eta}$) - happy to discuss}}
        \end{align*}
    \end{lemma}

\begin{proof}
Fix a $t \in [T]$. Note that $k_t \geq 2k$ and therefore, $k_t - k \geq \frac{1}{2} k_t$. Consequently, applying \Cref{l:generalized-mean} with $\lambda_1,\ldots,\lambda_{d - k} = \frac{\sigma_{k +1}(\ma)^2}{\sigma_d(\ma)^2},\ldots,\frac{\sigma_{d}(\ma)^2}{\sigma_d(\ma)^2}$ and $c = q/2$ for $q \defeq p^{-1} + \eta$, yields,  
\[
\frac{1}{k_t - k} \sum_{i > k_t} \frac{\sigma_i(\ma)^2}{\sigma_d(\ma)^2}
\leq \left(\frac{1}{k_t  - k} \sum_{i > k} \frac{\sigma_i(\ma)^q}{\sigma_d(\A)^q} \right)^{\frac{2}{q}}
= \left(\frac{d - k}{k_t - k}\right)^{\frac{2}{q}} \bar{\kappa}_{k,q}(\ma)^2
\leq \left(\frac{2d}{k_t}\right)^{\frac{2}{q}} \bar{\kappa}_{k,q}(\ma)^2\,.
\]
Since $\frac{\sigAvg_{k_t}(\ma)}{\sigma_d(\ma)^2} = \frac{1}{k_t} \sum_{i > k_t} \frac{\sigma_i(\ma)^2}{\sigma_d(\ma)^2} \leq \frac{1}{k_t - k} \sum_{i > k_t} \frac{\sigma_i(\ma)^2}{\sigma_d(\ma)^2}$ and $\frac{1}{q} = \frac{p}{1 + \eta p}$,
\[
    \sqrt{\frac{\sigAvg_{k_t}(\ma)}{\sigma_d(\ma)^2}}
\leq 
    \left(\frac{2d}{k_t}\right)^{\frac{1}{q}} \bar{\kappa}_{k,q}(\ma)
= 
O\left(
    \left(\frac{d}{k_{t-1}}\right)^{\frac{1}{q}} 
    \exp\left(\frac{\alpha}{q} (2t - 1) \right)
    \bar{\kappa}_{k,q}(\ma)
\right)
\]
where we used that $k_{t-1} = O(k_t \exp(\alpha(t^2 - (t - 1)^2)) = O(k_t \exp(\alpha(2 t -  1))$ and $\frac{1}{q} = O(1)$. Since $k_{t-1} \geq d \cdot \exp(-\alpha (t-1)^2)$ and $p \geq q^{-1}$, the result follows from the following calculation,
\begin{align*}
k_{t-1}^p\sqrt{\frac{\sigAvg_{k_t}(\A)}{\sigma_d^2(\A)}}\exp(\alpha t)
&= O\left( 
    d^p
    \left(\frac{k_{t-1}}{d}\right)^{p - \frac{1}{q}}
    \exp\left(\alpha\left(t + \frac{1}{q} (2t - 1)\right)\right)
    \bar{\kappa}_{k,q}(\ma)
\right)\\
&= O\left( 
    d^p
    \exp\left(\alpha\left(
    - \left(p - \frac{1}{q}\right) (t - 1)^2 + t + \frac{1}{q} (2t - 1)\right)\right)
    \bar{\kappa}_{k,q}(\ma)
\right)\\
&= O\left( 
    d^p
    \bar{\kappa}_{k,q}(\ma)
    \cdot \exp\left( V_t \right)\right)
\,.
\end{align*}
where (using that $-a t^2 + bt = -a (t - b/(2a))^2 + b^2/(4a) \leq b^2 / (4a)$ and $p - \frac{1}{q} = \frac{\eta p^2}{1 + \eta p}$)
\begin{align*}
V_t &\defeq - \left(p - \frac{1}{q}\right)(t - 1)^2 + 
\left(1 + \frac{2}{q} \right) t - \frac{1}{q}\\
&=- \frac{\eta p^2}{1 + \eta p} t^2 + 
\left(1 + 2 p \right) t - p
\leq \frac{(1 + \eta p)(1 + 2 p)^2}{4 \eta p^2}
- p
= (\eta^{-1} + p) \cdot c_p^2 - p  
\,.
\end{align*}
\end{proof}

Combining \Cref{thm:regression:optimized} with \Cref{l:running-time-general} (along with a similar approach as that in the proof of \Cref{t:ls}) proves our main theorem of this subsection regarding improved running times for solving regression under fine-grained assumptions on the condition number of $\A$.

\begin{theorem}\label{t:ls-optimized-warmup}
    There is an algorithm that given $\A\in\R^{n\times d}$, $\b\in\R^d$, $\epsilon\in(0,1)$, $k \in [d]$, and $\tilde{\sigma} \in \R^d$ where $\tilde{\sigma}_i^2 \approx_2 \sigma_i(\ma)^2$ for all $i \in [d]$, outputs $\xout\in\R^d$ that whp.\ satisfies $\|\A\xout-\A\xopt\|\leq \epsilon\|\A\xopt\|$, where $\x^* = \argmin_{\x\in\R^d}\|\A\x-\b\|^2$, in time
    \begin{align*}
        \tilde O\bigg(\Big(\nnz(\A) + \min_{\eta \in (0,1)}
    d^2\bar\kappa_{k,1+\eta}(\A) 
    \ln^{\frac{9}{4}(1+\eta^{-1})}(d \kappa(\ma))
    \Big)\log(1/\epsilon)\bigg) + O(k^\omega).
    \end{align*}        
\end{theorem}

\begin{proof}[Proof of Theorem \ref{t:ls-optimized-warmup}]
    First, note that following the same argument as in the proof of Theorem \ref{t:ls}, we can reduce the task of solving for $\A^\top\A$ to $\epsilon$-accuracy to $\tilde O(\sqrt{\tilde{\sigma}_d/\sigma_{d}^2(\A)}\log1/\epsilon) = O(\log(1/\epsilon))$ times solving for $\tilde\A^\top\tilde\A+\tilde{\sigma}_d^2 \I$ to constant accuracy, where $\tilde\A\in\R^{(O(d)+\tilde O(1))\times d}$ and $\tilde\A^\top\tilde\A\approx_2\A^\top\A$.
    
    If $d = O(1)$ or $k> d/2$ then we can obtain the result by  directly computing the inverse of $\tilde\A^\top\tilde\A+\tilde{\sigma}_d^2 \I$ in $\otilde(d)+ O(d^\omega)$ time and then solving $\tilde\A^\top\tilde\A+\tilde{\sigma}_d^2 \I$ when required in $O(d^2)$ time. Otherwise, we invoke \Cref{thm:regression:optimized} to solve $\tilde\A^\top\tilde\A+\tilde{\sigma}_d^2 \I$ to constant accuracy. 
    In this invocation of \Cref{thm:regression:optimized}, we pick $\alpha = \log(20\log(320 \cdot 256 d\cdot \tilde{\sigma}_1^2/\tilde{\sigma}_d^2))$, $\nu_0 = \sigma_d^2$, and $\nu_t = \max\{\nu_{t-1}, \frac{1}{k_t} \sum_{i > k_t} \tilde{\sigma}_i^2\}$ for all $t \in [T]$ for $k_t$ and $T$ as defined in \Cref{thm:regression:optimized}. Since $\tilde\A^\top\tilde\A\approx_2\A^\top\A$, each $\sigma_i(\ma)^2 \approx_2 \sigma_i(\tilde\A)^2$ and $\tilde{\sigma}_i^2 \approx_4 \sigma_i(\tilde\A)^2$. 
    Therefore, $\frac{1}{k_t} \sum_{i > k_t} \tilde{\sigma}_i^2 \approx_4 \sigAvg_{k_t}(\tilde\A)$ for all $t \in [T]$. Additionally, note that $\sigAvg_{k_1}(\tilde\A) \geq \frac{d - k_1}{k_1} \sigma_d(\ma)^2 \geq \sigma_d(\ma)^2$ since $k_1 \leq \frac{d}{2}$ (as $d = O(1)$ and $k > d/2$ do not hold). Consequently, the $\nu_t$ increase monotonically, with $\nu_0 \approx_4 \sigma_d(\tilde\A)^2$ and $\nu_t \approx_4 \sigAvg_{k_t}(\tilde\A)$ for each $t \in [T]$ (since if $a_1 \approx_4 b_1$ and $a_2 \approx_4 b_2$ for positive $a_1,a_2,b_1,b_2$ with $b_1 \leq b_2$ then $\max\{a_1, a_2\} \approx_4 b_2$). 
    Finally, note that $2k \leq d$ and $\log(20 \log(320 \nu_t / \nu_0)) \leq \alpha$ as\footnote{The factor of $d$ in $O(d \kappa(\ma))$ below arises since $\sigAvg_k(\tilde\ma) = O(d \sigma_1(\tilde\ma))$. However, when $v_t \geq \sigma_1(\tilde\ma)$ it is possible to solve $\tilde{\ma}^\top \tilde{\ma} + \nu_t \mI$ to accuracy $\epsilon$ in $O(\nnz(\ma)\log(1/\epsilon))$ directly just by applying simpler iterative methods, e.g., Richardson iteration \cite{golub1961chebyshev}.
    Consequently, more careful termination of the recursion could potentially remove this $d$ factor in the logs in this algorithm and others in the paper. However, we write bounds with this $d$ factor in logarithmic terms for the sake of simplicity.}
    \[
    \frac{1}{16} \nu_t / \nu_0 \leq \nu_T / \nu_0 \leq \sigAvg_{2k}(\tilde\A)/\sigma_d(\tilde\A)^2 = \frac{1}{2k} \sum_{i > 2k} \frac{\sigma_i(\tilde\A)^2}{\sigma_d(\tilde\A)^2}
    \leq \frac{d - k}{2k} \kappa(\tilde\A) \leq 16 d \frac{\tilde{\sigma}_1(\ma)^2}{\tilde{\sigma}_d(\ma)^2}
    = O(d \kappa(\ma)^2)
    \,.
    \]
    Consequently, \Cref{thm:regression:optimized} when combined with \Cref{l:running-time-general} implies that after spending $\otilde(d^2) + O(k^\omega)$ time, $\tilde\A^\top\tilde\A+\tilde{\sigma}_d^2 \I$ can be solved to constant accuracy in time
    \begin{align*}
    \runtime
    &=  
    O\left(
    \min_{\eta \in (0,1)}
    d^2\bar\kappa_{k,1+\eta}(\A) \exp\left(\alpha\left(((3/2)^2 - 1) + \frac{(3/2)^2}{\eta}\right)\right)
    \ln(\kappa(\tilde{\ma}))
    \right)\\
    &=
    O\left(
    \min_{\eta \in (0,1)}
    d^2\bar\kappa_{k,1+\eta}(\A) 
    \ln^{\frac{9}{4}(1+\eta^{-1})}(d \kappa(\ma))
    \right)
    \end{align*}
    where we used that $\exp(\alpha) = O(\log(d \kappa(\tilde{\ma}))) = O(\log(d \kappa(\ma)))$.
\end{proof}

    It is straightforward to lessen the assumption of needing to know $\sigma_i(\ma)^2$ approximately in \Cref{t:ls-optimized-warmup} at the cost of an extra sub-polynomial factor in the running time. To see how this is possible, first note that careful inspection of the proof of \Cref{t:ls-optimized-warmup} shows that the  algorithm in \Cref{t:ls-optimized-warmup} only requires $\tilde{\sigma}_1^2 \approx_2 \sigma_1(\ma)^2$, $\tilde{\sigma}_d^2 \approx_2 \sigma_d(\ma)^2$, and $\nu_t \approx_4 \sigAvg_{k_t}(\tilde\A)^2$ for all $t \in [T]$. Furthermore, $\|\tilde\A\|_F^2$ is easily computable in $O(\nnz(\ma))$ time and satisfies $\frac{1}{d}\|\A\|_F^2 \leq \sigma_1(\ma)^2 \leq \|\A\|_F^2$. Consequently, given any $\ell \in (0, \tilde{\sigma_d})$ one could consider the  $O(\log(\norm{\ma}_F/\ell))$ different possible values of $\tilde{\sigma}_1^2$ and $\tilde{\sigma}_d^2$ in the set $S = \{M 2^{-t} | t \in [\lceil \log_2(M/\ell)\rceil]\}$ and then for the induced $T$ consider the different possible values of $\sigma_d(\ma)^2$ in $S$. For each of these possible values, we could then run the algorithm in \Cref{t:ls-optimized-warmup} in parallel and output the point that has the lowest value within the desired time budget. This would involve considering at most $\log(\norm{\ma}_F^2 / \ell)^{O(\sqrt{\log(\norm{\ma}_F^2 / \ell)})} = \exp(O(\log(d \norm{\ma}^2 / \ell) \log \log (d \norm{\ma}^2 / \ell)))$ different executions of the algorithm.

    In the next section we show how to improve upon this algorithm both in terms of how it handles unknown parameters and its dependence on the average conditioning of $\ma$.

\section{Alternating Primal-Dual Regression Solver}
\label{s:primal-dual}
In this section, we present our primal-dual preconditioning scheme for regression, which allows us to further improve the conditioning dependence in our regression algorithms. We also describe an inductive procedure for finding the hyper-parameters in this algorithm without introducing sub-polynomial factors. With these tools, we prove \Cref{t:ls-optimized}. 

\subsection{Reduction Tools}
We first provide a lemma, which encapsulates the idea of dual preconditioning into a black-box reduction. The result assumes an approximation guarantee between matrices $\A\A^\top+\nu\I$ and $\B\B^\top+\nu\I$ (dual case), as opposed to between matrices $\A^\top\A+\nu\I$ and $\B^\top\B+\nu\I$ (primal case, considered in \Cref{sec:primal:regression}), while still using and constructing solvers for the latter. This is done to allow the second dimension $\A$ and $\B$ to differ.
\begin{lemma}\label{l:primal-dual-recursion}
    Consider $\A\in\R^{n\times d}$, $\B\in\R^{n\times s}$, and $\nu>0$. Let $\kappa=1+\|\A\|^2/\nu$ and suppose that
    \begin{align*}
        \A\A^\top + \nu\I \approx_4 \B\B^\top+\nu\I.
    \end{align*}
    Then, given a $\frac1{160\kappa^2}$-solver $f$ for $\A^\top\A+\nu\I$, we can construct an $\epsilon$-solver for $\B^\top\B+\nu\I$ that applies $f$  at most $16\log(8\kappa/\epsilon)$ times and spends additional 
    $O((n+ \runtime_{\A} + \runtime_{\B})\log(\kappa/\epsilon))$ time.
    \end{lemma}
    \begin{proof}
        We start by applying the Woodbury formula to $\A^\top\A+\nu\I$ and to $\tilde\A\tilde\A^\top+\nu\I$, obtaining that:
        \begin{align}
            (\B^\top\B+\nu\I)^{-1} 
            &= \frac1\lambda\big(\I - \B^\top(\B\B^\top+\nu\I)^{-1}\B\big),\label{eq:woodbury1}
            \\
            (\A\A^\top+\nu\I)^{-1}
            &= \frac1\lambda\big(\I - \A(\A^\top\A+\nu\I)^{-1}\A^\top\big).\label{eq:woodbury2}
        \end{align}
        To solve for $\B^\top\B+\nu\I$ using \eqref{eq:woodbury1}, we construct a solver for $\B\B^\top+\nu\I$ using $\A\A^\top+\nu\I$ as a preconditioner. To do this, we use \Cref{l:woodbury} to produce an $\epsilon_g$-solver $g$ for $\A\A^\top+\nu\I$ from the $\frac{\epsilon_g\nu^2}{\|\A\|^4}$-solver $f$ (which applies $f$, $\A$, and $\A^\top$ once 
        each). Then, using Lemma \ref{lem:precon} and that $\A^\top\A^\top+\nu\I\approx_4\B\B^\top+\nu\I$, we use a $\frac1{160}$-solver $g$, to obtain an $\epsilon_h$-solver $h$ for $\B\B^\top+\nu\I$ that applies $g$, $\B$, and $\B^\top$ at most $\lceil16\log(2/\epsilon_h)\rceil$ times each. Finally, applying a $\frac{\epsilon\nu^2}{16\|\A\|^4}$-solver $h$ along with $\B$ and $\B^\top$ once each, via Lemma \ref{l:woodbury} and \eqref{eq:woodbury1} we obtain a solver for $\B^\top\B+\nu\I$.
        Putting this together, with $\epsilon_g=\frac1{160}$ and $\epsilon_h=\frac{\epsilon\nu^2}{16\|\A\|^4}$, we see that the number of times that $f$ is applied is: 
        $\lceil16\log(2/\epsilon_h)\rceil \leq \lceil16\log(32\kappa^2/\epsilon)\rceil\leq 32\log(2\kappa/\epsilon)$.
    \end{proof}
Next, we show an auxiliary result, which ensures that the spectral sums $\sigAvg_k$ associated with the matrix $\A$ do not get significantly distorted after a right-hand-side application of a subspace embedding matrix.
\begin{lemma}\label{l:sigavg}
    Suppose that random matrix $\mPi\in\R^{s\times d}$ has $(\epsilon,\delta,1,\ell)$-OSE moments. Then, for any $n\geq d\geq k\geq 1$ and matrix $\A\in\R^{n\times d}$, with probability $1-\delta$, we have $\sigAvg_k(\A\mPi^\top)\leq(1+\epsilon)\sigAvg_k(\A)$.
\end{lemma}
\begin{proof}
    For a matrix $\B$, let $[\B]_k$ denote the best rank $k$ approximation of $\B$ (i.e., the top-$k$ part of its singular value decomposition). Then,
    \begin{align*}
        \sigAvg_k(\A\mPi^\top) 
        &= \frac1k\|\A\mPi^\top - [\A\mPi^\top]_k\|_F^2
        \leq \frac1k\|\A\mPi^\top-[\A]_k\mPi^\top\|_F^2 = \frac1k\|\B\mPi^\top\|_F^2,
    \end{align*}
    where $\B = \A-[\A]_k$. Let $X_i = \frac{\|\mPi\b_i\|^2}{\|\b_i\|^2} - 1$, where $\b_i$ is the $i$th rows of $\B$. Then, using 
    $\|X\|_\ell\defeq(\E|X|^\ell)^{1/\ell}$,
    \begin{align*}
        \Big\|\|\B\mPi^\top\|_F^2 - \|\B\|_F^2\Big\|_{\ell} 
        &= \bigg\|\sum_{i \in [n]}\Big(\|\mPi\b_i\|^2 - \|\b_i\|^2\Big)\bigg\|_{\ell}
        = \bigg\|\sum_{i \in [n]}\|\b_i\|^2 X_i\bigg\|_{\ell}
        \\
        &\leq \sum_{i \in [n]}\|\b_i\|^2\|X_i\|_{\ell}
        \leq \epsilon\delta^{1/\ell}\sum_{i \in [n]}\|\b_i\|^2 = \epsilon\delta^{1/\ell}\|\B\|_F^2.
    \end{align*}
    Markov's inequality now implies that $\|\B\mPi^\top\|_F^2-\|\B\|_F^2 \leq \epsilon\|\B\|_F^2$ with probability $1-\delta$. 
    This concludes the proof since $\sigAvg_k(\A\mPi^\top)=\frac1k\|\B\mPi^\top\|_F^2\leq (1+\epsilon)\frac1k\|\B\|_F^2=(1+\epsilon)\sigAvg_k(\A)$.
\end{proof}

\subsection{Primal-Dual Recursive Preconditioning}

We now define the primal-dual regularized preconditioning chain (an analog of \Cref{def:primalchain}). The chain includes two sequences of matrices, $\A_t$ and $\B_t$, so that the recursion can alternate between them, reducing the first dimension going from $\A_{t-1}$ to $\B_t$, and then reducing the second dimension going from $\B_t$ to $\A_t$.
\begin{definition}[Primal-Dual Regularized Preconditioning Chain]
\label{def:primaldualchain}
We call $D = ((\A_0,\nu_0),(\A_t,\B_t,\nu_t)_{t \in [T]})$ a \emph{length $T$ primal-dual regularized preconditioning chain} (denoted $D \in \primaldualchain_T$) if $0 < \nu_0 \leq \cdots \leq \nu_T$, and the matrices $\A_t$, $\B_t$ (along with their dimensions) satisfy:
     \begin{align*}
         \B_{t}^\top\B_{t}+\nu_t\I \ \approx_4\ \A_{t-1}^\top\A_{t-1}+\nu_t\I \quad\text{and}\quad
         \A_t\A_t^\top+\nu_t\I \ \approx_4\ \B_t\B_t^\top+\nu_t\I\quad\text{for all } t\in[T].
     \end{align*}
\end{definition}
The following lemma is an analog of \Cref{thm:apx-gen} that provides a general recursive preconditioning framework for working with primal-dual chains.
\begin{lemma}\label{l:primal-dual-general}
     Consider $D = ((\A_0,\nu_0),(\A_t,\B_t,\nu_t)_{t \in [T]}) \in \primaldualchain_T$, and
     let $\code{bsolve}$ be a solver for $\A_T^\top\A_T+\nu_T\I$. Given $\epsilon\in(0,1)$, there is an $\epsilon$-solver $\code{solve}$ for $\A_0^\top\A_0+\nu_0\I$ with
     \begin{align*}
         \runtime_{\code{solve}} &= O\bigg(K_T\runtime_{\code{bsolve}} + \sum_{t\in[T]}K_t(\runtime_{\A_t}+\runtime_{\B_t}+\runtime_{\A_{t-1}})\bigg),
         \\
             K_t& = 10^{3t}\sqrt{\nu_t/\nu_0}\log^{2t}(20\kappa)\log(2/\epsilon),\quad\text{ and }\quad\kappa =\max_{t\in[T]}\,\big(\|\B_t\|^2+\nu_{t-1}\big)/\nu_t.
     \end{align*}
\end{lemma}
\begin{proof}
    First, observe that for $t\in[T]$, since $\nu_{t-1}\leq\nu_t$, we have:
    \begin{align*}
        \A_{t-1}^\top\A_{t-1}+\nu_{t-1}\I\preceq 4(\B_t^\top\B_t+\nu_t\I)\preceq \frac{16\nu_t}{\nu_{t-1}}(\A_{t-1}^\top\A_{t-1}+\nu_{t-1}\I),
    \end{align*}
    which means that we can use a $\frac{\nu_{t-1}}{160\nu_t}$-solver $g_t$ for $\B_t^\top\B_t+\nu_t\I$ to obtain an $\epsilon_{f_{t-1}}$-solver $f_{t-1}$ for $\A_{t-1}^\top\A_{t-1}+\nu_{t-1}\I$, which applies $g_t$ at most $\lceil16\sqrt{\nu_t/\nu_{t-1}}\log(2/\epsilon_{f_{t-1}})\rceil$ times (\Cref{lem:precon}). Next, we apply Lemma \ref{l:primal-dual-recursion} to show that we can use a $\frac1{160\kappa^2}$-solver $f_t$ for $\A_t^\top\A_t+\nu_t$ to construct an $\epsilon_{g_t}$-solver for $g_t$ for $\B_t^\top\B_t+\nu_t\I$ that applies $f_t$ at most $32\log(2\kappa/\epsilon_{g_t})$ times and spends additional $O((\runtime_{\A_t}+\runtime_{\B_t})\log(\kappa/\epsilon_{g_t}))$ time. Putting this together, we can obtain a $\frac1{160\kappa^2}$-solver $f_{t-1}$ for $\A_{t-1}^\top\A_{t-1}+\nu_{t-1}\I$, with 
    \begin{align*}
        \runtime_{f_{t-1}} 
        &\leq \lceil16\sqrt{\nu_t/\nu_{t-1}}\log(320\kappa^2)\rceil\Big(\runtime_{g_t} +O(\runtime_{\B_t}+\runtime_{\A_{t-1}})\Big)
        \\
        &\leq \lceil16\sqrt{\nu_t/\nu_{t-1}}\log(320\kappa^2)\rceil32\log(320\kappa^2)\Big(\runtime_{f_t}+O(\runtime_{\A_t}+\runtime_{\B_t}+\runtime_{\A_{t-1}})\Big)
        \\
        &\leq 20^3\sqrt{\nu_t/\nu_{t-1}}\log^2(20\kappa)\Big(\runtime_{f_t}+ 
        O(\runtime_{\A_t}+\runtime_{\B_t}+\runtime_{\A_{t-1}})\Big),
    \end{align*}
    In the base case, we set $\epsilon_{f_0} = \epsilon$, which leads to an additional factor of $\log(2/\epsilon)$. By induction, 
    \begin{align*}
        \runtime_{f_0} = O\Big( K_T\runtime_{f_T} + \sum_{i \in [T]} K_t(\runtime_{\A_t}+\runtime_{\B_t}+\runtime_{\A_{t-1}})\Big).
    \end{align*}
    Letting $f_0=\code{solve}$ and $f_T=\code{bsolve}$, we conclude the proof.
\end{proof}

\subsection{Regression Algorithm}

We are ready to present the main result of this section, which gives a regression solver that runs in nearly linear time for any $(k,p)$-well-conditioned system with $p>1/2$, thereby recovering Theorem \ref{t:ls-optimized} as a corollary.
\begin{theorem}\label{t:primal-dual-optimized}
    There is an algorithm that given $\A\in\R^{n\times d}$, $\tilde\kappa\approx_2\kappa(\A)$, $k\in[d]$, $\tilde\lambda\approx_2\sigma_d^2(\A)$, $\tilde\Lambda\approx_2 \sigAvg_{2k}(\A)$, in $\tilde O(\nnz(\A)+ d^2)+O(k^\omega)$ time whp.\ produces a $\runtime$-time $\epsilon$-solver for $\A^\top\A$ with
        \begin{align*}
        \runtime =
         \tilde O\bigg(\Big(\nnz(\A) + \min_{\eta\in(0,1/2)}d^{2}\bar\kappa_{k,1/2+\eta}(\A)\log^{3+8/\eta}(d\kappa(\A))\Big)\log(1/\epsilon)\bigg).
    \end{align*}
\end{theorem}
\begin{remark}\label{r:almost-linear}
    If we set $\eta = \frac1{\sqrt{\log\bar\kappa_{k,\infty}(\A)}}$, then since  $\bar\kappa_{k,p+\eta}(\A) \leq \bar\kappa_{k,p}(\A)\bar\kappa_{k,\infty}(\A)^{\frac{\eta}{p+\eta}}$, we can also obtain an almost linear (up to sub-polynomial factors) running time that scales with $\bar\kappa_{k,1/2}(\A)$:
    \begin{align*}
        \runtime = \tilde O\bigg(\bigg(\nnz(\A) + d^2\bar\kappa_{k,1/2}(\A)\exp\Big(O\Big(\sqrt{\log\bar\kappa_{k,\infty}(\A)}\log \log (d\kappa(\A))\Big)\Big)\bigg)\log(1/\epsilon)\bigg).
    \end{align*}
\end{remark}
\begin{proof}[Proof of \Cref{t:primal-dual-optimized}]    
    We start by computing $k_t = \max\{\lceil d\exp(-\alpha t^2)\rceil,2k\}$ for all $t\in[T]$ where $\alpha=\log(10^3\log^2(80\cdot36 d\tilde\kappa^2))$ and $T=\lceil\sqrt{\log(d/(2k))/\alpha}\rceil$.    
    Let $\A_t = \S_t\A\mPi_{t}^\top$ and $\B_t = \S_{t}\A\mPi_{t-1}^\top$ for $t\in[T]$, where each $\S_t$ is an $s_t\times n$ sparse embedding with $(0.1,\delta,2k_t,O(\log(k_t/\delta)))$-OSE moments, $s_t=O(k_t)+\tilde O(1)$, and each $\mPi_t$ is an $(O(s_t)+\tilde O(1))\times d$ sparse embedding with $(1/4,\delta,s_t,O(\log(s_t/\delta)))$-OSE moments for $t\in[T]$, and we let $\mPi_0=\I_d$. Also, let $\A_0=\S_0\A$, where $\S_0$ is a $s_0\times n$ sparse embedding with $(1,\delta,d,O(\log(d/\delta)))$-OSE moments where $s_0=O(d)+\tilde O(1)$. In particular, this implies that whp.\ $\A_0^\top\A_0\approx_2 \A^\top\A$. Moreover, letting $\tilde\A_t = \A\mPi_t^\top$ and $\tilde\nu_t=\sigAvg_{k_t}(\tilde\A_{t-1})\geq \sigAvg_{k_{t-1}}(\tilde\A_{t-1})$, we have whp.:
    \begin{align*}
        \B_{t}^\top\B_{t} + \tilde\nu_t\I = \tilde\A_{t-1}^\top\S_{t}^\top\S_{t}\tilde\A_{t-1} + \tilde\nu_t\I \approx_{\sqrt {2}} \tilde\A_{t-1}^\top\tilde\A_{t-1}+\tilde\nu_t\I \approx_{\sqrt {2}} \tilde\A_{t-1}\S_{t-1}^\top\S_{t-1}\tilde\A_{t-1}+\tilde\nu_t\I=\A_{t-1}^\top\A_{t-1}+\tilde\nu_t\I.
    \end{align*}
    Using Lemma \ref{l:sigavg}, we have $\tilde\nu_t =\sigAvg_{k_t}(\tilde\A_{t-1}) \leq 2 \sigAvg_{k_t}(\A)$. So, the constructed matrices satisfy
    \begin{align}
        \B_t^\top\B_t+\nu\I \approx_c \A_{t-1}^\top\A_{t-1}+\nu\I\quad\text{for any $\nu\geq 2\sigAvg_{k_t}(\A)$ and $c\geq 2$.}\label{eq:primaldual-condition}
\end{align}
    Also, since $\mPi_t$ and $\mPi_{t-1}$ are both subspace embeddings for dimension $s_t$, we have whp.
    \begin{align*}
        \A_t\A_t^\top = (\S_t\A)\mPi_t^\top\mPi_t(\S_t\A)^\top \approx_{\sqrt 2} (\S_t\A)(\S_t\A)^\top \approx_{\sqrt 2} (\S_t\A)\mPi_{t-1}^\top\mPi_{t-1}(\S_t\A)^\top = \B_t\B_t^\top,
    \end{align*}
    which in particular implies that $\A_t\A_t^\top+\nu\I \approx_2 \B_t\B_t^\top+\nu\I$ for any $\nu\geq 0$. 

    Putting this together, if we could set $\nu_t = 2\sigAvg_{k_t}(\A)$ for each $t$, then we would obtain a primal-dual chain. However, as we cannot simply compute $\sigAvg_{k_t}(\A)$, we will instead show that it is possible to efficiently find a sequence of $\nu_0\leq ...\leq \nu_T$ such that $\nu_t\leq 8\sigAvg_{k_t}(\A)$ that still satisfies the conditions of a primal-dual chain.

    We will show this by induction. Suppose that we have $\nu_{t+1},...,\nu_T$ such that $\nu_0\leq\nu_i\leq 8\sigAvg_{k_i}(\A)$ for $i>t$ and $((\A_t,\nu_t),(\A_i,\B_i,\nu_i)_{i\in\{t+1,...,T\}})\in\primaldualchain_{T-t}$ for any $\nu_t\in(0,\nu_{t+1})$. First, observe that if $t=T-1$, then we obtain this using the provided estimate $\tilde\Lambda$ by setting $\nu_T=4\tilde\Lambda$. Next, we show: 1) how to construct an efficient solver for $\A_{t}^\top\A_{t}+\nu\I$ for any $\nu\in[\nu_0,\nu_{t+1}]$; and 2) how to find the next regularization parameter $\nu_t$.     
    
    \paragraph{Efficient solver for $\A_t^\top\A_t+\nu\I$.} First, observe that we can construct a direct solver $\code{bsolve}$ for $\A_T^\top\A_T+\nu_T\I$ in time $ O(k^\omega)+\tilde O(1)$. We now use that to construct a solver $f_t$ for $\A_t^\top\A_t+\nu\I$ for some $\nu\in[\nu_0,\nu_{t+1}]$. From Lemma \ref{l:primal-dual-general}, using that $\max_{i>t}(\|\B_i\|^2+\nu_{i})/\nu_{i-1} \leq 4(\|\A\|^2+8d\|\A\|^2)/\sigma_d(\A)^2\leq 36d\kappa(\A)^2$, and both $\runtime_{\A_t}$ and $\runtime_{\B_t}$ are $\tilde O(k_t^2)$, we have:
    \begin{align*}
        \runtime_{\code{f_0}} 
        &=  \tilde O\bigg(\sum_{i=t+1}^T
        10^{3i}k_{i-1}^2\sqrt{\nu_i/\nu}\log^{2i}(20\cdot36d\kappa(\A)^2)\log(1/\epsilon)\bigg)
        \\
        &=\tilde O\bigg(\sum_{i=t+1}^T k_{i-1}^2\sqrt{\frac{\sigAvg_{k_i}(\A)}{\sigma_d(\A)^2}}\exp(\alpha i)\log(1/\epsilon)\bigg)
        \\
        &=\tilde O\bigg(d^2\bar\kappa_{k,1/2+\eta}(\A)\exp\Big(\alpha\Big(\frac12(2^2-1)+\frac{2^2}{\eta}\Big)\Big)\log(1/\epsilon)\bigg)\\
        &=\tilde O\bigg(d^2\bar\kappa_{k,1/2+\eta}(\A)\log^{3+8/\eta}(d\kappa(\A))\log(1/\epsilon)\bigg),
    \end{align*}
    where we applied Lemma \ref{l:running-time-general} with $p=2$, and then, used that $\exp(\alpha)=O(\log^2(d\kappa(\A)))$.
    
    \paragraph{Finding the regularization parameter $\nu_t$.} We now discuss how to find the next regularization parameter $\nu_t$ that satisfies condition \eqref{eq:primaldual-condition} and is also bounded by $8\sigAvg_{k_t}(\A)$, given the already obtained sequence $\nu_{t+1},...,\nu_T$. To check \eqref{eq:primaldual-condition}, we use the following Lemma \ref{l:primaldual-condition-tester}, which applies the power method to estimate the extreme eigenvalues that determine the spectral approximation (proof in \Cref{s:primal-dual-testing}).
\begin{lemma}\label{l:primaldual-condition-tester}
    Given matrices $\A\in\R^{n_1\times d},\B\in\R^{n_2\times d}$ and $\nu>0$ such that $\A^\top\A+\nu\I\approx_8 \B^\top\B+\nu\I$, as well as a $d^{-O(1)}$-solver $f$ for $\A^\top\A+\nu\I$, there is an algorithm that in $O((\runtime_f+\runtime_{\A}+\runtime_{\B})\log(d) \log(1/\delta))$ time with probability $1-\delta$ returns $X\in\{0,1\}$ such that:
    \begin{enumerate}
        \item If $X=1$, then $\A^\top\A+\nu\I\approx_{4}\B^\top\B+\nu\I$,
        \item If $X=0$, then $\A^\top\A+\nu\I\not\approx_{2}\B^\top\B+\nu\I$.
    \end{enumerate}
\end{lemma}

We consider a sequence of candidates $\hat\nu_{t,i} := \nu_{T}/2^i$, and proceed by checking condition \eqref{eq:primaldual-condition} using Lemma~\ref{l:primaldual-condition-tester} with $\A=\A_{t-1}$ and $\B=\B_t$ for $i=0,1,2,...$, until the condition fails or we reach $\hat\nu_{t,i}<\nu_0$. First, observe that it would seem natural to only test candidates smaller than $\nu_{t+1}$. However, it is technically possible to reach a corner case where  $\nu=\nu_{t+1}$ is too small to satisfy \eqref{eq:primaldual-condition} with $\A_{t-1}$ and $\B_t$ even though it satisfies the condition with $\A_t$ and $\B_{t+1}$. Moreover, to be able to apply Lemma \ref{l:primaldual-condition-tester}, we need the condition to at least hold with $c=8$, which is why we must start from a candidate that is guaranteed to satisfy this. Since we assume that $\nu_T\approx_2 4\sigAvg_{k_T}(\A)$, this is a safe starting point. 

In order to apply the lemma, we must be able to efficiently solve linear systems with $\A_{t-1}^\top\A_{t-1}+\hat\nu_{t,i}\I$. Since $\hat\nu_{t,i-1}$ satisfies \eqref{eq:primaldual-condition} with $c=4$, we can temporarily consider the following extended primal-dual chain $((\A_{t-1},\hat\nu_{t,i}),(\A_t,,\B_t,\hat\nu_{t,i-1}),(\A_j,\B_j,\max\{\nu_j,\hat\nu_{t,i-1}\})_{j\in\{t+1,...,T\}})\in\primaldualchain_{T-t+1}$ for $\nu_t = \hat\nu_{t,i-1}$. Note that since $\hat\nu_{t,i}\geq\sigma_d^2(\A)/2$ and $\frac{\hat\nu_{t,i-1}}{\hat\nu_{t,i}}=2$, the running time bound obtained via Lemma \ref{l:primal-dual-general} for solving for $\A_{t-1}^\top\A_{t-1}+\hat\nu_{t,i}\I$ is no larger than the one we showed for $\A_t^\top\A_t+\nu\I$, which means the tester from Lemma~\ref{l:primaldual-condition-tester} can be implemented in time $\tilde O(d^2\bar\kappa_{1/2+\eta}(\A)\log^{3+8/\eta}(d\kappa(\A)))$.

Once the algorithm from Lemma \ref{l:primaldual-condition-tester} returns $X=0$ for some $i$, this means that the condition \eqref{eq:primaldual-condition} failed with $c=2$, implying that $\hat\nu_{t,i}<2\sigAvg_{k_t}(\A)$. Thus, we can extend our chain by setting $\nu_t := \hat\nu_{t,i-1}\leq 4\sigAvg_{k_t}(\A)$. Notice that if $\hat\nu_{t,i-1}> \nu_j$ for any $j>t$, then we can simply replace those $\nu_j$ with this larger value, preserving the chain properties.

\paragraph{Completing the proof.} Finding all of the regularization parameters requires  $O(T^2\log(d)\log(1/\delta))=O(\log^2(d)\log(1/\delta))$ linear solves, and leads to a $\tilde O(d^2\bar\kappa_{1/2+\eta}(\A)\log^{3+8/\eta}(d\kappa(\A))\log(1/\epsilon))$ time solver $f_0$ for $\A_0^\top\A_0+\nu_0\I$. To complete the proof, observe that since $\nu_0\approx_2\sigma_d^2(\A)$ and $\A_0^\top\A_0\approx_2\A^\top\A$, we have $\A_0^\top\A_0+\nu_0\I\approx_4 \A^\top\A$. Thus, Applying \Cref{lem:precon}, we can construct a solver for $\A^\top\A$ by applying $\A$ and a $\frac1{40}$-solver $f_0$ at most $O(\log(1/\epsilon))$ times, which recovers the claimed running time.
\end{proof}

We now describe the selection of the remaining hyper-parameters required for Theorem \ref{t:primal-dual-optimized}, leading to the proof of Theorem \ref{t:ls-optimized} (this is very similar to the grid search we did for \Cref{t:ls}).
\begin{proof}[Proof of Theorem \ref{t:ls-optimized}]
    To find the correct values of $\tilde\kappa$, $\tilde\lambda$ and $\tilde\Lambda$, we first compute $ M=\|\A\|_F^2$. Then, we choose $\tilde \kappa= 2,4,8, ...$ and for each of those values, perform a grid search over all pairs of $\tilde\lambda\leq\tilde\Lambda$ from the set $\{M 2^{-t} \mid t\in[\lceil\log_2(d\tilde\kappa)\rceil]\}$. Note that one of the configurations will satisfy the assumptions of Theorem \ref{t:primal-dual-optimized}. This requires evaluating the algorithm $O(\log^3(d\kappa(\A)))$ times, and we return the solution with the lowest least squares error obtained within the time budget.
\end{proof}

\subsection{Testing Spectral Approximation}
\label{s:primal-dual-testing}
In this section, we show how to efficiently test the spectral approximation condition required by our preconditioning chain, leading to the proof of \Cref{l:primaldual-condition-tester}.

Our main idea is to convert the spectral approximation condition into a spectral norm bound, so that it can be tested by estimating the spectral norm of a certain matrix. To do this, we apply a somewhat standard approach of using the power method that we analyze in detail for completeness. Below is a classical guarantee for estimating the spectral norm of a PSD matrix (i.e., its largest eigenvalue).

\begin{lemma}[\cite{kuczynski1992estimating}]\label{l:power}
    Given $\M\in\PD{d}$, let $\lambda_1 = \lambda_{\max}(\M)$. Also, for a given $q\geq 2$, let $\xi_q = \frac{\x^\top\M^{2q+1}\x}{\x^\top\M^{2q}\x}$, where $\x$ is a random vector with standard Gaussian entries. Then, $\xi_q\leq \lambda_1$ and:
    \begin{align*}
        \E\Big[\frac{\lambda_1 - \xi_q}{\lambda_1}\Big] \leq \frac{\log(d)}{q-1}.
    \end{align*}
\end{lemma}
The power method uses matrix-vector products with $\M$. However, we show that it is sufficient to have black-box access to inexact matrix-vector product (matvec) algorithm for $\M$, as defined below. 
\begin{definition}\label{l:matvec}
    We call $f:\R^d\rightarrow\R^d$ an \emph{$\epsilon$-matvec} for $\M\in\R^{d\times d}$ if for all $\b\in\R^d$,  $\|f(\b)-\M\b\|^2\leq\epsilon\|\M\b\|^2$.
\end{definition}
As an auxiliary result, we bound how the matvec error propagates when we apply the matrix repeatedly.
\begin{lemma}\label{l:power-matvec}
    Given $\M\in\R^{d\times d}$ and $q\geq 1$, if $f$ is a $\frac{\epsilon}{(3\kappa(\M))^{2q}}$-matvec for $\M$, then $f^{(q)}$ is a $\epsilon$-matvec for $\M^q$.
\end{lemma}
\begin{proof}
    Let $\hat \x_k = f^{(k)}(\b)$ be an $\epsilon_k$-matvec for $\M^k$ where $\epsilon_k = \epsilon_f(3\kappa(\M))^{2k}$ and $\epsilon_f=\frac{\epsilon}{(3\kappa(\M))^{2q}}$. Observe that:
    \begin{align*}
        \|f(\hat\x_k) - \M^{k+1}\b\| 
        &\leq \|f(\hat\x_k) - \M\hat\x_k\| + \|\M(\hat\x_k - \M^k\b)\|
        \\
        &\leq \sqrt{\epsilon_f}\|\M\hat\x_k\| + \|\M(\hat\x_k - \M^k\b)\|
        \\
        &\leq \sqrt{\epsilon_f}\|\M^{k+1}\b\| + (1+\sqrt{\epsilon_1})\|\M(\hat\x_k-\M^k\b)\|
        \\
        &\leq \sqrt{\epsilon_f}\|\M^{k+1}\b\| + (1+\sqrt{\epsilon_f})\sqrt{\epsilon_k}\|\M\|\|\M^k\b\|
        \\
        &\leq \sqrt{\epsilon_f} + (1+\sqrt{\epsilon_f})\sqrt{\epsilon_k}\kappa(\M)\|\M^{k+1}\b\|
        \leq 3\sqrt{\epsilon_k}\kappa(\M)\|\M^{k+1}\b\|.
    \end{align*}
    In particular, this implies that $\epsilon_q\leq \epsilon$, which concludes the claim.
\end{proof}
Combining the inexact matvecs with \Cref{l:power}, we obtain a guarantee for the inexact power method.
\begin{lemma}\label{l:norm-estimate}
    Given a $\frac{\epsilon/8}{(3\kappa(\M))^{5q+3}}$-matvec $f$ for $\M\in\PD{d}$ and $q\geq 32\log(d)/\epsilon$, we can compute $X\in \R$ such that $X\approx_{1+\epsilon} \|\M\|$ with probability $1-\delta$ in time $O((\runtime_f + d)q\log 1/\delta)$.
\end{lemma}
\begin{proof}
    Using Lemma \ref{l:power-matvec}, $g=f^{(2q+1)}$ is an $\epsilon_g$-matvec for $\M^{2q+1}$, where $\epsilon_g=\frac{\epsilon/8}{(\kappa(\M))^{q+1}}$. Now, let $\x_i$ be i.i.d. vectors with standard Gaussian entries, and define $X_{i,j}=\x_i^\top f^{(j)}(\x_i)$. Then:
    \begin{align*}
        |X_{i,2q+1} - \x_i^\top\M^{2q+1}\x_i| 
        \leq \epsilon_g\|\x_i\|\|\M^{2q+1}\x_i\|\leq \epsilon_g(\kappa(\M))^{q+1}\x_i^\top\M^{2q+1}\x_i \leq (\epsilon/8)\x_i^\top\M^{2q+1}\x_i,
    \end{align*}
    It follows that $X_{i,2q+1} \approx_{1+\epsilon/4} \x_i^\top\M^{2q+1}\x_i$ and a similar guarantee follows for $X_{i,2q}$. We use this to construct $X_i = X_{i,2q+1}/X_{i,2q}\approx_{(1+\epsilon/4)^2}\frac{\x_i^\top\M^{2q+1}\x_i}{\x_i^\top\M^2q\x_i}=:Y_i$. Combining Lemma \ref{l:power} with Markov's inequality, it follows that for $q\geq 32\log(d)/\epsilon$ with probability 0.75 we have $ Y_i\approx_{1+\epsilon/4} \|\M\|$. Thus, for each $i$, with probability $0.75$ we have $X_i\approx_{(1+\epsilon/4)^3}\|\M\|$. Taking a median over $O(\log1/\delta)$ such estimates, we obtain the desired estimate with probability $1-\delta$.
\end{proof}
We are now ready to use the inexact power method for efficiently testing the spectral approximation guarantee in our preconditioning chain, thereby proving \Cref{l:primaldual-condition-tester}.
\begin{proof}[Proof of Lemma \ref{l:primaldual-condition-tester}]
    Let $\M = \A^\top\A+\nu\I$, $\N = \B^\top\B+\nu\I$, $\A_\nu = [\A^\top\mid \sqrt\nu\I]^\top$, and $\B_{\nu}=[\B^\top\mid\sqrt\nu\I]^\top$, so that $\M=\A_\nu^\top\A_\nu$ and $\N = \B_\nu^\top\B_\nu$. Then, we have:
    \begin{align*}
        \A^\top\A+\nu\I\approx_c \B^\top\B+\nu\I\quad\Leftrightarrow\quad
        \|\A_{\nu}\N^{-1}\A_{\nu}^\top\|\leq c\text{ and }\|\B_{\nu}\M^{-1}\B_{\nu}^\top\|\leq c.
    \end{align*}
    We are going to construct two estimates, $\hat a \approx_{4/3} \|\A_{\nu}\N^{-1}\A_{\nu}^\top\|$ and $\hat b\approx_{4/3} \|\B_{\nu}\M^{-1}\B_{\nu}^\top\|$, and then define:
    \begin{align*}
        X := \mathbb{1}_{\hat a\leq 3}\cdot \mathbb{1}_{\hat b\leq 3}.
    \end{align*}
    Note that if $X=1$, then both spectral norms are bounded by $(4/3)\cdot 3=4$, so the spectral approximation holds with $c=4$. On the other hand if $X=0$, then one of the norms is larger than $(3/4)\cdot 3>2$, so the spectral approximation fails with $c=2$.

    It remains to produce the estimates $\hat a$ and $\hat b$. We do this by applying Lemma \ref{l:power-matvec} to $\A_{\nu}\N^{-1}\A_{\nu}^\top$ and $\B_{\nu}\M^{-1}\B_{\nu}^\top$. 
    Note that, letting $\b=\A_\nu^\top\x$, 
    \begin{align*}
    \|\A_\nu f(\A_\nu^\top\x) - \A_\nu\N^{-1}\A_\nu^\top\x\|^2= \|f(\b)-\N^{-1}\b\|_{\M}^2\leq 8\epsilon_f\|\N^{-1}\b\|_{\N}^2 \leq 64\epsilon_f\|\A_\nu\N^{-1}\A_\nu^\top\x\|^2.
    \end{align*}
    So $g(\x) = \A_\nu f(\A_\nu^\top\x)$ is a $64\epsilon_f$-matvec for $\A_\nu\N^{-1}\A_\nu^\top$. Since $\kappa(\A_\nu\M^{-1}\A_\nu^\top)\leq 64$, applying Lemma~\ref{l:power-matvec} with $\epsilon_f = \frac1{64}\frac{(4/3)/8}{(3\cdot 9)^{O(\log d)}}=d^{-O(1)}$ we can construct the desired estimate $\hat a$, and $\hat b$ follows analogously.    
\end{proof}

\section{Implicit Dual Solver for Positive Definite Systems}
\label{s:psd}

In this section, we describe the implicit dual variant of our recursive preconditioning framework, which is specialized for PD linear systems. In particular the framework applies to the problem of solving $\M\x=\b$ for input PD matrix $\M\in\PD{d}$ and vector $\b\in\R^d$. We use this framework to Theorem \ref{t:psd-optimized}.

\subsection{Reduction Tools}

First, we provide an auxiliary lemma which shows how to use a solver for $\A$ in order to solve for $\A^2$.
\begin{lemma}\label{l:square}
Let $f$ be a $\frac{\epsilon}{9 \kappa^2(\A)}$-solver for $\A\in\PD{n}$. Then $f^{(2)}(\b):=f(f(\b))$, is a $\epsilon$-solver for $\A^2$.
\end{lemma}
\begin{proof}
    Let $\hat\x_1=f(\b)$ and $\hat\x_2 = f(\hat\x_1)$, and define $\epsilon'=\frac{\sqrt{\epsilon}}{40 \kappa(\A)}$. Also let $\x_1^*=\A^{-1}\b$ and $\x_2^*=\A^{-1}\x_1^*$. Recall that by definition of $f$ we have $\|\hat\x_1-\x_1^*\|_{\A}^2\leq\epsilon'\|\x_1^*\|_{\A}^2$ and also $\|\hat\x_2-\A^{-1}\hat\x_1\|_{\A}^2\leq\epsilon'\|\A^{-1}\hat\x_1\|_{\A}^2$. Thus, 
    \begin{align*}
        \|\hat\x_2-\x_2^*\|_{\A}
        &= \|\hat\x_2-\A^{-1}\hat\x_1 + \A^{-1}(\hat\x_1-\x_1^*)\|_{\A}
        \\
        &\leq \|\hat\x_2-\A^{-1}\hat\x_1\|_{\A}+\|\A^{-1}(\hat\x_1-\x_1^*)\|_{\A}
        \\
        &\leq \sqrt{\epsilon'}\|\A^{-1}\hat\x_1\|_{\A} 
        + \|\A^{-1}\|\cdot\|\hat\x_1-\x_1^*\|_{\A}
        \\
        &=\sqrt{\epsilon'}\|\x_2^* + \A^{-1}(\hat\x_1-\x_1^*)\|_{\A} 
        + \|\A^{-1}\|\cdot\|\hat\x_1-\x_1^*\|_{\A}
        \\
        &\leq \sqrt{\epsilon'}\|\x_2^*\|_{\A}+(\sqrt{\epsilon'} + 1)\|\A^{-1}\|\|\hat\x_1-\x_1^*\|_{\A}
        \\
        &\leq  \sqrt{\epsilon'}\|\x_2^*\|_{\A} + (\sqrt{\epsilon'}+1)\|\A^{-1}\|\sqrt{\epsilon'}\|\x_1^*\|_{\A}
        \\
        &\leq \big(\sqrt{\epsilon'}+\sqrt{\epsilon'}(\sqrt{\epsilon'}+1)\kappa(\A)\big)\|\x_2^*\|_{\A}.
    \end{align*}
    It remains to show that the resulting error quantity is bounded by $\epsilon$. This follows from the definition of $\epsilon'$, since:
    $\big(\sqrt{\epsilon'}+\sqrt{\epsilon'}(\sqrt{\epsilon'}+1)\kappa(\A)\big)^2
        \leq \big(\sqrt{\epsilon/9}+2\sqrt{\epsilon/9}\big)^2=\epsilon$.
\end{proof}

Now, we provide our main reduction tool, which makes the implicit dual recursive preconditioning possible.   Namely, we show that given two matrices $\A$ and $\B$ such that $\B^\top\B$ is a low-rank approximation of $\A^\top\A$ , we can use a solver for  $\B\B^\top+\nu\I$ in order to solve $\A\A^\top+\nu\I$. This lemma might appear similar in statement to Lemma \ref{l:primal-dual-recursion}, but crucially, it avoids direct access to  matrices $\A$ and $\B$ in favor of matrix-vector operations with the ``dual" matrices $\A\A^\top$ and $\A\B^\top$, for which we provide a different solver construction.

\begin{lemma}\label{l:dual-reduction}
    Consider matrices $\A\in\R^{{n_1}\times d}$, $\B\in\R^{{n_2}\times d}$, and $\nu >0$. Let $\kappa=1+\|\A\|^2/\nu$ and suppose that 
    \begin{align*}
        \B^\top\B+\nu\I\approx_4 \A^\top\A+\nu\I.
    \end{align*}
    Given a $\frac1{20}$-solver $f$ for $\B\B^\top+\nu\I_{n_2}$, we can construct an $\epsilon$-solver for 
    $\A\A^\top+\nu\I_{n_1}$ that applies $f$ and $\B\B^\top$ at most $10^5\log^2(2\kappa)\log(2/\epsilon)$ times and spends an additional $O((\runtime_{\A\B^\top} + \runtime_{\A\A^\top} + n_1+n_2)\log(1/\epsilon))$ time. 
\end{lemma}
\begin{proof}
First, we define the following preconditioner for the matrix $\M=\A\A^\top+\nu\I$, which is an extension of the classical Nystr\"om approximation \cite{Williams01Nystrom}, except in our case both inner and outer matrices are regularized:
\begin{align*}
\N = \C\W_{\nu}^{-1}\C^\top+\nu\I,\quad \text{for}\quad \C=\A\B^\top \text{ and }\W_{\nu}=\B\B^\top+\nu\I,
\end{align*}
Using $\B^\top\B+\nu\I_d\approx_4\A^\top\A+\nu\I_d$, it follows that:
 \begin{align*}
     \M - \N
     &= \A(\I - \B^\top(\B\B^\top+\nu\I)^{-1}\B)\A^\top
     \\
     &= \A(\I - \B^\top\B(\B^\top\B+\nu\I)^{-1})\A^\top
     \\
     &= \A\Big((\B^\top\B+\nu\I)(\B^\top\B+\nu\I)^{-1} - \B^\top\B(\B^\top\B+\nu\I)^{-1}\Big)\A^\top
     \\
     &\preceq \nu\cdot \A(\B^\top\B+\nu\I)^{-1}\A^\top\\
     &\preceq 4\nu \cdot \A(\A^\top\A+\nu\I)^{-1}\A^\top\preceq 4\nu \I.
 \end{align*}
 The first line of the above calculation also shows that $\N\preceq \M$, so altogether:
 \begin{align*}
     \M\succeq\N\succeq \frac1{5}\big(\C\W_{\nu}^{-1}\C^\top+5\nu\I\big)
     \succeq \frac15(\N + \M - \N)= \frac1{5}\M.
 \end{align*}
 Thus, using Lemma \ref{lem:precon}, given a $\frac1{50}$-solver $g$ for $\N$, we can construct an $\epsilon$-solver for $\M$ that applies $g$ at most $\lceil4\sqrt5\log(2/\epsilon)\rceil$ times.

 Next, we show how to construct a solver for $\N$, given the solver $f$ for $\W_{\nu}$. First, we use the Woodbury formula to write:
 \begin{align*}
     \N^{-1} = \frac1\nu\big(\I - \C(\C^\top\C+\nu\W_{\nu})^{-1}\C^\top\big),
 \end{align*}
 which means that it suffices to construct a solver for $\C^\top\C+\nu\W_{\nu}$. Here, observe that:
 \begin{align*}
     \C^\top\C+\nu\W_{\nu} 
     &= \B\A^\top\A\B^\top + \nu\B\B^\top+\nu^2\I
= \B(\A^\top\A + \nu\I)\B^\top + \nu^2\I
     \\
     &\approx_4 \B(\B^\top\B+\nu\I)\B^\top+\nu^2\I
= (\B\B^\top)^2 + \nu\B\B^\top + \nu^2\I
     \\
     &\approx_4 (\B\B^\top)^2 + \nu\B\B^\top + \frac{\nu^2}4\I
     = \Big(\B\B^\top+ \frac\nu 2\I\Big)^2 = (\W_{\nu/2})^2.
 \end{align*}
Since $\frac12\W_{\nu}\preceq \W_{\nu/2}\preceq \W_{\nu}$, we can construct a $\epsilon_h$-solver $h$ for $\W_{\nu/2}$ that applies the $\frac{1}{20}$-solver $f$ at most $\lceil4\sqrt{2}\log(2/\epsilon_h)\rceil$, and using Lemma~\ref{l:square}, $h^{(2)}$ is a $9\epsilon_h\hat\kappa^2$-solver for $(\W_{\nu/2})^2$, where $\hat\kappa = \kappa(\W_{\nu/2})$. Next, from the above calculation we have $\frac14(\W_{\nu/2})^2\preceq\C^\top\C+\nu\W\preceq 16(\W_{\nu/2})^2$, so we can use $\frac1{640}$-solver $h^{(2)}$ to obtain a $\epsilon_{g_0}$-solver $g_0$ for $\C^\top\C+\nu\W_\nu$, which requires $\lceil32\log(2/\epsilon_{g_0})\rceil$ applications of $h^{(2)}$.

Since $\B^\top\B+\nu\I\approx_{4}\A^\top\A+\nu\I$, we have that
$\hat\kappa = \kappa(\W_{\nu/2}) \leq 4(\|\A\|^2+\nu)/(\nu/2) = 8\kappa$.
Finally, we construct the $\epsilon_g$-solver $g$ for $\N$ from a $\frac{\epsilon_g\nu^2}{\|\N\|^2}$-solver $g_0$, by relying on Lemma \ref{l:woodbury}, which shows that this is possible with one application of $g_0$, $\C$, and $\C^\top$ each. Note that $\|\N\|/\nu\leq \|\M\|/\nu=\kappa$.

Now, setting $\epsilon_g = \frac1{50}$ and propagating it back, we need $\epsilon_{g_0}=\frac1{50\kappa^2}$, and $\epsilon_h = \frac1{4\cdot 10^5\kappa^2}$, which gives the total number of applications of $f$ and of matrix $\B\B^\top$ to obtain an $\epsilon$-solver for $\M$ as:
\begin{align*}
    \lceil4\sqrt{2}&\log(2/\epsilon_h)\rceil\cdot \lceil32\log(2/\epsilon_{g_0})\rceil\cdot \lceil4\sqrt{5}\log(2/\epsilon)\rceil\\
    &=
    \lceil4\sqrt{2}\log(8\cdot10^5\kappa^2)\rceil\cdot \lceil32\log(100\kappa^2)\rceil\cdot \lceil4\sqrt{5}\log(2/\epsilon)\rceil\\
    &\leq  10^5 \log^2(2\kappa)\log(2/\epsilon).
\end{align*}
Furthermore, the overall cost of the $\epsilon$-solver for $\M$ also involves applying $\A\A^\top$ and $\A\B^\top$ at most $O(\log1/\epsilon)$ many times, plus additional $O((n_1+n_2)\log1/\epsilon)$ runtime.
 \end{proof}

\subsection{Implicit Dual Recursive Preconditioning}

Next, we present our implicit dual recursive preconditioning framework. We start with an auxiliary result which shows that chaining together a number of subspace embedding matrices preserves the OSE property.

\begin{lemma}\label{l:ose-chaining}
     Suppose that random matrices $\S_t\in\R^{s_t\times s_{t-1}}$ have $(\epsilon,\delta,d,\ell)$-OSE moments for all $t\in[T]$, $\epsilon,\delta\in(0,1)$ and $\epsilon\delta^{1/\ell}\leq \frac1{2\ell}$. Then, the matrix $\bar\S=\S_T\S_{T-1}... \S_1$ has $(\epsilon,3\delta,d,\ell)$-OSE moments.
 \end{lemma}
\begin{proof}
    We apply the argument from Remark 4 of \cite{cohen2016optimal} recursively. Namely, suppose that $\bar\S_t = \S_t\S_{t-1}...\S_1$ satisfies $(\epsilon,\delta_t,d,\ell)$-OSE moments with $\epsilon\delta_t^{1/\ell}\leq 1/\ell$. Then, using Lemma 2 from \cite{cohen2016optimal}, for any $\U\in\R^{s_0\times d}$ such that $\U^\top\U=\I$, we have:
    \begin{align*}
        \Big(\E\|(\S_{t+1}\bar\S_t\U)^\top\S_{t+1}\bar\S_t\U-\I\|^\ell\Big)^{1/\ell}
        &\leq \epsilon\delta^{1/\ell} \Big(\E\|(\bar\S_t\U)^\top\bar\S_t\U\|^\ell\Big)^{1/\ell}
        \\
        &\leq \epsilon\delta^{1/\ell} \big(1 + \epsilon\delta_t^{1/\ell}\big)
        \\
        &\leq \epsilon\delta^{1/\ell}\big(1+1/\ell\big)\leq \epsilon(3\delta)^{1/\ell}.
    \end{align*}
    Thus, $\bar\S_{t+1}$ has $(\epsilon,\delta_{t+1},d,\ell)$-OSE moments for some $\delta_{t+1}\leq 3\delta$ such that $\epsilon\delta_{t+1}^{1/\ell}\leq \epsilon\delta^{1/\ell}(1+1/\ell)\leq 1/\ell$. Thus, applying the recursion we obtain the claim.
\end{proof}
Next, we describe the construction of the regularized preconditioning chain for the dual setting. Interestingly, we rely here on the same preconditioning chain as the one we used in \Cref{s:ls} (see \Cref{def:primalchain}). However, unlike in that section, we must construct the $\A_t$ matrices recursively, i.e., so that $\A_{t+1} = \S_{t+1}\A_t$, rather than by directly sketching the original matrix.
\begin{lemma}\label{l:dual-chain}
    Consider a matrix $\A_0\in\R^{n\times d}$, and random matrices $\S_t\in\R^{s_t\times s_{t-1}}$ with $s_0=n$ such that each $S_t$ has $(\epsilon,\delta,2k_t,\ell)$-OSE moments for $d\geq k_1\geq k_2\geq ...\geq k_T> 0$, $t\in[T]$ and $\epsilon\leq1/6$. If $\epsilon\delta^{1/\ell}\leq \frac1{2\ell}$ and $\ell\geq 2$, then the matrices $\A_t = \S_t\S_{t-1}...\S_1\A_0$ with probability $1-3T\delta$ form a (primal) regularized preconditioning chain $((\A_t,\bar\Sigma_{k_t}(\A_t))_{t\in\{0,1,...,T\}})\in\primalchain_T^{d}$.
\end{lemma}
\begin{proof}
    It is easy to verify that each matrix $\S_t$ has $(\epsilon,\delta,2k_i,\ell)$-OSE moments for all $i\in\{t,...,T\}$. Indeed, for any $\U\in\R^{s_{t-1}\times k_t}$ with orthonormal columns, let $\U_{[k_i]}$ be the matrix consisting of the first $k_i$ columns of $\U$. Then $(\S_t\U_{[k_i]})^\top\S_t\U_{[k_i]}$ is a principal submatrix of $(\S_t\U)^\top\S_t\U$, so the extreme eigenvalues of the former are upper/lower-bounded by the extreme eigenvalues of the latter, and so $\|(\S_t\U_{[k_i]})^\top\S_t\U_{[k_i]}-\I_{k_i}\|\leq \|(\S_t\U)^\top\S_t\U-\I_{k_t}\|$, which implies the corresponding inequality for the moments.  
    
    Thus, using Lemma \ref{l:ose-chaining}, we conclude that matrices $\bar\S_t=\S_t\S_{t-1}...\S_1$ have $(\epsilon,3\delta,2k_t,\ell)$-OSE moments for each $t\in[T]$. Using Lemma \ref{l:reg-approx} this implies that, with probability $1-3T\delta$, for each $t$ and $\eta_t=\sigAvg_{k_t}(\A)$:
    \begin{align*}
        \A_t^\top\A_t+\nu_t\I \ \approx_{1+6\epsilon} \A_0^\top\A_0+\nu_t\I\ \approx_{1+6\epsilon} \A_{t-1}^\top\A_{t-1}+\nu_t\I
    \end{align*}
    which concludes the claim since $\epsilon\leq 1/6$.
\end{proof}

We now present our implicit dual recursive preconditioning. Note that the result is stated here in terms of the $\A_t$ matrices, even though during runtime we do not have access to these matrices but rather only to $\A_t\A_t^\top$ and to $\A_t\A_{t-1}^\top$. We make this explicit once we apply the framework to construct our PD solver.
\begin{lemma}\label{l:dual-general}
    Let $P = ((\A_t, \nu_t)_{t \in \{0,1,\ldots,T\}})\in\primalchain_T^d$ be a regularized preconditioning chain.  Suppose $\code{bsolve}$ is a $\frac1{20}$-solver  for $\A_T\A_T^\top+\nu_T\I$. Given $\epsilon\in(0,1)$, there is an $\epsilon$-solver $\code{solve}$ for $\A_0\A_0^\top+\nu_0\I$ with 
    \begin{align*}
        \runtime_{\code{solve}} &= O\bigg(K_T \runtime_{\code{bsolve}} + \sum_{t\in[T]} K_t\Big(\runtime_{\A_{t-1}\A_{t-1}^\top} + \runtime_{\A_t\A_{t-1}^\top} + \runtime_{\A_{t}\A_{t}^\top}\Big)\bigg),
        \\
        \text{where}\qquad K_t &:= (8\cdot 10^5)^t\sqrt{\nu_t/\nu_0}\log^{3t}(4\kappa)\log(2/\epsilon),\quad\text{ and }\quad
        \kappa = \max_{t\in[T]}\,(\|\A_{t-1}\|^2+\nu_t)/\nu_{t-1}.
    \end{align*}
\end{lemma}
\begin{proof}
    Since $\A_t^\top\A_t+\nu_t\I\approx_{4}\A_{t-1}^\top\A_{t-1}+\nu_t\I$, we can recursively apply Lemma \ref{l:dual-reduction} to each pair of $\A_{t-1}$ and $\A_t$. Let $\kappa_t = \frac{\nu_t}{\nu_{t-1}}\leq \kappa$. The lemma implies that given a $\frac1{20}$-solver $f_t$ for $\A_t\A_t^\top+\nu_t\I$, we can construct a $\frac1{10\kappa_{t}}$-solver $\tilde f_{t-1}$ for $\A_{t-1}\A_{t-1}^\top+\nu_t\I$ which applies $f_t$, $\A_t\A_t^\top$, $\A_t\A_{t-1}^\top$ and $\A_{t-1}\A_{t-1}^\top$ at most $10^5\log^2(4\kappa)\log(2\kappa_t)\leq 10^5\log^3(4\kappa)$ times each.
    
    Now, since $\A_{t-1}\A_{t-1}^\top+\nu_{t-1}\I\preceq \A_{t-1}\A_{t-1}^\top+\nu_t\I\preceq \frac{\nu_t}{\nu_{t-1}}(\A_{t-1}\A_{t-1}^\top+\nu_{t-1}\I)$, using Lemma \ref{lem:precon} we can construct a $\epsilon_{t-1}$-solver $f_{t-1}$ for $\A_{t-1}\A_{t-1}^\top+\nu_{t-1}\I$ which applies $\tilde f_{t-1}$ and $\A_{t-1}\A_{t-1}^\top$ at most $\lceil4\sqrt{\nu_t/\nu_{t-1}}\log(2/\epsilon_{t-1})\rceil$ times. Here, to put this together recursively we can use $\epsilon_{t-1}=\frac1{20}$ until we reach the target solver with $\epsilon_0=\epsilon$. Thus, using the shorthand $\runtime_{\A_{t-1},\A_t}=\runtime_{\A_{t-1}\A_{t-1}^\top} + \runtime_{\A_t\A_{t-1}^\top}+\runtime_{\A_t\A_t^\top}$, for any $t\geq 2$ we have:
    \begin{align*}
    \runtime_{f_{t-1}} 
    &= \lceil4\sqrt{\nu_t/\nu_{t-1}}\log(40)\rceil(\runtime_{\tilde f_{t-1}} + O(\runtime_{\A_{t-1}\A_{t-1}^\top}))
    \\
    &\leq
    \lceil4\sqrt{\nu_t/\nu_{t-1}}\log(40)\rceil
    \Big(10^5\log^3(4\kappa)\big(\runtime_{f_t} + O(\runtime_{\A_{t-1},\A_t})\big)+O(\runtime_{\A_{t-1}\A_{t-1}^\top})\Big)
    \\
    &\leq 8\cdot 10^5\sqrt{\nu_t/\nu_{t-1}}\log^3(4\kappa)\big(\runtime_{f_t} + O(\runtime_{\A_{t-1},\A_t})\big),
    \end{align*}
whereas for $t=1$, the right-hand side needs to additionally be multiplied by $\log(2/\epsilon)$. To unroll this recursion, we perform induction over decreasing $t$. To that end, let $\alpha_t = 8\cdot 10^5\sqrt{\nu_t/\nu_{t-1}}\log^3(4\kappa)$ for any $t\geq 2$, with $\alpha_1=8\cdot 10^5\sqrt{\nu_1/\nu_{0}}\log^3(4\kappa)\log(2/\epsilon)$. Also, let $K_{t_1,t_2} = \prod_{i=t_1}^{t_2}\alpha_i$. Then, the above recursion states that $\runtime_{f_{t-1}}\leq \alpha_t(\runtime_{f_t} + \runtime_{\A_{t-1},\A_t})$.

By induction, the recursion implies that 
\begin{align}
    \runtime_{f_t}\leq K_{t+1,T}\runtime_{\code{bsolve}} + \sum_{i \in [T-t]}K_{t+1,t+i}\runtime_{\A_{t+i-1},\A_{t+i}},\label{eq:dual-induction}
\end{align}
where $\runtime_{\code{bsolve}}=\runtime_{f_T}$. Indeed, for $t=T-1$ this is simply the recursion. Suppose that this holds for some $t$. Then, 
\begin{align*}
    \runtime_{f_{t-1}}
    &\leq \alpha_t\Big(K_{t+1,T}\runtime_{\code{bsolve}} + \sum_{i=1}^{T-t}K_{t+1,t+i}\runtime_{\A_{t+i-1},\A_{t+i}}\Big) + \alpha_t\runtime_{\A_{t-1},\A_t}
    \\
    &= K_{t,T}\runtime_{\code{bsolve}} + K_{t,t}\runtime_{\A_{t-1},\A_t} + \sum_{i=2}^{T-t+1}K_{t-1+1,t-1+i}\runtime_{\A_{t-1+i-1},\A_{t-1+i}}.
\end{align*}
Finally, applying \eqref{eq:dual-induction} to  $f_0=\code{solve}$ and observing that $K_t = K_{1,t}$ we obtain the claim.
\end{proof}
\subsection{Positive Definite Solver}
We next describe how our implicit dual recursive preconditioning can be applied to construct a fast solver for PD linear systems that attains an improved conditioning dependence. Below is the main result of this section, from which \Cref{t:psd-optimized} follows as a corollary.

\begin{theorem}\label{t:dual-main-optimized}
 There is an algorithm that given $\M\in\PD{d}$, $k\in[d]$, $\tilde\kappa\approx_2\kappa(\M)$, $\tilde\lambda\approx_2 \lambda_{d}(\M)$, $\tilde\Lambda\approx_2\lambAvg_{2k}(\M)$, in $\tilde O(d^2)+O(k^\omega)$ time whp.\ produces $\runtime$-time $\epsilon$-solver for $\M$ with
        \begin{align*}
        \runtime = \tilde O\bigg(
        \min_{\eta\in(0,1/4)}d^{2} \sqrt{\bar\kappa_{k,1/4+\eta}(\M)}\log^{9/2+6/\eta}(d\kappa(\M))\log(1/\epsilon)\bigg).
    \end{align*}
\end{theorem}
\begin{proof}
    Define $\alpha=\log(8\cdot10^5\log^3(8\cdot 20d\tilde\kappa))$ and $k_t = \max\{\lceil d\exp(-\alpha t^2)\rceil,2k\}$ for each $t\in[T]$, where $T=\lceil\sqrt{\log(d/(2k))/\alpha}\rceil$.
    We take $\A=\A_0$ such that $\M=\A\A^\top$. Consider sparse subspace embeddings $\S_t\in\R^{s_t\times s_{t-1}}$ from \Cref{l:ose-sparse} with $s_t = O(k_t)+\tilde O(1)$ such that each $\S_t$ has $(0.1,\delta,2k_t,O(\log(k_t/\delta)))$-OSE moments. Then, using \Cref{l:dual-chain}, the matrices $\A_t = \S_t\S_{t-1}...\S_1\A_0$ whp.\ satisfy the following guarantee:
    \begin{align}
        \A_t^\top\A_t+\nu\I \approx_c \A_{t-1}^\top\A_{t-1}+\nu\I \quad\text{for any }\nu\geq \sigAvg_{k_t}(\A)
        \text{ and }c\geq 2. \label{eq:dual-condition}
    \end{align}
     Since we do not have direct access to matrix $\A_0$, nor any of the matrices $\A_t$, we will work with the following matrices $\M_t = \A_t\A_t^\top = \S_t\M_{t-1}\S_t^\top$ and $\C_t = \A_{t-1}\A_{t}^\top = \M_{t-1}\S_t^\top$. Constructing each pair $\M_{t+1}$ and $\C_{t+1}$ from $\M_{t}$ via Lemma \ref{l:ose-sparse} takes $\tilde O(\nnz(\M_t) + \nnz(\C_t)) = \tilde O(k_{t-1}^2)$, so the total construction cost is $\tilde O(d^2)$.

    Next, we construct $\code{bsolve}$, the $\frac1{20}$-solver for $\A_T\A_T^\top+\nu_T\I=\M_T+\nu_T\I$. In fact, since $s_T=O(k)+\tilde O(1)$, we can do this by simply computing the inverse $(\M_T+\nu_T\I)^{-1}$ at the cost of $O(k^\omega)+\tilde O(1)$, and then the solver applies that matrix to a vector at the cost of $\runtime_{\code{bsolve}} = \tilde O(k^2)$.

    Our goal is to now construct a regularized preconditioning chain, following Definition \ref{def:primalchain}. To achieve this, we will show how to efficiently find $\nu_t\leq 4\sigAvg_{k_t}(\A)$ that satisfy condition \eqref{eq:dual-condition} with $c=4$, thus ensuring the preconditioning chain property for $((\A_t,\nu_t)_{t\in\{0,1,...,T\}})\in\primalchain_T^d$. 

    The argument follows by induction, similarly to the one in Section \ref{s:primal-dual}. Let $\nu_T = 2\tilde\Lambda$, $\nu_0=\tilde\lambda$, and suppose that $\nu_{t+1},...,\nu_T$ satisfy $\nu_0\leq\nu_i\leq 4\sigAvg_{k_i}(\A)$ for all $i>t$ and $((\A_i,\nu_i)_{i\in\{t,t+1,...,T\}})\in\primalchain_{T-t}^d$ for any $\nu_t\in[\nu_0,\nu_{t+1}]$. The base case of the induction follows from our definition of $\nu_T$. In order to show the inductive case, we start by describing an efficient solver for $\A_t\A_t^\top+\nu\I$.
    
    \paragraph{Efficient solver for $\A_t\A_t^\top+\nu\I$.} We apply Lemma \ref{l:dual-general} with $((\A_i,\nu_i)_{i\in\{t,t+1,...,T\}})\in\primalchain_{T-t}^d$ for any $\nu_t\in[\nu_0,\nu_{t+1}]$, using that $\max_{i>t}(\|\A_{i-1}\|^2+\nu_{i})/\nu_{i-1} \leq 4(\|\A\|^2+4d\|\A\|^2)/\sigma_d(\A)^2\leq 20d\kappa(\M)$, to construct an $\epsilon$-solver $f_t$ for $\A_t\A_t^\top+\nu\I$ running in time:
    \begin{align*}
        \runtime_{f_t} 
        &= \tilde O\bigg(\sum_{i=t+1}^T k_{i-1}^2\sqrt{\frac{\nu_i}{\nu}}\Big(8\cdot 10^5\log^3(4\cdot 20d\kappa(\M))\Big)^i\log(1/\epsilon)\bigg)
        =\tilde O\bigg(\sum_{i=t+1}^T k_{i-1}^2\sqrt{\frac{\sigAvg_{k_i}(\A)}{\sigma_d(\A)^2}}\exp(\alpha i) \log(1/\epsilon)\bigg).
    \end{align*}
    Applying Lemma \ref{l:running-time-general} with $p=2$, we obtain that:
    \begin{align}
        \runtime_{f_t} &= \tilde O\bigg(d^2\bar\kappa_{k,1/2+\eta}(\A)\exp\Big(\alpha\Big(\frac32 + \frac4\eta\Big)\Big)\log(1/\epsilon)\bigg)\\
        &= \tilde O\bigg(d^2\sqrt{\bar\kappa_{k,1/4+\eta/2}(\M)}\log^{9/2+12/\eta}(d\kappa(\M))\log(1/\epsilon)\bigg).\label{eq:dual-bound}
    \end{align}    
    where we use that since $\M=\A\A^\top$, we have $\bar\kappa_{k,p}(\A):=\sqrt{\bar\kappa_{k,p/2}(\M)}$. Performing a change of variable from $\eta/2$ to $\eta$ we recover the correct exponent in the logarithm.

\paragraph{Finding the regularization parameter $\nu_t$.} We now discuss how to find the parameter $\nu_t$ that can be used to extend the preconditioning chain. This step follows analogously as in the proof of Theorem \ref{t:dual-main-optimized}. The key difference lies in how we test the spectral approximation condition \eqref{eq:dual-condition}. Here, the power method must be implemented with additional care, since we do not have direct access to the matrices $\A_t$. This is attained with the following lemma, proven below in Section \ref{s:dual-condition-tester}.
\begin{lemma}\label{l:dual-condition-tester}
    Consider matrices $\A\in\R^{n_1\times d}$, $\B\in\R^{n_2\times d}$, and $\nu>0$,
    such that $\A^\top\A+\nu\I\approx_8\B^\top\B+\nu\I$ and let $\kappa = 1+\|\A\|^2/\nu$.
    Given $\kappa^{-2}d^{-O(1)}$-solvers $f$ and $g$ for $\A\A^\top+\nu\I$ and $\B\B^\top+\nu\I$, there is an algorithm that in $O((\runtime_f + \runtime_{g}+\runtime_{\A\A^\top}+\runtime_{\B\B^\top}+\runtime_{\A\B^\top})\log (d)\log(1/\delta))$ time whp.\ returns $X\in\{0,1\}$ such that:
    \begin{enumerate}
        \item If $X=1$, then $\A^\top\A+\nu\I\approx_4\B^\top\B+\nu\I$,
        \item If $X=0$, then $\A^\top\A+\nu\I\not\approx_2 \B^\top\B+\nu\I$.
    \end{enumerate}
\end{lemma}
We are going to apply Lemma \ref{l:dual-condition-tester} with $\A=\A_{t-1}$ and $\B=\A_t$. Crucially, the algorithm only requires matrix vector products with $\M_{t-1}$, $\M_t$, and $\C_{t}$, which are all explicitly constructed. As in the proof of Theorem \ref{t:dual-main-optimized}, we consider a sequence of candidates $\hat\nu_{t,i} = \nu_T/2^i$, and run the tester algorithm until it returns $X=0$. Suppose that we have already shown that $\hat\nu_{t,i}$ satisfies \eqref{eq:dual-condition} with $c=4$. 
Now, to apply Lemma \ref{l:dual-condition-tester} on $\hat\nu_{t,i+1}$ we need to be able to run solvers for both $\A_t\A_t^\top+\hat\nu_{t,i+1}\I$ 
and $\A_{t-1}\A_{t-1}^\top+\hat\nu_{t,i+1}\I$. 
For the latter, we consider the following preconditioning chain $((\A_{t-1},\hat\nu_{t,i+1}),(\A_t,\hat\nu_{t,i}),((\A_j,\max\{\nu_j,\hat\nu_{t,i}\}))_{j\in\{t+1,...,T\}})$. 
The running time bound obtained from Lemma \ref{l:dual-general} will be no larger than \eqref{eq:dual-bound}, as the only additional term that appears in this bound is $\tilde O(k_{t-1}^2\sqrt{\hat\nu_{t,i}/\hat\nu_{t,i+1}}\exp(\alpha t))$ which is bounded by \eqref{eq:dual-bound} since $\hat\nu_{t,i+1} = \hat\nu_{t,i}/2$.

    Having found all of the regularization parameters, the claim now follows by observing that $\M\preceq\M+\nu_0\I\preceq 4\M$, and using $O(\log1/\epsilon)$ applications of a $\Theta(1)$-solver $f_0$ for $\M+\nu_0$ we can construct an $\epsilon$-solver~for~$\M$.
\end{proof}
We now describe how Theorem \ref{t:psd-optimized} follows from Theorem \ref{t:dual-main-optimized}, which involves a grid search through all of the problem-dependent quantities, again following the strategy from the previous sections.
\begin{proof}[Proof of Theorem \ref{t:psd-optimized}]
    To find the correct values of all of the parameters, we start by computing $M=\tr(\M)$. Then, we choose $\tilde \kappa= 2,4,8, ...$ and for each of those values, perform a grid search over the following configurations: all pairs of $\tilde\lambda\leq\tilde\Lambda$ from the set $\{\tilde M 2^{-t} \mid t\in[\lceil\log_2(d\tilde\kappa)\rceil]\}$. We return the solution with the smallest least squares loss. Requires evaluating the algorithm $O(\log^3(d\kappa(\M)))$ times.
    \end{proof}

\subsection{Testing Spectral Approximation}
\label{s:dual-condition-tester}
In this section, we give the proof of \Cref{l:dual-condition-tester}. This argument requires some extra care compared to the corresponding analysis for the primal-dual regression solver (\Cref{l:primaldual-condition-tester}), because the matrices $\A$ and $\B$ are not given to us explicitly.
\begin{proof}[Proof of \Cref{l:dual-condition-tester}]
We are going to use the power method to check the following spectral norm condition with appropriately chosen $c>1$:
\begin{align}
    \|\A(\B^\top\B+\nu\I)^{-1}\A^\top+\I\|\leq c
    \quad\text{and}\quad
    \|\B(\A^\top\A+\nu\I)^{-1}\B^\top+\I\|\leq c.
    \label{eq:dual-condition2}
\end{align}
Note that if the above norm bounds hold, then this implies the spectral approximation $\M\approx_c\N$, where $\M := \A^\top\A+\nu\I$ and $\N=\B^\top\B+\nu\I$, because:
\begin{align*}
    \|\N^{-1/2}\M\N^{-1/2}\|\leq 
    \|\M^{-1/2}\A^\top\A\N^{-1/2}\| + \|\nu\N^{-1}\|\leq \|\A\N^{-1}\A^\top\|+1\leq c,
\end{align*}
and similarly, we get $\|\M^{-1/2}\N\M^{-1/2}\|\leq c$, which together imply $\M\approx_c\N$. In the other direction, we require a slightly stronger spectral approximation: if $\M\approx_{c-1}\N$, then \eqref{eq:dual-condition2} follows since
\begin{align*}
    \|\A\N^{-1}\A^\top + \I\| = \|\N^{-1/2}\A^\top\A\N^{-1/2}\| + 1\leq \|\N^{-1/2}\M\N^{-1/2}\| + 1\leq c,
\end{align*}
with an analogous calculation for the other bound.

The crucial advantage of condition \eqref{eq:dual-condition2} is that we can apply the matrices under the norms using only matrix vector products with $\A\A^\top$, $\B\B^\top$, and $\A\B^\top$, thanks to the Woodbury solver from Lemma \ref{l:woodbury}. Specifically, observe that:
\begin{align*}
    \A(\B^\top\B+\nu\I)^{-1}\A^\top = \frac1\nu\Big(\A\A^\top - \A\B^\top(\B\B^\top+\nu\I)^{-1}\B\A^\top\Big).
\end{align*}
From Lemma \ref{l:woodbury}, we know that $\hat\y := \frac1\nu(\A^\top\x - \B^\top g(\B\A^\top\x))$  satisfies $\|\hat\y-\y\|_{\N}^2\leq \epsilon\|\y\|_{\N}^2$ as an approximation of $\y = \N^{-1}\A^\top\x$ for $\epsilon=d^{-O(1)}$. It follows that $\hat\z = \A\hat\y + \x$ as an approximation of $(\A\N^{-1}\A^\top+\I)\x$ satisfies:
\begin{align*}
    \|\hat\z - \z\|^2 &= \|\A(\hat\y-\y)\|^2 \leq \|\hat\y-\y\|_{\M}^2\leq 8\|\hat\y-\y\|_{\N}^2
    \\
    &\leq 8\epsilon\|\y\|_{\N}^2 = 8\epsilon\|\N^{-1/2}\A^\top\x\|^2\leq 64\epsilon\|\x\|^2\leq 64\epsilon\|(\A\N^{-1}\A^\top+\I)\x\|^2,
\end{align*}
where  we used that since $\M\approx_8\N$ we have $\|\N^{-1/2}\A^\top\|^2\leq \|\N^{-1/2}\M\N^{-1/2}\|\leq 8$. Thus, we have shown how to perform a $64\epsilon$-matvec for $\A\N^{-1}\A^\top+\I$ in $O(\runtime_{\A\A^\top} + \runtime_{\A\B^\top}+\runtime_g)$ time. Similarly, we can construct a $64\epsilon$-matvec for $\B\M^{-1}\B^\top+\I$ in $O(\runtime_{\B\B^\top} + \runtime_{\A\B^\top}+\runtime_f)$ time. Therefore, using Lemma \ref{l:power-matvec} in time $O((\runtime_f + \runtime_{g}+\runtime_{\A\A^\top}+\runtime_{\B\B^\top}+\runtime_{\A\B^\top})\log(d)\log(1/\delta))$ with probability $1-\delta$ we can construct estimates:
\begin{align*}
    \hat a \approx_{8/7} \|\A\N^{-1}\A^\top + \I\|\quad\text{and}\quad
    \hat b \approx_{8/7} \|\B\M^{-1}\B^\top+\I\|.
\end{align*}
The final output of the algorithm now becomes:
\begin{align*}
    X = \mathbb{1}_{\hat a\leq 3.5}\cdot\mathbb{1}_{\hat b\leq 3.5}.
\end{align*}
Note that if $X=1$, then $\max\{\|\A\N^{-1}\A^\top+\I\|,\|\B\M^{-1}\B^\top+\I\|\}\leq (8/7)\cdot 3.5 =4$, which implies that $\M\approx_4\N$. On the other hand, if $X=0$, then one of the two norms is bigger than $(7/8)\cdot 3.5>3$, which means that $\M\not\approx_2\N$.
\end{proof}

\section{Spectrum Approximation}
\label{s:spectrum}

In this section, we discuss how our linear system solvers can be used as subroutines in other matrix approximation tasks, leading to the proof of \Cref{t:nuclear}. As a concrete example, we focus on the task of approximating the Schatten $p$-norm of a matrix $\A$, building on the algorithmic framework of \cite{musco2018spectrum}. Other approximation tasks described in \cite{musco2018spectrum}, such as estimating Orlicz norms and Ky Fan norms, can be similarly accelerated with our algorithms. Other natural applications of our framework include principal component projections \cite{frostig2016principal} and eigenvector computation~\cite{garber2016faster}.

The task of Schatten $p$-norm approximation is defined as follows: Given $\A\in\R^{n\times d}$, constant $p>0$, and $\epsilon\in(0,1)$, compute a $(1\pm\epsilon)$ factor approximation of:
\begin{align*}
    \|\A\|_p^p = \sum_{i \in [d]} \sigma_i^p(\A).
\end{align*}
For the case of $p\geq 2$, \cite{musco2018spectrum} gave running times that are already nearly-linear for dense matrices with $\epsilon= \Theta(1)$, and thus the main remaining question is improving Schatten $p$-norm approximation for $p\in(0,2)$, with the case of $p=1$ (the nuclear/trace norm) being of particular significance. The algorithmic framework developed by \cite{musco2018spectrum} provides a reduction from this task to the task of solving linear systems with matrices $\A^\top\A+\lambda\I$ for a carefully chosen $\lambda$. The following lemma describes this reduction.
\begin{lemma}[Based on the proof of Corollary 12 in \cite{musco2018spectrum}]\label{l:spectrum-reduction}
Given a matrix $\A\in\R^{n\times d}$, $p\in(0,2)$ and $\epsilon\in(0,1)$, and a $\poly(\epsilon/d)$-solver $f_\lambda$ for $\A^\top\A+\lambda\I$ with $\lambda \geq (\frac\epsilon d\|\A\|_p^p)^{2/p}$, there is an algorithm that, using $\tilde O(\frac1{\epsilon^5p})$ calls to $f_\lambda$ and $\tilde O(\runtime_{\A})$ additional time, whp.\ returns $X\in(1\pm O(\epsilon))\|\A\|_p^p$.
\end{lemma}
The above lemma immediately implies a simple reduction from Schatten $p$-norm approximation to solving linear systems over $(0,p)$-well-conditioned matrices (see Definition \ref{d:kappabar}). 
\begin{corollary}\label{c:schatten-reduction}
    Let $p\in(0,2)$ and suppose that $f$ is an algorithm that given a $(0,p)$-well-conditioned $\A\in\R^{n\times d}$ whp.\ returns a $\poly(1/d)$-solver $g_{\A}=f(\A)$ for $\M=\A^\top\A$. Then, for any $\epsilon=\Theta(1)$, there is an algorithm $h$ that given a matrix $\A$ whp.\ returns $X\in(1\pm\epsilon)\|\A\|_p^p$ in time 
    \[\runtime_h = \tilde O(\runtime_\A+\runtime_f+\runtime_{g_{\A}}).\]
\end{corollary}
The corollary also applies when $f$ receives a $(0,p)$-well-conditioned $\M\in\PD{d}$, and we seek $(1\pm\epsilon)\|\M\|_p^p$. 
So, our solvers from Theorems \ref{t:ls-optimized} and \ref{t:psd-optimized} directly imply nearly-linear time constant-factor approximation algorithms for Schatten $p$-norms, with $p>1/2$ for general matrices, and $p>1/4$ for PD matrices.

We also provide a more fine-grained analysis for Schatten $p$-norm approximation of general matrices via the reduction of Lemma \ref{l:spectrum-reduction}, showing that our runtime improvements extend to all $\epsilon\in(0,1)$ and $p\in(0,2)$.
\begin{theorem}
    There is an algorithm $f$ that given $\A\in\R^{n\times d}$, constant $0<p\leq 2$, and $\epsilon\in(0,1)$, whp.\ returns $X\in(1\pm\epsilon)\|\A\|_p^p$ in time:
        \begin{align*}
        \runtime_f 
        = \begin{cases}
        \tilde O\Big(\big(\nnz(\A) + d^{2}\big)\poly(1/\epsilon)\Big)&\text{ when }p\in(1/2,2],\\[2mm]
        \tilde O\Big(\big(\nnz(\A) + d^{2+o(1)}\big)\poly(1/\epsilon)\Big)&\text{ when }p=1/2,\\[2mm]
        \tilde O\Big(\big(\nnz(\A)+ 
        d^{2+(\omega-2)(1+\frac{p\omega}{0.5-p})^{-1}
        +o(1)}\big)\poly(1/\epsilon)\Big)&\text{ when }p\in(0,1/2).
        \end{cases}
        \end{align*}
\end{theorem}
\begin{remark}
    This is the first nearly-linear time dense matrix Schatten $p$-norm approximation algorithm for any $p\in(1/2,2)$, including the nuclear norm ($p=1$), as all previously known running times are $\Omega(d^{2+\theta_p})$ for some $\theta_p>0$. For $p\in(0,1/2]$, the previously best known runtime was $\tilde O((\nnz(\A)+d^{2+(\omega-2)(1+p\omega)^{-1}})\poly(1/\epsilon))$, due to \cite{derezinski2024faster}. Thus, our algorithms also achieve best known time complexity in this regime. E.g., with $p=1/4$, we get $\tilde O((\nnz(\A)+d^{2.1103})\poly(1/\epsilon))$, compared to the previous best $\tilde O((\nnz(\A) + d^{2.2336})\poly(1/\epsilon))$.
\end{remark}
\begin{proof}
    First, let us observe that for any fixed $p\in(0,2)$, we can use a subspace embedding to reduce the dimension $n$ down to $\tilde O(d/\epsilon^2)$. If we let $\S\in\R^{s\times n}$ be a sparse embedding matrix with sketch size $s = \tilde O(d/\epsilon^2)$, then whp.\ $\tilde\A=\S\A$ satisfies $\sigma_i(\tilde\A)^2\approx_{1+\epsilon}\sigma_i(\A)^2$ for each $i\in[d]$. Taking both sides to power $p/2\leq 1$, this implies that $\sigma_i(\tilde\A)^p\approx_{1+\epsilon}\sigma_i(\A)^p$, which in turn implies that $\|\tilde\A\|_p^p\approx_{1+\epsilon}\|\A\|_p^p$. Thus, with additional $\tilde O(\nnz(\A)/\epsilon)$ runtime, without loss of generality we can assume that $n = \tilde O(d/\epsilon^2)$.

    Let us start with the case of $p> 1/2$. 
    Set $\lambda \approx_2 (\frac\epsilon d\|\A\|_p^p)^{2/p}$, and let $\A_\lambda$ be the matrix $\A$ with $\sqrt\lambda\I_d$ appended at the bottom. Then, it follows that:
    \begin{align*}
        \bar\kappa_{0,p}(\A_\lambda) \leq \frac{\|\A\|_p}{\sqrt\lambda d^{1/p}}\leq 2\epsilon^{-1/p}.
    \end{align*}
    Thus, we can solve the linear system $\A^\top\A+\lambda\I=\A_\lambda^\top\A_\lambda$ to sufficient accuracy in $\tilde O(d^{2}\poly(1/\epsilon))$ time, which yields the desired runtime via Lemma \ref{l:spectrum-reduction}.

    Next, we consider the case of $p\in(0,1/2]$. Setting $\lambda \approx_2 (\frac\epsilon d\|\A\|_p^p)^{2/p}$ and using $\sigma_i:=\sigma_i(\A)$, we have:
    \begin{align*}
        \bar\kappa_{k,1/2}(\A_\lambda)^{1/2} \leq \frac{\sum_{i>k}\sigma_i^{1/2}}{\lambda^{1/4} d} \leq 2\frac{d^{\frac1{2p}-1}}{\epsilon^{1/2p}}\frac{\sigma_k^{1/2-p}\sum_{i>k}\sigma_i^{p}}{\|\A\|_p^{1/2}}\leq 
        2\frac{d^{\frac1{2p}-1}}{\epsilon^{1/2p}}\frac{\sigma_k^{1/2-p}}{\|\A\|_p^{1/2-p}}\leq \frac 2{\epsilon^{1/2p}}\Big(\frac dk\Big)^{\frac1{2p}-1},
    \end{align*}
    where we used that $\sigma_k^p\leq \frac1k\|\A\|_p^p$.
    Thus, we can solve the linear system $\A^\top\A+\lambda\I$  in time:
    \begin{align*}
        \tilde O\bigg( d^{2+o(1)}\Big(\frac dk\Big)^{\frac1{2p}-1}\poly(1/\epsilon) + k^\omega\bigg),
    \end{align*}
    where note that here $\epsilon$ still refers to the accuracy of the norm approximation task, whereas the accuracy of the linear solver has only a $\poly(\log d)$ factor dependence on runtime. Balancing out the two terms (without optimizing the $\epsilon$-dependence), we choose $k= d^{\frac{2p+1}{2p\omega+1-2p}}$, which yields the desired runtime via Lemma~\ref{l:spectrum-reduction}.
\end{proof}

\section*{Acknowledgements} 
Thank you to anonymous reviewers for their feedback. Micha{\l} Derezi\'nski was supported in part by NSF CAREER Grant CCF-233865 and a Google ML and Systems Junior Faculty Award.
Aaron Sidford was supported in part by a Microsoft Research Faculty Fellowship, NSF CAREER Grant CCF1844855, NSF Grant CCF-1955039, and a PayPal research award.

\bibliographystyle{plain}
\bibliography{pap}

\begin{thebibliography}{10}

\bibitem{agarwal2020leverage}
Naman Agarwal, Sham Kakade, Rahul Kidambi, Yin-Tat Lee, Praneeth Netrapalli,
  and Aaron Sidford.
\newblock Leverage score sampling for faster accelerated regression and erm.
\newblock In {\em Algorithmic Learning Theory}, pages 22--47. PMLR, 2020.

\bibitem{allen2018katyusha}
Zeyuan Allen-Zhu.
\newblock Katyusha: The first direct acceleration of stochastic gradient
  methods.
\newblock {\em Journal of Machine Learning Research}, 18(221):1--51, 2018.

\bibitem{alman2025more}
Josh Alman, Ran Duan, Virginia~Vassilevska Williams, Yinzhan Xu, Zixuan Xu, and
  Renfei Zhou.
\newblock More asymmetry yields faster matrix multiplication.
\newblock In {\em Proceedings of the 2025 Annual ACM-SIAM Symposium on Discrete
  Algorithms (SODA)}, pages 2005--2039. SIAM, 2025.

\bibitem{avron2017faster}
Haim Avron, Kenneth~L Clarkson, and David~P Woodruff.
\newblock Faster kernel ridge regression using sketching and preconditioning.
\newblock {\em SIAM Journal on Matrix Analysis and Applications},
  38(4):1116--1138, 2017.

\bibitem{axelsson1986rate}
Owe Axelsson and Gunhild Lindskog.
\newblock On the rate of convergence of the preconditioned conjugate gradient
  method.
\newblock {\em Numerische Mathematik}, 48:499--523, 1986.

\bibitem{chenakkod2024optimal}
Shabarish Chenakkod, Micha{\l} Derezi{\'n}ski, and Xiaoyu Dong.
\newblock Optimal oblivious subspace embeddings with near-optimal sparsity.
\newblock In {\em 52nd International Colloquium on Automata, Languages, and
  Programming (ICALP 2025)}, 2024.

\bibitem{chenakkod2023optimal}
Shabarish Chenakkod, Micha{\l} Derezi{\'n}ski, Xiaoyu Dong, and Mark Rudelson.
\newblock Optimal embedding dimension for sparse subspace embeddings.
\newblock {\em Symposium on Theory of Computing (STOC)}, 2024.

\bibitem{cw-sparse}
Kenneth~L. Clarkson and David~P. Woodruff.
\newblock Low-rank approximation and regression in input sparsity time.
\newblock {\em J. ACM}, 63(6):54:1--54:45, January 2017.

\bibitem{cohen2016nearly}
Michael~B Cohen.
\newblock Nearly tight oblivious subspace embeddings by trace inequalities.
\newblock In {\em Proceedings of the twenty-seventh annual ACM-SIAM symposium
  on Discrete algorithms}, pages 278--287. SIAM, 2016.

\bibitem{cohen2018solving}
Michael~B Cohen, Jonathan Kelner, Rasmus Kyng, John Peebles, Richard Peng,
  Anup~B Rao, and Aaron Sidford.
\newblock Solving directed laplacian systems in nearly-linear time through
  sparse lu factorizations.
\newblock In {\em 2018 IEEE 59th Annual Symposium on Foundations of Computer
  Science (FOCS)}, pages 898--909. IEEE, 2018.

\bibitem{CohenKPPRSV17}
Michael~B. Cohen, Jonathan~A. Kelner, John Peebles, Richard Peng, Anup~B. Rao,
  Aaron Sidford, and Adrian Vladu.
\newblock Almost-linear-time algorithms for markov chains and new spectral
  primitives for directed graphs.
\newblock In {\em 2017 Proceedings of the 49th Annual {ACM} {SIGACT} Symposium
  on Theory of Computing, {STOC}}, pages 410--419. {ACM}, 2017.

\bibitem{CohenKMPPRX14}
Michael~B. Cohen, Rasmus Kyng, Gary~L. Miller, Jakub~W. Pachocki, Richard Peng,
  Anup~B. Rao, and Shen~Chen Xu.
\newblock Solving {SDD} linear systems in nearly $m \log^{1/2} n$ time.
\newblock In {\em Symposium on Theory of Computing, {STOC} 2014}, pages
  343--352. {ACM}, 2014.

\bibitem{CohenLMMPS15}
Michael~B. Cohen, Yin~Tat Lee, Cameron Musco, Christopher Musco, Richard Peng,
  and Aaron Sidford.
\newblock Uniform sampling for matrix approximation.
\newblock In {\em Proceedings of the 2015 Conference on Innovations in
  Theoretical Computer Science, {ITCS} 2015}, pages 181--190. {ACM}, 2015.

\bibitem{cohen2016optimal}
Michael~B Cohen, Jelani Nelson, and David~P Woodruff.
\newblock Optimal approximate matrix product in terms of stable rank.
\newblock In {\em 43rd International Colloquium on Automata, Languages, and
  Programming (ICALP 2016)}. Schloss-Dagstuhl-Leibniz Zentrum f{\"u}r
  Informatik, 2016.

\bibitem{coppersmith1987matrix}
Don Coppersmith and Shmuel Winograd.
\newblock Matrix multiplication via arithmetic progressions.
\newblock In {\em Proceedings of the nineteenth annual ACM symposium on Theory
  of computing}, pages 1--6, 1987.

\bibitem{demmel2007fast}
James Demmel, Ioana Dumitriu, and Olga Holtz.
\newblock Fast linear algebra is stable.
\newblock {\em Numerische Mathematik}, 108(1):59--91, 2007.

\bibitem{derezinski2024fine}
Micha{\l} Derezi{\'n}ski, Daniel LeJeune, Deanna Needell, and Elizaveta
  Rebrova.
\newblock Fine-grained analysis and faster algorithms for iteratively solving
  linear systems.
\newblock {\em arXiv preprint arXiv:2405.05818}, 2024.

\bibitem{derezinski2024faster}
Micha{\l} Derezi{\'n}ski, Christopher Musco, and Jiaming Yang.
\newblock Faster linear systems and matrix norm approximation via multi-level
  sketched preconditioning.
\newblock {\em ACM-SIAM Symposium on Discrete Algorithms (SODA)}, 2025.

\bibitem{derezinski2025randomized}
Micha{\l} Derezi{\'n}ski, Deanna Needell, Elizaveta Rebrova, and Jiaming Yang.
\newblock Randomized kaczmarz methods with beyond-krylov convergence.
\newblock {\em arXiv preprint arXiv:2501.11673}, 2025.

\bibitem{derezinski2023solving}
Micha{\l} Derezi\'nski and Jiaming Yang.
\newblock Solving linear systems faster than via preconditioning.
\newblock In {\em Proceedings of the Symposium on Theory of Computing (STOC)},
  2024.

\bibitem{frangella2021randomized}
Zachary Frangella, Joel~A Tropp, and Madeleine Udell.
\newblock Randomized {N}ystrom preconditioning.
\newblock {\em arXiv preprint arXiv:2110.02820}, 2021.

\bibitem{frostig2015regularizing}
Roy Frostig, Rong Ge, Sham Kakade, and Aaron Sidford.
\newblock Un-regularizing: approximate proximal point and faster stochastic
  algorithms for empirical risk minimization.
\newblock In {\em International Conference on Machine Learning}, pages
  2540--2548. PMLR, 2015.

\bibitem{frostig2016principal}
Roy Frostig, Cameron Musco, Christopher Musco, and Aaron Sidford.
\newblock Principal component projection without principal component analysis.
\newblock In {\em International Conference on Machine Learning}, pages
  2349--2357. PMLR, 2016.

\bibitem{garber2016faster}
Dan Garber, Elad Hazan, Chi Jin, Cameron Musco, Praneeth Netrapalli, and Aaron
  Sidford.
\newblock Faster eigenvector computation via shift-and-invert preconditioning.
\newblock In {\em International Conference on Machine Learning}, pages
  2626--2634. PMLR, 2016.

\bibitem{golub2013matrix}
Gene~H Golub and Charles~F Van~Loan.
\newblock {\em Matrix computations}.
\newblock JHU press, 2013.

\bibitem{golub1961chebyshev}
Gene~H Golub and Richard~S Varga.
\newblock Chebyshev semi-iterative methods, successive overrelaxation iterative
  methods, and second order richardson iterative methods.
\newblock {\em Numerische Mathematik}, 3(1):157--168, 1961.

\bibitem{gonen2016solving}
Alon Gonen, Francesco Orabona, and Shai Shalev-Shwartz.
\newblock Solving ridge regression using sketched preconditioned svrg.
\newblock In {\em International conference on machine learning}, pages
  1397--1405. PMLR, 2016.

\bibitem{greenbaum1989behavior}
Anne Greenbaum.
\newblock Behavior of slightly perturbed lanczos and conjugate-gradient
  recurrences.
\newblock {\em Linear Algebra and its Applications}, 113:7--63, 1989.

\bibitem{hestenes1952methods}
Magnus~Rudolph Hestenes and Eduard Stiefel.
\newblock {\em Methods of conjugate gradients for solving linear systems},
  volume~49.
\newblock NBS Washington, DC, 1952.

\bibitem{higham2011gaussian}
Nicholas~J Higham.
\newblock Gaussian elimination.
\newblock {\em Wiley Interdisciplinary Reviews: Computational Statistics},
  3(3):230--238, 2011.

\bibitem{Jambulapati0MSS23}
Arun Jambulapati, Jerry Li, Christopher Musco, Kirankumar Shiragur, Aaron
  Sidford, and Kevin Tian.
\newblock Structured semidefinite programming for recovering structured
  preconditioners.
\newblock In {\em Advances in Neural Information Processing Systems 36: Annual
  Conference on Neural Information Processing Systems 2023, NeurIPS 2023},
  2023.

\bibitem{js_talg}
Arun Jambulapati and Aaron Sidford.
\newblock Ultrasparse ultrasparsifiers and faster laplacian system solvers.
\newblock {\em ACM Trans. Algorithms}, February 2024.

\bibitem{johnson2013accelerating}
Rie Johnson and Tong Zhang.
\newblock Accelerating stochastic gradient descent using predictive variance
  reduction.
\newblock {\em Advances in neural information processing systems}, 26, 2013.

\bibitem{KapralovLMMS17}
Michael Kapralov, Yin~Tat Lee, Cameron Musco, Christopher Musco, and Aaron
  Sidford.
\newblock Single pass spectral sparsification in dynamic streams.
\newblock {\em {SIAM} Journal of Computing}, 46(1):456--477, 2017.

\bibitem{KoutisMP10}
Ioannis Koutis, Gary~L. Miller, and Richard Peng.
\newblock Approaching optimality for solving {SDD} linear systems.
\newblock In {\em 2010 {IEEE} 51th Annual Symposium on Foundations of Computer
  Science, {FOCS}}, pages 235--244. {IEEE} Computer Society, 2010.

\bibitem{KoutisMP11}
Ioannis Koutis, Gary~L. Miller, and Richard Peng.
\newblock A nearly-m log n time solver for {SDD} linear systems.
\newblock In {\em 2011 {IEEE} 52nd Annual Symposium on Foundations of Computer
  Science, {FOCS}}, pages 590--598. {IEEE} Computer Society, 2011.

\bibitem{kuczynski1992estimating}
Jacek Kuczy{\'n}ski and Henryk Wo{\'z}niakowski.
\newblock Estimating the largest eigenvalue by the power and lanczos algorithms
  with a random start.
\newblock {\em SIAM journal on matrix analysis and applications},
  13(4):1094--1122, 1992.

\bibitem{kyng2016sparsified}
Rasmus Kyng, Yin~Tat Lee, Richard Peng, Sushant Sachdeva, and Daniel~A
  Spielman.
\newblock Sparsified cholesky and multigrid solvers for connection laplacians.
\newblock In {\em Proceedings of the forty-eighth annual ACM symposium on
  Theory of Computing}, pages 842--850, 2016.

\bibitem{kyng2016approximate}
Rasmus Kyng and Sushant Sachdeva.
\newblock Approximate gaussian elimination for laplacians-fast, sparse, and
  simple.
\newblock In {\em 2016 IEEE 57th Annual Symposium on Foundations of Computer
  Science (FOCS)}, pages 573--582. IEEE, 2016.

\bibitem{le2012faster}
Francois Le~Gall.
\newblock Faster algorithms for rectangular matrix multiplication.
\newblock In {\em 2012 IEEE 53rd annual symposium on foundations of computer
  science}, pages 514--523. IEEE, 2012.

\bibitem{lee2013efficient}
Yin~Tat Lee and Aaron Sidford.
\newblock Efficient accelerated coordinate descent methods and faster
  algorithms for solving linear systems.
\newblock In {\em 2013 ieee 54th annual symposium on foundations of computer
  science}, pages 147--156. IEEE, 2013.

\bibitem{leventhal2010randomized}
Dennis Leventhal and Adrian~S Lewis.
\newblock Randomized methods for linear constraints: convergence rates and
  conditioning.
\newblock {\em Mathematics of Operations Research}, 35(3):641--654, 2010.

\bibitem{iterative-row-sampling}
Mu~Li, Gary~L. Miller, and Richard Peng.
\newblock Iterative row sampling.
\newblock {\em 2013 IEEE 54th Annual Symposium on Foundations of Computer
  Science}, 00:127--136, 2014.

\bibitem{lin2015universal}
Hongzhou Lin, Julien Mairal, and Zaid Harchaoui.
\newblock A universal catalyst for first-order optimization.
\newblock {\em Advances in neural information processing systems}, 28, 2015.

\bibitem{musco2018stability}
Cameron Musco, Christopher Musco, and Aaron Sidford.
\newblock Stability of the lanczos method for matrix function approximation.
\newblock In {\em Proceedings of the Twenty-Ninth Annual ACM-SIAM Symposium on
  Discrete Algorithms}, pages 1605--1624. SIAM, 2018.

\bibitem{musco2018spectrum}
Cameron Musco, Praneeth Netrapalli, Aaron Sidford, Shashanka Ubaru, and David~P
  Woodruff.
\newblock Spectrum approximation beyond fast matrix multiplication: Algorithms
  and hardness.
\newblock In {\em 9th Innovations in Theoretical Computer Science Conference
  (ITCS 2018)}. Schloss-Dagstuhl-Leibniz Zentrum f{\"u}r Informatik, 2018.

\bibitem{nn-sparse}
Jelani Nelson and Huy~L. Nguy\^{e}n.
\newblock {OSNAP}: Faster numerical linear algebra algorithms via sparser
  subspace embeddings.
\newblock In {\em Proceedings of the Symposium on Foundations of Computer
  Science}, FOCS '13, pages 117--126, 2013.

\bibitem{Nesterov12}
Yurii~E. Nesterov.
\newblock Efficiency of coordinate descent methods on huge-scale optimization
  problems.
\newblock {\em {SIAM} Journal of Optimization}, 22(2):341--362, 2012.

\bibitem{NesterovS17}
Yurii~E. Nesterov and Sebastian~U. Stich.
\newblock Efficiency of the accelerated coordinate descent method on structured
  optimization problems.
\newblock {\em {SIAM} Journal of Optimization}, 27(1):110--123, 2017.

\bibitem{Peng18}
Richard Peng.
\newblock {\em Algorithm Design Using Spectral Graph Theory}.
\newblock PhD thesis, Carnegie Mellon University, {USA}, 2013.

\bibitem{PengS22}
Richard Peng and Zhuoqing Song.
\newblock Sparsified block elimination for directed laplacians.
\newblock In {\em {STOC} '22: 54th Annual {ACM} {SIGACT} Symposium on Theory of
  Computing}, pages 557--567. {ACM}, 2022.

\bibitem{rokhlin2008fast}
Vladimir Rokhlin and Mark Tygert.
\newblock A fast randomized algorithm for overdetermined linear least-squares
  regression.
\newblock {\em Proceedings of the National Academy of Sciences},
  105(36):13212--13217, 2008.

\bibitem{sarlos-sketching}
Tamas Sarlos.
\newblock Improved approximation algorithms for large matrices via random
  projections.
\newblock In {\em Proceedings of the Symposium on Foundations of Computer
  Science}, FOCS '06, pages 143--152, 2006.

\bibitem{SpielmanT14}
Daniel~A. Spielman and Shang{-}Hua Teng.
\newblock Nearly linear time algorithms for preconditioning and solving
  symmetric, diagonally dominant linear systems.
\newblock {\em {SIAM} Journal on Matrix Analysis and Applications},
  35(3):835--885, 2014.

\bibitem{sw09}
Daniel~A Spielman and Jaeoh Woo.
\newblock A note on preconditioning by low-stretch spanning trees.
\newblock {\em arXiv preprint arXiv:0903.2816}, 2009.

\bibitem{Strassen1969}
Volker Strassen.
\newblock Gaussian elimination is not optimal.
\newblock {\em Numerische Mathematik}, 13(4):354--356, Aug 1969.

\bibitem{strohmer2009randomized}
Thomas Strohmer and Roman Vershynin.
\newblock A randomized {K}aczmarz algorithm with exponential convergence.
\newblock {\em Journal of Fourier Analysis and Applications}, 15(2):262--278,
  2009.

\bibitem{Williams01Nystrom}
Christopher K.~I. Williams and Matthias Seeger.
\newblock Using the {N}ystr\"{o}m method to speed up kernel machines.
\newblock In {\em Advances in Neural Information Processing Systems 13}, pages
  682--688. 2001.

\bibitem{williams2012multiplying}
Virginia~Vassilevska Williams.
\newblock Multiplying matrices faster than coppersmith-winograd.
\newblock In {\em Proceedings of the forty-fourth annual ACM symposium on
  Theory of computing}, pages 887--898, 2012.

\bibitem{williams2024multiplication}
Virginia~Vassilevska Williams, Yinzhan Xu, Zixuan Xu, and Renfei Zhou.
\newblock New bounds for matrix multiplication: from alpha to omega.
\newblock In {\em Proceedings of the 2024 Annual ACM-SIAM Symposium on Discrete
  Algorithms (SODA)}, pages 3792--3835, 2024.

\bibitem{woodruff2014sketching}
David~P Woodruff.
\newblock Sketching as a tool for numerical linear algebra.
\newblock {\em Foundations and Trends{\textregistered} in Theoretical Computer
  Science}, 10(1--2):1--157, 2014.

\bibitem{ZhuQRY16}
Zeyuan~Allen Zhu, Zheng Qu, Peter Richt{\'{a}}rik, and Yang Yuan.
\newblock Even faster accelerated coordinate descent using non-uniform
  sampling.
\newblock In {\em Proceedings of the 33nd International Conference on Machine
  Learning, {ICML} 2016}, volume~48 of {\em {JMLR} Workshop and Conference
  Proceedings}, pages 1110--1119. JMLR.org, 2016.

\end{thebibliography}

\end{document}